\documentclass[%
reprint,
superscriptaddress,
frontmatterverbose, 
preprintnumbers,
 amsmath,amssymb,
 aps,
pra,
notitlepage,
longbibliography
]{revtex4-1}

\usepackage{tabularx}
\usepackage{ragged2e}

\usepackage{graphicx}
\usepackage{dcolumn}
\usepackage{bm}
\usepackage{xcolor}
\definecolor{darkblue}{rgb}{0,0,0.5}
\usepackage{hyperref}
\hypersetup{
    bookmarksnumbered=true, 
    unicode=false, 
    pdfstartview={FitH}, 
    pdftitle={}, 
    pdfauthor={}, 
    pdfsubject={}, 
    pdfcreator={}, 
    pdfproducer={}, 
    pdfkeywords={}, 
    pdfnewwindow=true, 
    colorlinks=true, 
    linkcolor=darkblue, 
    citecolor=darkblue, 
    filecolor=blue, 
    urlcolor=darkblue 
}
\usepackage{subfigure}

\usepackage{amsthm}

\usepackage{enumerate}

\usepackage{comment}

\usepackage{mathrsfs}

\theoremstyle{plain}
\newtheorem{thm}{Theorem}
\newtheorem{lem}[thm]{Lemma}
\newtheorem{prop}[thm]{Proposition}
\newtheorem{cor}[thm]{Corollary}

\theoremstyle{definition}
\newtheorem{defn}{Definition}

\newtheorem*{condffr}{Condition\,(FFR)}
\newtheorem*{condch}{Condition\,(CH)}
\newtheorem*{condct}{Condition\,(CT)}

\newcommand{\eq}[1]{(\hyperref[eq:#1]{\ref*{eq:#1}})}

\renewcommand{\sec}[1]{\hyperref[sec:#1]{Section~\ref*{sec:#1}}}
\newcommand{\thrm}[1]{\hyperref[thrm:#1]{Theorem~\ref*{thrm:#1}}}
\newcommand{\lemm}[1]{\hyperref[lemm:#1]{Lemma~\ref*{lemm:#1}}}
\newcommand{\pro}[1]{\hyperref[pro:#1]{Proposition~\ref*{pro:#1}}}
\newcommand{\corr}[1]{\hyperref[corr:#1]{Corollary~\ref*{corr:#1}}}
\newcommand{\fig}[1]{\hyperref[fig:#1]{~\ref*{fig:#1}}}
\newcommand{\deff}[1]{\hyperref[deff:#1]{~\ref*{deff:#1}}}

\newcommand{\ket}[1]{|#1\rangle}
\newcommand{\bra}[1]{\langle #1|}
\newcommand{\bracket}[2]{\langle #1|#2\rangle}
\newcommand{\ketbra}[2]{|#1\rangle\langle #2|}
\newcommand{\dm}[1]{\ketbra{#1}{#1}}

\newcommand{\norm}[1]{\left\lVert#1\right\rVert}

\DeclareMathOperator{\Tr}{Tr}

\newcommand{\barx}{\mathscr{F}_{{\rm NG}}}

\newcommand{\eball}{\mathcal{B}^\epsilon}

\newcommand{\dmax}{D_{\mathrm{max}}}
\newcommand{\dmin}{D_{\mathrm{min}}}

\newcommand{\rdmax}{\mathfrak{D}_{\mathrm{max}}}
\newcommand{\rdmaxl}{\mathfrak{D}_{\mathrm{max},\lambda}}

\newcommand{\mH}{\mathcal{H}}
\newcommand{\mL}{\mathcal{L}}
\newcommand{\mT}{\mathcal{T}}
\newcommand{\mF}{\mathcal{F}}
\newcommand{\mD}{\mathcal{D}}

\newcommand{\ba}{\begin{eqnarray}}
\newcommand{\ea}{\end{eqnarray}}

\begin{document}


\title{One-Shot Operational Quantum Resource Theory}

\author{Zi-Wen Liu}\email{zliu1@perimeterinstitute.ca}
\affiliation{Perimeter Institute for Theoretical Physics, Waterloo, Ontario N2L 2Y5, Canada}
\affiliation{Center for Theoretical Physics and Department of Physics, Massachusetts Institute of Technology, Cambridge, Massachusetts 02139, USA}

\author{Kaifeng Bu}\email{kfbu@fas.harvard.edu}
 \affiliation{School of Mathematical Sciences, Zhejiang University, Hangzhou, Zhejiang 310027, China}
\affiliation{Department of Physics, Harvard University, Cambridge, Massachusetts 02138, USA}

\author{Ryuji Takagi}\email{rtakagi@mit.edu}
\affiliation{Center for Theoretical Physics and Department of Physics, Massachusetts Institute of Technology, Cambridge, Massachusetts 02139, USA}

\date{\today}

\begin{abstract}
A fundamental approach for the characterization and quantification of all kinds of resources is to study the conversion between different resource objects under certain constraints.
Here we analyze, from a resource-non-specific standpoint, the optimal efficiency of resource formation and distillation tasks with only a single copy of the given quantum state, thereby establishing a unified framework of one-shot quantum resource manipulation. We find general bounds on the optimal rates characterized by resource measures based on the smooth max/min-relative entropies and hypothesis testing relative entropy, as well as the free robustness measure, providing them with general operational meanings in terms of optimal state conversion. Our results encompass a wide class of resource theories via the theory-dependent coefficients we introduce, and the discussions are solidified by important examples, such as entanglement, coherence, superposition, magic states, asymmetry, and thermal non-equilibrium.

\end{abstract}

\maketitle


\section{Introduction}
The manipulation and characterization of resources are ubiquitous subjects of concern. 
In recent years, substantial research effort originated from the quantum information community has been devoted to a framework known as resource theory, which significantly advances the study of quantum physics and quantum technologies  (see Ref.~\cite{qrtreview} for a recent overview).
The framework centers around the task of quantifying the value of certain resource features (e.g.~quantum entanglement) in various scenarios, in order to rigorously understand the essence of these resources and how to best utilize them.
Resource theory is particularly interesting and powerful because of its versatility --- similar  methodologies are successfully applied to a plethora of important resource entities, such as entanglement \cite{virmani,entrevrt}, coherence \cite{Aberg2006,PhysRevLett.113.140401,Streltsov2017rev}, superposition \cite{Theurer2017}, magic states \cite{magic,howard_2017}, asymmetry \cite{frame,speakable}, purity \cite{Horodecki2003purity,Chiribella_2017}, thermal non-equilibrium \cite{HorodeckiOppenheim13,thermort,Brandao_secondlaws2015}, non-Gaussianity \cite{Genoni2008,Takagi2018,albarelli2018resource}. 
Therefore, a research line of fundamental importance is to investigate the unified, non-resource-specific aspects of resource theory and how they fit into different contexts \cite{Horodecki2013quantumness,bg,2015arXiv151108818D,COECKE201659,fritz_2017,rdm,PhysRevA.95.062314,2018arXiv180604937S,2018arXiv180802607G,2018arXiv180901672T,anshu_2018catalytic,regula,lami_2018,2018arXiv181202572L,uola2018quantifying,2019arXiv190108127T,Oszmaniec2019,LiuYuan:channel,LiuWinter2018}.

In this work, we establish such a general scheme for operationally quantifying the resource content of quantum states through their value in fundamental ``resource trading'' tasks.  More specifically, we are interested in the optimal rate of forming a quantum state using some standard resource states that serve as the ``currency'', and conversely that of using the given state to distill standard states, under typical free operations.
Many specific forms of such tasks are of independent interest; for example, the task of entanglement formation induces the entanglement cost, an important entanglement measure \cite{PhysRevA.54.3824,Hayden_2001}, and the tasks of entanglement distillation \cite{PhysRevA.53.2046,PhysRevLett.76.722,PhysRevA.54.3824} and magic state distillation \cite{bravyikitaev} play key roles in quantum information and computation.
Here we focus on the practical scenario where only one copy (or finite copies) of the state is available (i.e.~the one-shot setting), and some amount of error is allowed. 
Unlike the asymptotic theory (the limit of infinite i.i.d.~copies) \cite{bg},
only a few resource-specific results about entanglement \cite{brandaodatta,BuscemiDatta11}, coherence \cite{coh_dilution,PhysRevLett.121.010401,fullstory}, and (generalized) quantum thermodynamics \cite{Dahlsten_2011,HorodeckiOppenheim13,GOUR2015inf_oneshot,Halpern2016_beyond1,Halpern_2018_beyond2} (and magic states in a very recent work \cite{wang_wilde_su}) are known.
Here we consider two important classes of free operations easily characterized by the theory of resource destroying (RD) maps \cite{rdm}: the maximal free operations (e.g.~non-entangling operations for entanglement, maximal incoherent operations (MIO) for coherence, Gibbs-preserving maps for thermodynamics),  and the commuting operations (e.g.~dephasing-covariant incoherent operations (DIO) for coherence \cite{examination}, isotropic channels for discord when restricted to local operations and qudit systems \cite{dqd}), which induce general distance monotones without optimization.   We prove highly generic limits to the optimal rates of standard one-shot formation and distillation tasks under the above free operations, and show that they can be nearly achieved in many cases.  These general bounds take unified and simple forms in terms of resource monotones based on the smoothed max-relative entropy or the free (also called ``standard'') robustness for formation, and the smoothed min-relative entropy or the hypothesis testing relative entropy for distillation, divided by a certain modification coefficient that encodes the resource value of the standard states. 
To put it another way, the results endow these resource monotones with operational meanings in terms of ``normalized'' one-shot resource conversion tasks, providing a general operational interpretation to the min-relative entropy measure and supplementing those of the max-relative entropy and free robustness measures recently unveiled via erasure \cite{anshu_2018catalytic} and discrimination tasks \cite{2018arXiv180901672T,2019arXiv190108127T}.  
In particular, we find that taking maximum resource states as the currency not only makes the most sense out of formation/distillation tasks conceptually, but also leads to nice mathematical structures of the results.
For example, we show that several key resource measures (and therefore the corresponding modification coefficients) of \emph{golden states} (a notion of max-resource states we introduce) collapse to the same value in generic convex theories, which leads to nearly tight bounds.
Our results generalize the existing resource-specific ones, and we shall also elucidate the results by suitable new examples.

\section{Preliminaries}
Let $\mH_d$ be the Hilbert space of dimension $d<\infty$, and $\mathcal{D}(\mH_d)$ be the set of density operators acting on $\mH_d$.  Also let $\mF(\mH_d)\subseteq\mathcal{D}(\mH_d)$ be the set of free states in the resource theory under consideration
(the brackets are dropped onwards when the Hilbert space is clear from the context).  We assume that the set of free states is topologically closed, so that the maxima or minima over it are well-defined.

We first formally define several information-theoretic quantities and resource measures.
Let $\rho,\sigma$ be density operators \footnote{In fact, the definitions of max- and min-relative entropies only require $\sigma$ to be a positive semidefinite operator.}.  
The Uhlmann fidelity of $\rho$ and $\sigma$ is given by $f(\rho,\sigma):=\left(\Tr\sqrt{\sqrt{\sigma}\rho\sqrt{\sigma}}\right)^2 = \norm{\sqrt{\rho}\sqrt{\sigma}}_1^2$.  The free fidelity of $\rho$, which measures the maximum overlap with free states, is defined as $\mathfrak{f}(\rho) := \max_{\sigma\in\mF}f(\rho,\sigma)$. 
The \emph{max-relative entropy} and \emph{min-relative entropy} between $\rho$ and $\sigma$ are respectively given by \cite{dattarel}
\begin{equation*}
    \dmax(\rho\|\sigma) :=\log \min \{\lambda: \rho\leq \lambda \sigma\},
    \end{equation*}
    which is well-defined when $\mathrm{supp}(\rho)\subseteq\mathrm{supp}(\sigma)$, and
    \begin{equation*}
\dmin(\rho\|\sigma) := -\log\Tr\{\Pi_\rho\sigma\}
\end{equation*}
where $\Pi_\rho$ denotes the projector onto $\mathrm{supp}(\rho)$, which is well-defined when $\mathrm{supp}(\rho)\cap\mathrm{supp}(\sigma)\backslash\{0\}$ is non-empty. They roughly represent two ends of the spectrum of quantum R\'enyi relative entropy (see Appendix A \cite{sm} for more rigorous statements).  
To account for finite accuracy, the smoothed versions are needed. Let
$\mathcal{B}^\epsilon(\rho) := \{\rho': f(\rho',\rho)\geq 1-\epsilon\}$. The smoothed max- (min-) relative entropy between $\rho$ and $\sigma$ is then given by minimizing (maximizing) over this $\epsilon$-vicinity of $\rho$:
\begin{eqnarray*}
   D_{\max}^\epsilon(\rho\|\sigma) &:=& \mathop{\min}_{\rho'\in\mathcal{B}^\epsilon(\rho)} \,D_{\max}(\rho'\|\sigma),\\
   D_{\min}^\epsilon(\rho\|\sigma) &:=& \mathop{\max}_{\rho'\in\mathcal{B}^\epsilon(\rho)} \,D_{\min}(\rho'\|\sigma).
\end{eqnarray*}
For the min-relative entropy we also consider a slightly different type of smoothing known as the operator-smoothing:
       \begin{equation*}
   D_H^\epsilon(\rho\|\sigma) := \max_{0\leq P\leq I, \Tr\{P\rho\}\geq 1-\epsilon} (-\log\Tr\{P\sigma\}).
\end{equation*}
We use the notation $D_H^\epsilon$ since this is equivalent to the hypothesis testing relative entropy defined in Ref.~\cite{hypothesis}. 

One can then define corresponding resource measures by the minimum divergence with free states:
\begin{eqnarray*}
    \mathfrak{D}_{\max(\min)}(\rho) &:=& \min_{\sigma\in\mF}D_{\max(\min)}(\rho\|\sigma).
\end{eqnarray*}
Due to the data processing inequalities for $\dmax$ \cite{qrenyi}, $\dmin$ \cite{LIEB1973267,uhlmann1977,PETZ198657} and the purified distance $P(\rho,\sigma) = \sqrt{1-f(\rho,\sigma)}$ \cite{marco_thesis}, it holds that $\mathfrak{D}_{\max,\min}$ are monotonically non-increasing ($\mathfrak{f}$ is non-decreasing) under all free operations. 
The smoothed versions of these resource measures are simply defined by replacing the divergences with smoothed ones.
Another important type of resource measure is the \emph{free robustness/log-robustness}:
\begin{eqnarray*}
    R(\rho) &:=& \min\{s\geq 0: \exists\sigma\in \mF, \frac{1}{1+s}\rho + \frac{s}{1+s}\sigma \in \mF\},\\
    {LR}(\rho) &:=& \log(1+ {R}(\rho)).
\end{eqnarray*}
The smoothed versions are similarly given by minimizing over $\mathcal{B}^\epsilon(\rho)$.
By definition, if $\mF$ is an affine set, i.e.~any state expressed by an affine combination of free states is free (in e.g.~coherence, asymmetry theories), then any resource state $\rho\notin \mF$ does not have finite free robustness, although infinite free robustness does not necessarily indicate that $\mF$ is affine (see Appendix B \cite{sm}). 
We formally introduce the following condition for $\mF$ for convenience of later discussions:
\begin{condffr}
All states have finite free robustness, i.e. $R(\rho)<\infty, \forall \rho$.
\end{condffr}
By allowing $\sigma$ to be any state (instead of a free state) in the definition of free robustness, one obtains the so-called generalized robustness/log-robustness, $R_G$/$LR_G$. It can be easily verified that $LR_G(\rho) = \mathfrak{D}_{\max}(\rho)$.


We next briefly overview the theory of \emph{resource destroying (RD) maps} \cite{rdm}.   A map $\lambda$ from states to states is an RD map if it satisfies the following conditions: i) mapping all
non-free states to free states, i.e.~$\forall\rho\not\in \mF, \lambda(\rho)\in \mF$; ii) preserving free states,
i.e.~$\forall\sigma\in \mF, \lambda(\sigma)=\sigma$. 
Two types of RD maps are of particular importance:  i) \emph{Exact RD maps}, which output the closest free state as measured by the relative entropy.  Simple forms are known in e.g.~coherence, asymmetry and non-Gaussianity theories. (See Appendix C \cite{sm} for a detailed introduction.) ii) \emph{RD channels}.  They often induce desirable features, e.g.~the output free state is continuous under variation of the input state due to data processing inequalities. Examples include the dephasing channel for coherence theory and the twirling channel for asymmetry theory.
In Appendix D \cite{sm}, we show that if any state takes finite free robustness, then there does not exist an RD channel in that theory. 

An RD map $\lambda$ induces typical classes of quantum channels via a collection of simple, general conditions.  This work focuses on the following two important ones:  i) the resource non-generating operations $\mathscr{F}_{{\rm NG}}:=\{\mathcal{E}\,|\,\lambda\circ\mathcal{E}\circ\lambda = \mathcal{E}\circ\lambda\}$ \footnote{Note that this class of free operations is uniquely determined by the set of free states and do not depend on the choice of RD maps in a certain theory}, which induces the maximal set of free operations in the sense that any other operation can create resource from a free state; ii) the commuting operations $\mathscr{F}_{\lambda,{\rm Comm}}=\{\mathcal{E}\,|\,\lambda\circ\mathcal{E} = \mathcal{E}\circ\lambda\}$. One can then construct simple resource measures $\delta_{\lambda}(\cdot)=\delta(\cdot,\lambda(\cdot))$ where $\delta$ is any contractive distance measure, which is monotonically non-increasing under the commuting operations \cite{rdm}. 
Here we shall use $\mathfrak{D}_{\max(\min),\lambda}(\rho) := D_{\max(\min)}(\rho\|\lambda(\rho))$, with smoothed versions defined by minimizing (or maximizing) over $\mathcal{B}^\epsilon(\rho)$.  The counterpart for free fidelity is similarly given by $\mathfrak{f}_{\lambda}(\rho):=f(\rho,\lambda(\rho))$. 

We implicitly assume that the resource measures appearing throughout the paper are well-defined (the free robustness case is highlighted since it is of crucial importance in resource theories). 

\section{Resource currencies and modification coefficients}\label{sec:currency}
Resource manipulation tasks are commonly defined relative to some standard or unit resource that serve as the ``currency'': the formation task is about preparing the target state with a supply of standard resource, while the distillation task is about producing standard resource from the given state.  
More generally, consider some family of states $\{\phi_d\}$ consisting of a state $\phi_d\in\mD(\mH_d)$ for each different $d\in\mathbb{D}$ where $\mathbb{D}\subseteq\mathbb{Z}_+$ is a set of valid dimensions as a definition of a resource currency (e.g.~$\mathbb{D}=\{2^k\},k\in\mathbb{Z}_+$ for multi-qubit theories), and call them \emph{reference states}.  
Also let $d^{\downarrow(\uparrow)}\in\mathbb{D}$ be some dimension smaller (greater) than $d$ (e.g.~take $d^{\uparrow} = d+1$ when $\mathbb{D} = \mathbb{Z}_+$).
We then introduce the following \emph{modification coefficients}, which will naturally emerge in the later discussions on one-shot rates:
\begin{eqnarray}
m_f(\phi_d) &:=& {-\log\mathfrak{f}(\phi_d)}/{\log d}, \label{eq:modification fidelity}\\
m_{\max(\min)}(\phi_d) &:=&   {\mathfrak{D}_{\max(\min)}(\phi_d)}/{\log d}, \label{eq:modification minmax}\\
m_{LR}(\phi_d)&:=& LR(\phi_d)/{\log d}.
\end{eqnarray}
Similarly, $m_{f,\lambda}$ and $m_{\max(\min),\lambda}$ are defined by using $\mathfrak{f}_{\lambda}$ and $\mathfrak{D}_{\max(\min),\lambda}$ for Eqs.~\eqref{eq:modification fidelity} and \eqref{eq:modification minmax}.   

It is common to consider certain notions of ``maximum'' resource states as the reference states (so that the formation and distillation tasks essentially achieve the effect of dilution and concentration of resource respectively), although one can in principle choose more general classes of states.
The modification coefficients of max-resource states may encode key features of the resource theory, such as the ``size'' of the set of free states. For example, compare the qubit coherence and magic state theories: in the Bloch representation, the incoherent states only form a zero-measure axis, while the stabilizer (non-magic) states form an octahedron which occupies a significant chunk of the Bloch sphere \cite{bravyikitaev};  Loosely speaking, the maximum magic state is thus much ``closer'' to stabilizer states, which leads to a smaller modification coefficient, as compared with the case of coherence.

Now, we point out the remarkable fact that there is a family of pure max-resource states such that the different types of modification coefficients may collapse to the same value, in a generic class of theories.
To this end, we introduce the following condition.
\begin{condch}
The set of free states $\mF$ is formed by a convex hull of pure free states.
\end{condch}
This property is quite lenient and holds for many theories such as entanglement, coherence, superposition, magic states. We then obtain the following:
\begin{thm}\label{thm:collapse_main}
 Suppose the resource theory satisfies Condition (CH). Then, for any $d$, there exists a pure state $\hat\Phi_d\in\mD(\mH_d)$ such that $m_f(\hat\Phi_d) = m_{\min}(\hat\Phi_d) = m_{\max}(\hat\Phi_d) := g_d$ where $\hat\Phi_d$ achieves the maxima of $m_f,m_{\min},m_{\max}$. 
Furthermore, if $\tilde\lambda$ is an exact RD map, then $m_{f,\tilde\lambda}(\hat\Phi_d) = m_{\min,\tilde\lambda}(\hat\Phi_d) = m_{\max,\tilde\lambda}(\hat\Phi_d) = g_d$.
\end{thm}
See Appendices E and F \cite{sm} for proofs and discussions concerning this result.
We call such $\hat\Phi_d$ a \emph{golden state} and $g_d$ the \emph{golden coefficient} for dimension $d$.   

Now we briefly discuss a few important examples of golden states and coefficients.  The coherence theory comes with golden states $\ket{\hat\Phi_d}=\frac{1}{\sqrt{d}}\sum_j\ket{j}$, and the complete dephasing channel is an exact RD map; Theorem~\ref{thm:collapse_main} fully applies, and $g_d = 1$ for all $d\in\mathbb{D}=\mathbb{Z}_+$.   
For entanglement theory, golden states can take the form $\ket{\hat\Phi_d}=\frac{1}{d^{1/4}}\sum_{j=1}^{\sqrt{d}}\ket{j}\ket{j}$ where $d$ is the dimension of the bipartite system with local dimension $\sqrt{d}$, and  $g_d = 1/2$ for all $d\in\mathbb{D} = \{k^2|k\in\mathbb{Z}_+\}$.  Note that the simple forms of one-shot entanglement/coherence manipulation results \cite{brandaodatta,coh_dilution,PhysRevLett.121.010401,fullstory} rely heavily on this specific property of the golden coefficients being constant for any valid dimension. 
The theory of magic states has golden state $\hat\Phi_2=\frac12(I+(X+Y+Z)/\sqrt{3})$ with $g_2 = \log(3-\sqrt{3})\approx 0.34$ \cite{regula} for a single-qubit system, where $X,Y,Z$ are Pauli matrices. 
Another interesting case is the theory of quantum thermodynamics, where the only free state is the Gibbs state and Condition (CH) is not satisfied. But it can be shown that golden states with the same maximal resource and collapsing properties still exist and the $g_d$ can be easily calculated (see Appendix E \cite{sm}). In particular, the infinite-temperature case (i.e.~the purity theory) has $g_d=1$ for all $d$ (where every pure state is a golden state). 



\section{Optimal rates of one-shot resource manipulation}
Before stating the results, we define another useful condition defined for the set of free states $\mF$ and pure reference states $\{\Phi_d\}$, which we call Condition (CT).
\begin{condct} 
For any given $d\in\mathbb{D}$, $\Tr\{\Phi_d\sigma\}$ is constant for any $\sigma\in\mF$ (equivalently, $\Phi_d$ belongs to the dual set of $\mF$ as introduced in Ref.~\cite{PhysRevA.95.062314}).
\end{condct}
For instance, for the theory of coherence, it can be easily verified that $\Tr\{\hat{\Phi}_d\sigma\} = 1/d$ for any incoherent state $\sigma$, so that the Condition (CT) is satisfied when the golden states are chosen as reference states. 
In Appendix G \cite{sm}, we provide a new example based on a multi-qubit superposition theory.  A diagram that illustrates the classification of resource theories relevant to this work can be found in Appendix H \cite{sm}.

For state $\rho$, given reference states $\{\phi_d\}$, the optimal rate of one-shot formation task with $\epsilon$ error tolerance i.e.~the \emph{one-shot $\epsilon$-formation cost}, under the set of operations $\mathscr{F}$, is defined to be the minimum size of reference state that achieves the task:
\begin{align*}
    &\Omega_{C,\mathscr{F}}^\epsilon(\rho\leftarrow\{\phi_d\}) \\&\quad:= \log\min\{d\in\mathbb{D}: \exists \mathcal{E}\in \mathscr{F}, \mathcal{E}(\phi_d)\in\mathcal{B}^\epsilon(\rho)\}.
\end{align*}
Below let $\mathfrak{R}$ be some resource measure and $m$ be some type of modification coefficient that will be specified.
The following theorem establishes general bounds for the one-shot $\epsilon$-formation cost under the two aforementioned classes of free operations (proofs in Appendices I, J \cite{sm}):   
\begin{thm}\label{thm:cost}
 For reference states $\{\phi_d\}$, let $d_0 = \min\{d\in\mathbb{D}: \mathfrak{R}(\phi_d) \geq \mathfrak{R}^\epsilon(\rho)\}$. Then
\begin{equation}\label{eq:1}
    \Omega_{C,\mathscr{F}}^\epsilon(\rho\leftarrow\{\phi_d\}) \geq \frac{\mathfrak{R}^\epsilon(\rho)}{m(\phi_{d_0})}.
\end{equation}
for
i) $\mathscr{F}=\mathscr{F}_{{\rm NG}}$,  $\mathfrak{R}=\mathfrak{D}_{\max}$, $m=m_{\max}$ for any $\mF$;
ii) $\mathscr{F}=\mathscr{F}_{{\rm NG}}$, $\mathfrak{R}=LR$, $m=m_{LR}$ for $\mF$ satisfying Condition (FFR); 
iii) $\mathscr{F}=\mathscr{F}_{\lambda,{\rm Comm}}$, $\mathfrak{R}=\mathfrak{D}_{\max,\lambda}$, $m=m_{\max,\lambda}$ for any $\mF$ and $\lambda$.

On the other hand, for pure reference states $\{\Phi_d\}$, let $d'_0 = \min\{d\in\mathbb{D}:-\log\mathfrak{f}(\Phi_d) \geq \mathfrak{R}^\epsilon(\rho)\}$.   
Then 
\begin{equation}\label{eq:2}
    \Omega_{C,\mathscr{F}}^\epsilon(\rho\leftarrow\{\Phi_d\}) < \frac{\mathfrak{R}^\epsilon(\rho)}{m_f(\Phi_{{d'_0}^{\downarrow}})} + \log\frac{d'_0}{{d'_0}^{\downarrow}}.
\end{equation}
for i) $\mathscr{F}=\mathscr{F}_{{\rm NG}}$, $\mathfrak{R}=\mathfrak{D}_{\max}$ for $\mF$ satisfying Condition (CT);
ii) $\mathscr{F}=\mathscr{F}_{{\rm NG}}$, $\mathfrak{R}=LR$ for convex $\mF$ satisfying Condition (FFR);
iii) $\mathscr{F}=\mathscr{F}_{\lambda,{\rm Comm}}$, $\mathfrak{R}=\mathfrak{D}_{\max,\lambda}$ for $\mF$ satisfying Condition (CT) and any $\lambda$.
\end{thm}

By combining the above results with Theorem~\ref{thm:collapse_main}, we can reduce the modification coefficients to golden ones and obtain roughly matching bounds:
\begin{cor}
For golden states $\{\hat\Phi_d\}$, suppose Conditions (CH) and (CT) are satisfied, and let $d_0 = \min\{d\in\mathbb{D}: g_d\log{d} \geq \mathfrak{R}^\epsilon(\rho)\}$.  Then
\begin{equation}
  \frac{\mathfrak{R}^\epsilon(\rho)}{g_{d_0}}  \leq \Omega_{C,\mathscr{F}}^\epsilon(\rho\leftarrow\{\hat\Phi_d\})  < \frac{\mathfrak{R}^\epsilon(\rho)}{g_{d_0^{\downarrow}}} + \log\frac{d_0}{d_0^{\downarrow}}.
\end{equation}
for i) $\mathscr{F}=\mathscr{F}_{{\rm NG}}$, $\mathfrak{R}=\mathfrak{D}_{\max}$;
ii) $\mathscr{F}=\mathscr{F}_{\tilde{\lambda},{\rm Comm}}$, $\mathfrak{R}=\mathfrak{D}_{\max,\tilde{\lambda}}$ for exact RD map $\tilde{\lambda}$.
\end{cor}

The constructions used for showing the achievable formation costs provide interesting implications to the existence of \emph{root states}, max-resource states in the strongest sense, from which any state defined on the same Hilbert space can be obtained by some free operation:
\begin{cor}\label{cor:root}
 For any $\mF(\mH_d)$ such that the maxima of $m\in\{m_f,m_{\min},m_{\max}\}$ coincide at some pure (golden) state $\hat\Phi_d$ (e.g. $\mF(\mH_d)$ satisfying Condition (CH)), $\hat\Phi_d$ serves as a root state if $\mF(\mH_d)$ further satisfies either of the following: i) Condition (CT), ii) Condition (FFR) and $m_{\max}(\Phi_d)=m_{LR}(\Phi_d)$ for any pure state $\Phi_d\in \mD(\mH_d)$. 
\end{cor}
We provide the proof, as well as an extensive discussion on root states, in Appendix K \cite{sm}. 
This in particular implies that if there exist no root states, then the free and generalized robustness measures do not coincide at pure states in general. (see Ref. \cite{Contreras2019entanglement} for a related discussion for the theory of multipartite entanglement).

As for distillation, we consider the standard version with error tolerance on the output.
The optimal rate, namely the \emph{one-shot $\epsilon$-distillation yield}, under free operations $\mathscr{F}$, is defined to be the maximum size of the target reference state:
\begin{align*}
    &\Omega_{D,\mathscr{F}}^\epsilon(\rho\rightarrow\{\phi_d\})\\ &\quad:= \log\max\{d\in\mathbb{D}: \exists \mathcal{E}\in \mathscr{F}, \mathcal{E}(\rho)\in\mathcal{B}^\epsilon(\phi_d)\}. 
\end{align*}

We first provide the following bounds for the one-shot $\epsilon$-distillation yield under resource non-generating operations (proofs and additional results in Appendix L \cite{sm}): 
\begin{thm}\label{thm:dist_ng}
For pure reference states $\{\Phi_d\}$,  let $d_0 = \max\{d\in\mathbb{D}: -\log\mathfrak{f}(\Phi_d)  \leq \mathfrak{D}_H^\epsilon(\rho)\}$.  Then for any $\mF$,
  \begin{equation}
  \Omega_{D,\mathbb{\barx}}^{\epsilon}(\rho\rightarrow\{\Phi_d\}) \leq \frac{\mathfrak{D}_{H}^\epsilon(\rho)}{m_f(\Phi_{d_0})}.  \label{eq:opsm_main} 
 \end{equation}

 Suppose further that $\mF$ satisfies Condition (FFR).
For reference states $\{\phi_d\}$, let $d_0 = \max\{d\in\mathbb{D}: LR(\phi_d) \leq \mathfrak{D}_H^\epsilon(\rho)\}$. 
Then
 \begin{equation}
     \Omega_{D,\mathbb{\barx}}^{\epsilon}(\rho\rightarrow\{\phi_d\})  > \frac{\mathfrak{D}_{H}^\epsilon(\rho)}{m_{LR}(\phi_{d_0^{\uparrow}})}-\log{\frac{d_0^{\uparrow}}{d_0}}.  \label{eq:opsmlr_main} 
 \end{equation}
 \end{thm}
 
 For commuting operations, we find the following upper bound (proof in Appendix M \cite{sm}):
\begin{thm}\label{thm:dist_comm}
For pure reference states $\{\Phi_d\}$ and RD channel $\Lambda$, let $d_0 = \max\{d\in\mathbb{D}: \mathfrak{f}_\Lambda(\Phi_d)  \geq 2^{-\mathfrak{D}_{H,\Lambda}^\epsilon(\rho)}-2\sqrt{\epsilon}\}$.  Then for any $\mF$,
  \begin{equation}
  \Omega_{D,\mathscr{F}_{{\Lambda},{\rm Comm}}}^{\epsilon}(\rho\rightarrow\{\Phi_d\})\leq \frac{-\log(2^{-\mathfrak{D}_{H,\Lambda}^\epsilon(\rho)}-2\sqrt{\epsilon})}{m_{f,\Lambda}(\Phi_{d_0})}.  \label{eq:opsm2_main} 
 \end{equation}
\end{thm}
For now we are only able to obtain lower bounds for some special notions of commuting operations (see Appendix M \cite{sm}).

Moreover, in Appendix N \cite{sm}, we instead consider distillation with error tolerance on the input, for which a greater collection of bounds in slightly different forms (using e.g.~state-smoothing $\mathfrak{D}_{\min}^\epsilon$ or continuity bounds) can be established.

These results allow us to obtain nontrivial bounds for resource trading in specific theories by computing the modification coefficients (which can be efficiently done in many cases \cite{vidal_1999,Steiner2003robustness,Harrow2003robustness,jafarizadeh_2005,howard_2017,ringbauer_2018,johnston_2018}). For example, the golden coefficients of coherence, entanglement and purity theories induce bounds directly given by the smooth resource measures without modification, which is consistent with previous results \cite{brandaodatta,coh_dilution,PhysRevLett.121.010401,Streltsov_2018}. 
As a more informative example, we briefly remark on the theory of magic states.
It can be inferred from recent results in \cite{bravyi2018simulation} that $m_f(\Phi)=m_{\min}(\Phi)=m_{\max}(\Phi)$ holds for the so-called ``Clifford magic states'' $\Phi$, and $m(\Phi_2^{\otimes m})= m(\Phi_2)$ where $m\in\{m_f, m_{\min}, m_{\max}\}$ for any qubit pure state $\Phi_2$  \footnote{Logarithm of the stabilizer extent introduced in \cite{bravyi2018simulation} is equivalent to the max-relative entropy of magic \cite{regula}.}  (meanwhile, $m_{LR}$ is generically larger and non-constant).  
This in particular is relevant to the conventional magic state distillation where the reference states are copies of $\ket{T}:=(\ket{0}+e^{i\pi/4}\ket{1})/\sqrt{2}$. 
By using $m(T^{\otimes m})=\log(4-2\sqrt{2})\approx 0.23$ as can be easily verified and the known values of $m_{LR}(T^{\otimes m})$ \cite{howard_2017,Heinrich2019robustnessofmagic}, one can obtain several bounds for manipulating multiple $T$-states/gates under all stabilizer-preserving operations, which complements the recent results in Ref.~\cite{wang_wilde_su} for a slightly different setup. 
We leave extended discussions on the implications to magic states and quantum computation for follow-up works.

We also note that the resource measures considered in this work often admit efficient SDP formulation \cite{howard_2017,ringbauer_2018} as well as analytical expressions \cite{vidal_1999,Steiner2003robustness,Harrow2003robustness,jafarizadeh_2005,howard_2017,johnston_2018}, which make our bounds of practical use in many important circumstances.


\section{Concluding remarks}
This work establishes general bounds that relate the optimal rates of typical one-shot resource formation and distillation tasks to resource monotones based on one-shot divergences and log-robustness, without specifying the resource theory.   
We introduce the modification coefficients to take into account the resourcefulness of the currency, and find that they exhibit the remarkable collapsing property for a simple notion of max-resource states.  
We examined two important classes of free operations, namely the resource non-generating operations and operations that commute with the RD map.  

Our results not only provide nontrivial and practically useful bounds for these tasks, but also characterize the resourcefulness of quantum states defined in general resource theories in terms of direct one-shot resource conversion, providing general operational meanings to the resource measures discussed in this work. 
They are potentially applicable to a large class of theories beyond the specific ones studied earlier (e.g.~entanglement, coherence, thermal non-equilibrium), allowing one to obtain nontrivial bounds for optimal resource manipulation in specific contexts.
Our results also complement the studies on the complete set of monotones \cite{PhysRevA.95.062314,girard_2017,gour_2018_majorization,2019arXiv190108127T}, which provide the necessary and sufficient conditions for state transformations between two states under free operations.  
A complete set of monotones generally consists of infinite number of resource monotones \cite{COECKE201659}, which makes the computation impractical. Therefore, the simpler expressions obtained in this work would give clearer insights into resource manipulation tasks.


For future work, it would be intriguing to further investigate the achievability of these fundamental limits (especially for distillation), apply this framework to specific contexts such as magic states and superposition to gain new insights into these theories, explore the connections and implications to the asymptotic theory, and extend the ideas to resource theory settings beyond quantum states, in accordance with \cite{2018arXiv181202572L,2018arXiv180802607G,2019arXiv190108127T,2015arXiv151108818D,COECKE201659,fritz_2017,LiuWinter2018,LiuYuan:channel,uola2018quantifying,Oszmaniec2019,rosset_2018}.

\smallskip

\begin{acknowledgments}
We thank Mil\'an Mosonyi for discussions on R\'enyi divergences and Bartosz Regula for pointing out an error on Fig.~2 in Supplemental Material.
ZWL acknowledges support by AFOSR, ARO, and Perimeter Institute for Theoretical Physics.  Research at Perimeter Institute is supported by the Government
of Canada through Industry Canada and by the Province of Ontario through the Ministry
of Research and Innovation.
KB acknowledges the Templeton Religion Trust for the partial support of this research under grant TRT0159
and  Zhejiang University for the support of an Academic Award for
Outstanding Doctoral Candidates.
RT acknowledges the support
of NSF, ARO, IARPA, and the Takenaka Scholarship Foundation.
\end{acknowledgments}

\smallskip

\emph{Note added.}  Recently, we became aware of a related work by M.\ K.\ Vijayan et al.\ \cite{Vijayan19}.

\bibliography{one-shot}

\begin{thebibliography}{96}%
\makeatletter
\providecommand \@ifxundefined [1]{%
 \@ifx{#1\undefined}
}%
\providecommand \@ifnum [1]{%
 \ifnum #1\expandafter \@firstoftwo
 \else \expandafter \@secondoftwo
 \fi
}%
\providecommand \@ifx [1]{%
 \ifx #1\expandafter \@firstoftwo
 \else \expandafter \@secondoftwo
 \fi
}%
\providecommand \natexlab [1]{#1}%
\providecommand \enquote  [1]{``#1''}%
\providecommand \bibnamefont  [1]{#1}%
\providecommand \bibfnamefont [1]{#1}%
\providecommand \citenamefont [1]{#1}%
\providecommand \href@noop [0]{\@secondoftwo}%
\providecommand \href [0]{\begingroup \@sanitize@url \@href}%
\providecommand \@href[1]{\@@startlink{#1}\@@href}%
\providecommand \@@href[1]{\endgroup#1\@@endlink}%
\providecommand \@sanitize@url [0]{\catcode `\\12\catcode `\$12\catcode
  `\&12\catcode `\#12\catcode `\^12\catcode `\_12\catcode `\%12\relax}%
\providecommand \@@startlink[1]{}%
\providecommand \@@endlink[0]{}%
\providecommand \url  [0]{\begingroup\@sanitize@url \@url }%
\providecommand \@url [1]{\endgroup\@href {#1}{\urlprefix }}%
\providecommand \urlprefix  [0]{URL }%
\providecommand \Eprint [0]{\href }%
\providecommand \doibase [0]{http://dx.doi.org/}%
\providecommand \selectlanguage [0]{\@gobble}%
\providecommand \bibinfo  [0]{\@secondoftwo}%
\providecommand \bibfield  [0]{\@secondoftwo}%
\providecommand \translation [1]{[#1]}%
\providecommand \BibitemOpen [0]{}%
\providecommand \bibitemStop [0]{}%
\providecommand \bibitemNoStop [0]{.\EOS\space}%
\providecommand \EOS [0]{\spacefactor3000\relax}%
\providecommand \BibitemShut  [1]{\csname bibitem#1\endcsname}%
\let\auto@bib@innerbib\@empty
\bibitem [{\citenamefont {Chitambar}\ and\ \citenamefont
  {Gour}(2019)}]{qrtreview}%
  \BibitemOpen
  \bibfield  {author} {\bibinfo {author} {\bibfnamefont {Eric}\ \bibnamefont
  {Chitambar}}\ and\ \bibinfo {author} {\bibfnamefont {Gilad}\ \bibnamefont
  {Gour}},\ }\bibfield  {title} {\enquote {\bibinfo {title} {Quantum resource
  theories},}\ }\href {\doibase 10.1103/RevModPhys.91.025001} {\bibfield
  {journal} {\bibinfo  {journal} {Rev. Mod. Phys.}\ }\textbf {\bibinfo {volume}
  {91}},\ \bibinfo {pages} {025001} (\bibinfo {year} {2019})}\BibitemShut
  {NoStop}%
\bibitem [{\citenamefont {Plenio}\ and\ \citenamefont
  {Virmani}(2007)}]{virmani}%
  \BibitemOpen
  \bibfield  {author} {\bibinfo {author} {\bibfnamefont {Martin~B.}\
  \bibnamefont {Plenio}}\ and\ \bibinfo {author} {\bibfnamefont {Shashank}\
  \bibnamefont {Virmani}},\ }\bibfield  {title} {\enquote {\bibinfo {title} {An
  introduction to entanglement measures},}\ }\href
  {http://dl.acm.org/citation.cfm?id=2011706.2011707} {\bibfield  {journal}
  {\bibinfo  {journal} {Quantum Info. Comput.}\ }\textbf {\bibinfo {volume}
  {7}},\ \bibinfo {pages} {1--51} (\bibinfo {year} {2007})}\BibitemShut
  {NoStop}%
\bibitem [{\citenamefont {Horodecki}\ \emph {et~al.}(2009)\citenamefont
  {Horodecki}, \citenamefont {Horodecki}, \citenamefont {Horodecki},\ and\
  \citenamefont {Horodecki}}]{entrevrt}%
  \BibitemOpen
  \bibfield  {author} {\bibinfo {author} {\bibfnamefont {Ryszard}\ \bibnamefont
  {Horodecki}}, \bibinfo {author} {\bibfnamefont {Pawe\l{}}\ \bibnamefont
  {Horodecki}}, \bibinfo {author} {\bibfnamefont {Micha\l{}}\ \bibnamefont
  {Horodecki}}, \ and\ \bibinfo {author} {\bibfnamefont {Karol}\ \bibnamefont
  {Horodecki}},\ }\bibfield  {title} {\enquote {\bibinfo {title} {Quantum
  entanglement},}\ }\href {\doibase 10.1103/RevModPhys.81.865} {\bibfield
  {journal} {\bibinfo  {journal} {Rev. Mod. Phys.}\ }\textbf {\bibinfo {volume}
  {81}},\ \bibinfo {pages} {865--942} (\bibinfo {year} {2009})}\BibitemShut
  {NoStop}%
\bibitem [{\citenamefont {\AA{}berg}(2006)}]{Aberg2006}%
  \BibitemOpen
  \bibfield  {author} {\bibinfo {author} {\bibfnamefont {J.}~\bibnamefont
  {\AA{}berg}},\ }\bibfield  {title} {\enquote {\bibinfo {title} {{Quantifying
  Superposition}},}\ }\href@noop {} {\bibfield  {journal} {\bibinfo  {journal}
  {eprint arXiv:quant-ph/0612146}\ } (\bibinfo {year} {2006})},\ \Eprint
  {http://arxiv.org/abs/quant-ph/0612146} {quant-ph/0612146} \BibitemShut
  {NoStop}%
\bibitem [{\citenamefont {Baumgratz}\ \emph {et~al.}(2014)\citenamefont
  {Baumgratz}, \citenamefont {Cramer},\ and\ \citenamefont
  {Plenio}}]{PhysRevLett.113.140401}%
  \BibitemOpen
  \bibfield  {author} {\bibinfo {author} {\bibfnamefont {T.}~\bibnamefont
  {Baumgratz}}, \bibinfo {author} {\bibfnamefont {M.}~\bibnamefont {Cramer}}, \
  and\ \bibinfo {author} {\bibfnamefont {M.~B.}\ \bibnamefont {Plenio}},\
  }\bibfield  {title} {\enquote {\bibinfo {title} {Quantifying coherence},}\
  }\href {\doibase 10.1103/PhysRevLett.113.140401} {\bibfield  {journal}
  {\bibinfo  {journal} {Phys. Rev. Lett.}\ }\textbf {\bibinfo {volume} {113}},\
  \bibinfo {pages} {140401} (\bibinfo {year} {2014})}\BibitemShut {NoStop}%
\bibitem [{\citenamefont {Streltsov}\ \emph
  {et~al.}(2017{\natexlab{a}})\citenamefont {Streltsov}, \citenamefont
  {Adesso},\ and\ \citenamefont {Plenio}}]{Streltsov2017rev}%
  \BibitemOpen
  \bibfield  {author} {\bibinfo {author} {\bibfnamefont {Alexander}\
  \bibnamefont {Streltsov}}, \bibinfo {author} {\bibfnamefont {Gerardo}\
  \bibnamefont {Adesso}}, \ and\ \bibinfo {author} {\bibfnamefont {Martin~B.}\
  \bibnamefont {Plenio}},\ }\bibfield  {title} {\enquote {\bibinfo {title}
  {Colloquium: Quantum coherence as a resource},}\ }\href {\doibase
  10.1103/RevModPhys.89.041003} {\bibfield  {journal} {\bibinfo  {journal}
  {Rev. Mod. Phys.}\ }\textbf {\bibinfo {volume} {89}},\ \bibinfo {pages}
  {041003} (\bibinfo {year} {2017}{\natexlab{a}})}\BibitemShut {NoStop}%
\bibitem [{\citenamefont {Theurer}\ \emph {et~al.}(2017)\citenamefont
  {Theurer}, \citenamefont {Killoran}, \citenamefont {Egloff},\ and\
  \citenamefont {Plenio}}]{Theurer2017}%
  \BibitemOpen
  \bibfield  {author} {\bibinfo {author} {\bibfnamefont {T.}~\bibnamefont
  {Theurer}}, \bibinfo {author} {\bibfnamefont {N.}~\bibnamefont {Killoran}},
  \bibinfo {author} {\bibfnamefont {D.}~\bibnamefont {Egloff}}, \ and\ \bibinfo
  {author} {\bibfnamefont {M.~B.}\ \bibnamefont {Plenio}},\ }\bibfield  {title}
  {\enquote {\bibinfo {title} {Resource theory of superposition},}\ }\href
  {\doibase 10.1103/PhysRevLett.119.230401} {\bibfield  {journal} {\bibinfo
  {journal} {Phys. Rev. Lett.}\ }\textbf {\bibinfo {volume} {119}},\ \bibinfo
  {pages} {230401} (\bibinfo {year} {2017})}\BibitemShut {NoStop}%
\bibitem [{\citenamefont {Veitch}\ \emph {et~al.}(2014)\citenamefont {Veitch},
  \citenamefont {Mousavian}, \citenamefont {Gottesman},\ and\ \citenamefont
  {Emerson}}]{magic}%
  \BibitemOpen
  \bibfield  {author} {\bibinfo {author} {\bibfnamefont {Victor}\ \bibnamefont
  {Veitch}}, \bibinfo {author} {\bibfnamefont {S~A~Hamed}\ \bibnamefont
  {Mousavian}}, \bibinfo {author} {\bibfnamefont {Daniel}\ \bibnamefont
  {Gottesman}}, \ and\ \bibinfo {author} {\bibfnamefont {Joseph}\ \bibnamefont
  {Emerson}},\ }\bibfield  {title} {\enquote {\bibinfo {title} {The resource
  theory of stabilizer quantum computation},}\ }\href
  {http://stacks.iop.org/1367-2630/16/i=1/a=013009} {\bibfield  {journal}
  {\bibinfo  {journal} {New J. Phys.}\ }\textbf {\bibinfo {volume} {16}},\
  \bibinfo {pages} {013009} (\bibinfo {year} {2014})}\BibitemShut {NoStop}%
\bibitem [{\citenamefont {Howard}\ and\ \citenamefont
  {Campbell}(2017)}]{howard_2017}%
  \BibitemOpen
  \bibfield  {author} {\bibinfo {author} {\bibfnamefont {Mark}\ \bibnamefont
  {Howard}}\ and\ \bibinfo {author} {\bibfnamefont {Earl}\ \bibnamefont
  {Campbell}},\ }\bibfield  {title} {\enquote {\bibinfo {title} {Application of
  a {{Resource Theory}} for {{Magic States}} to {{Fault}}-{{Tolerant Quantum
  Computing}}},}\ }\href {\doibase 10.1103/PhysRevLett.118.090501} {\bibfield
  {journal} {\bibinfo  {journal} {Phys. Rev. Lett.}\ }\textbf {\bibinfo
  {volume} {118}},\ \bibinfo {pages} {090501} (\bibinfo {year}
  {2017})}\BibitemShut {NoStop}%
\bibitem [{\citenamefont {Gour}\ and\ \citenamefont {Spekkens}(2008)}]{frame}%
  \BibitemOpen
  \bibfield  {author} {\bibinfo {author} {\bibfnamefont {Gilad}\ \bibnamefont
  {Gour}}\ and\ \bibinfo {author} {\bibfnamefont {Robert~W}\ \bibnamefont
  {Spekkens}},\ }\bibfield  {title} {\enquote {\bibinfo {title} {The resource
  theory of quantum reference frames: manipulations and monotones},}\ }\href
  {http://stacks.iop.org/1367-2630/10/i=3/a=033023} {\bibfield  {journal}
  {\bibinfo  {journal} {New J. Phys.}\ }\textbf {\bibinfo {volume} {10}},\
  \bibinfo {pages} {033023} (\bibinfo {year} {2008})}\BibitemShut {NoStop}%
\bibitem [{\citenamefont {Marvian}\ and\ \citenamefont
  {Spekkens}(2016)}]{speakable}%
  \BibitemOpen
  \bibfield  {author} {\bibinfo {author} {\bibfnamefont {Iman}\ \bibnamefont
  {Marvian}}\ and\ \bibinfo {author} {\bibfnamefont {Robert~W.}\ \bibnamefont
  {Spekkens}},\ }\bibfield  {title} {\enquote {\bibinfo {title} {How to
  quantify coherence: Distinguishing speakable and unspeakable notions},}\
  }\href {\doibase 10.1103/PhysRevA.94.052324} {\bibfield  {journal} {\bibinfo
  {journal} {Phys. Rev. A}\ }\textbf {\bibinfo {volume} {94}},\ \bibinfo
  {pages} {052324} (\bibinfo {year} {2016})}\BibitemShut {NoStop}%
\bibitem [{\citenamefont {Horodecki}\ \emph
  {et~al.}(2003{\natexlab{a}})\citenamefont {Horodecki}, \citenamefont
  {Horodecki},\ and\ \citenamefont {Oppenheim}}]{Horodecki2003purity}%
  \BibitemOpen
  \bibfield  {author} {\bibinfo {author} {\bibfnamefont {Micha\l{}}\
  \bibnamefont {Horodecki}}, \bibinfo {author} {\bibfnamefont {Pawe\l{}}\
  \bibnamefont {Horodecki}}, \ and\ \bibinfo {author} {\bibfnamefont
  {Jonathan}\ \bibnamefont {Oppenheim}},\ }\bibfield  {title} {\enquote
  {\bibinfo {title} {Reversible transformations from pure to mixed states and
  the unique measure of information},}\ }\href {\doibase
  10.1103/PhysRevA.67.062104} {\bibfield  {journal} {\bibinfo  {journal} {Phys.
  Rev. A}\ }\textbf {\bibinfo {volume} {67}},\ \bibinfo {pages} {062104}
  (\bibinfo {year} {2003}{\natexlab{a}})}\BibitemShut {NoStop}%
\bibitem [{\citenamefont {Chiribella}\ and\ \citenamefont
  {Scandolo}(2017)}]{Chiribella_2017}%
  \BibitemOpen
  \bibfield  {author} {\bibinfo {author} {\bibfnamefont {Giulio}\ \bibnamefont
  {Chiribella}}\ and\ \bibinfo {author} {\bibfnamefont {Carlo~Maria}\
  \bibnamefont {Scandolo}},\ }\bibfield  {title} {\enquote {\bibinfo {title}
  {Microcanonical thermodynamics in general physical theories},}\ }\href
  {\doibase 10.1088/1367-2630/aa91c7} {\bibfield  {journal} {\bibinfo
  {journal} {New Journal of Physics}\ }\textbf {\bibinfo {volume} {19}},\
  \bibinfo {pages} {123043} (\bibinfo {year} {2017})}\BibitemShut {NoStop}%
\bibitem [{\citenamefont {Horodecki}\ and\ \citenamefont
  {Oppenheim}(2013{\natexlab{a}})}]{HorodeckiOppenheim13}%
  \BibitemOpen
  \bibfield  {author} {\bibinfo {author} {\bibfnamefont {Micha{\l}}\
  \bibnamefont {Horodecki}}\ and\ \bibinfo {author} {\bibfnamefont {Jonathan}\
  \bibnamefont {Oppenheim}},\ }\bibfield  {title} {\enquote {\bibinfo {title}
  {Fundamental limitations for quantum and nanoscale thermodynamics},}\ }\href
  {https://doi.org/10.1038/ncomms3059} {\bibfield  {journal} {\bibinfo
  {journal} {Nature Communications}\ }\textbf {\bibinfo {volume} {4}},\
  \bibinfo {pages} {2059} (\bibinfo {year} {2013}{\natexlab{a}})}\BibitemShut
  {NoStop}%
\bibitem [{\citenamefont {Brand\~ao}\ \emph {et~al.}(2013)\citenamefont
  {Brand\~ao}, \citenamefont {Horodecki}, \citenamefont {Oppenheim},
  \citenamefont {Renes},\ and\ \citenamefont {Spekkens}}]{thermort}%
  \BibitemOpen
  \bibfield  {author} {\bibinfo {author} {\bibfnamefont {Fernando G. S.~L.}\
  \bibnamefont {Brand\~ao}}, \bibinfo {author} {\bibfnamefont {Micha\l{}}\
  \bibnamefont {Horodecki}}, \bibinfo {author} {\bibfnamefont {Jonathan}\
  \bibnamefont {Oppenheim}}, \bibinfo {author} {\bibfnamefont {Joseph~M.}\
  \bibnamefont {Renes}}, \ and\ \bibinfo {author} {\bibfnamefont {Robert~W.}\
  \bibnamefont {Spekkens}},\ }\bibfield  {title} {\enquote {\bibinfo {title}
  {Resource theory of quantum states out of thermal equilibrium},}\ }\href
  {\doibase 10.1103/PhysRevLett.111.250404} {\bibfield  {journal} {\bibinfo
  {journal} {Phys. Rev. Lett.}\ }\textbf {\bibinfo {volume} {111}},\ \bibinfo
  {pages} {250404} (\bibinfo {year} {2013})}\BibitemShut {NoStop}%
\bibitem [{\citenamefont {Brand{\~a}o}\ \emph {et~al.}(2015)\citenamefont
  {Brand{\~a}o}, \citenamefont {Horodecki}, \citenamefont {Ng}, \citenamefont
  {Oppenheim},\ and\ \citenamefont {Wehner}}]{Brandao_secondlaws2015}%
  \BibitemOpen
  \bibfield  {author} {\bibinfo {author} {\bibfnamefont {Fernando}\
  \bibnamefont {Brand{\~a}o}}, \bibinfo {author} {\bibfnamefont {Micha{\l}}\
  \bibnamefont {Horodecki}}, \bibinfo {author} {\bibfnamefont {Nelly}\
  \bibnamefont {Ng}}, \bibinfo {author} {\bibfnamefont {Jonathan}\ \bibnamefont
  {Oppenheim}}, \ and\ \bibinfo {author} {\bibfnamefont {Stephanie}\
  \bibnamefont {Wehner}},\ }\bibfield  {title} {\enquote {\bibinfo {title} {The
  second laws of quantum thermodynamics},}\ }\href {\doibase
  10.1073/pnas.1411728112} {\bibfield  {journal} {\bibinfo  {journal}
  {Proceedings of the National Academy of Sciences}\ }\textbf {\bibinfo
  {volume} {112}},\ \bibinfo {pages} {3275--3279} (\bibinfo {year}
  {2015})}\BibitemShut {NoStop}%
\bibitem [{\citenamefont {Genoni}\ \emph {et~al.}(2008)\citenamefont {Genoni},
  \citenamefont {Paris},\ and\ \citenamefont {Banaszek}}]{Genoni2008}%
  \BibitemOpen
  \bibfield  {author} {\bibinfo {author} {\bibfnamefont {Marco~G.}\
  \bibnamefont {Genoni}}, \bibinfo {author} {\bibfnamefont {Matteo G.~A.}\
  \bibnamefont {Paris}}, \ and\ \bibinfo {author} {\bibfnamefont {Konrad}\
  \bibnamefont {Banaszek}},\ }\bibfield  {title} {\enquote {\bibinfo {title}
  {Quantifying the non-{G}aussian character of a quantum state by quantum
  relative entropy},}\ }\href {\doibase 10.1103/PhysRevA.78.060303} {\bibfield
  {journal} {\bibinfo  {journal} {Phys. Rev. A}\ }\textbf {\bibinfo {volume}
  {78}},\ \bibinfo {pages} {060303} (\bibinfo {year} {2008})}\BibitemShut
  {NoStop}%
\bibitem [{\citenamefont {Takagi}\ and\ \citenamefont
  {Zhuang}(2018)}]{Takagi2018}%
  \BibitemOpen
  \bibfield  {author} {\bibinfo {author} {\bibfnamefont {Ryuji}\ \bibnamefont
  {Takagi}}\ and\ \bibinfo {author} {\bibfnamefont {Quntao}\ \bibnamefont
  {Zhuang}},\ }\bibfield  {title} {\enquote {\bibinfo {title} {Convex resource
  theory of non-gaussianity},}\ }\href {\doibase 10.1103/PhysRevA.97.062337}
  {\bibfield  {journal} {\bibinfo  {journal} {Phys. Rev. A}\ }\textbf {\bibinfo
  {volume} {97}},\ \bibinfo {pages} {062337} (\bibinfo {year}
  {2018})}\BibitemShut {NoStop}%
\bibitem [{\citenamefont {Albarelli}\ \emph {et~al.}(2018)\citenamefont
  {Albarelli}, \citenamefont {Genoni}, \citenamefont {Paris},\ and\
  \citenamefont {Ferraro}}]{albarelli2018resource}%
  \BibitemOpen
  \bibfield  {author} {\bibinfo {author} {\bibfnamefont {Francesco}\
  \bibnamefont {Albarelli}}, \bibinfo {author} {\bibfnamefont {Marco~G.}\
  \bibnamefont {Genoni}}, \bibinfo {author} {\bibfnamefont {Matteo G.~A.}\
  \bibnamefont {Paris}}, \ and\ \bibinfo {author} {\bibfnamefont {Alessandro}\
  \bibnamefont {Ferraro}},\ }\bibfield  {title} {\enquote {\bibinfo {title}
  {Resource theory of quantum non-{G}aussianity and wigner negativity},}\
  }\href {\doibase 10.1103/PhysRevA.98.052350} {\bibfield  {journal} {\bibinfo
  {journal} {Phys. Rev. A}\ }\textbf {\bibinfo {volume} {98}},\ \bibinfo
  {pages} {052350} (\bibinfo {year} {2018})}\BibitemShut {NoStop}%
\bibitem [{\citenamefont {Horodecki}\ and\ \citenamefont
  {Oppenheim}(2013{\natexlab{b}})}]{Horodecki2013quantumness}%
  \BibitemOpen
  \bibfield  {author} {\bibinfo {author} {\bibfnamefont {Michal}\ \bibnamefont
  {Horodecki}}\ and\ \bibinfo {author} {\bibfnamefont {Jonathan}\ \bibnamefont
  {Oppenheim}},\ }\bibfield  {title} {\enquote {\bibinfo {title}
  {({Q}uantumness in the context of) resource theories},}\ }\href {\doibase
  10.1142/S0217979213450197} {\bibfield  {journal} {\bibinfo  {journal}
  {International Journal of Modern Physics B}\ }\textbf {\bibinfo {volume}
  {27}},\ \bibinfo {pages} {1345019} (\bibinfo {year}
  {2013}{\natexlab{b}})}\BibitemShut {NoStop}%
\bibitem [{\citenamefont {Brand\~ao}\ and\ \citenamefont {Gour}(2015)}]{bg}%
  \BibitemOpen
  \bibfield  {author} {\bibinfo {author} {\bibfnamefont {Fernando G. S.~L.}\
  \bibnamefont {Brand\~ao}}\ and\ \bibinfo {author} {\bibfnamefont {Gilad}\
  \bibnamefont {Gour}},\ }\bibfield  {title} {\enquote {\bibinfo {title}
  {Reversible framework for quantum resource theories},}\ }\href {\doibase
  10.1103/PhysRevLett.115.070503} {\bibfield  {journal} {\bibinfo  {journal}
  {Phys. Rev. Lett.}\ }\textbf {\bibinfo {volume} {115}},\ \bibinfo {pages}
  {070503} (\bibinfo {year} {2015})}\BibitemShut {NoStop}%
\bibitem [{\citenamefont {{del Rio}}\ \emph {et~al.}(2015)\citenamefont {{del
  Rio}}, \citenamefont {{Kraemer}},\ and\ \citenamefont
  {{Renner}}}]{2015arXiv151108818D}%
  \BibitemOpen
  \bibfield  {author} {\bibinfo {author} {\bibfnamefont {L.}~\bibnamefont {{del
  Rio}}}, \bibinfo {author} {\bibfnamefont {L.}~\bibnamefont {{Kraemer}}}, \
  and\ \bibinfo {author} {\bibfnamefont {R.}~\bibnamefont {{Renner}}},\
  }\bibfield  {title} {\enquote {\bibinfo {title} {{Resource theories of
  knowledge}},}\ }\href@noop {} {\bibfield  {journal} {\bibinfo  {journal}
  {ArXiv e-prints}\ } (\bibinfo {year} {2015})},\ \Eprint
  {http://arxiv.org/abs/1511.08818} {arXiv:1511.08818 [quant-ph]} \BibitemShut
  {NoStop}%
\bibitem [{\citenamefont {Coecke}\ \emph {et~al.}(2016)\citenamefont {Coecke},
  \citenamefont {Fritz},\ and\ \citenamefont {Spekkens}}]{COECKE201659}%
  \BibitemOpen
  \bibfield  {author} {\bibinfo {author} {\bibfnamefont {Bob}\ \bibnamefont
  {Coecke}}, \bibinfo {author} {\bibfnamefont {Tobias}\ \bibnamefont {Fritz}},
  \ and\ \bibinfo {author} {\bibfnamefont {Robert~W.}\ \bibnamefont
  {Spekkens}},\ }\bibfield  {title} {\enquote {\bibinfo {title} {A mathematical
  theory of resources},}\ }\href {\doibase
  https://doi.org/10.1016/j.ic.2016.02.008} {\bibfield  {journal} {\bibinfo
  {journal} {Information and Computation}\ }\textbf {\bibinfo {volume} {250}},\
  \bibinfo {pages} {59 -- 86} (\bibinfo {year} {2016})},\ \bibinfo {note}
  {{Q}uantum Physics and Logic}\BibitemShut {NoStop}%
\bibitem [{\citenamefont {Fritz}(2017)}]{fritz_2017}%
  \BibitemOpen
  \bibfield  {author} {\bibinfo {author} {\bibfnamefont {Tobias}\ \bibnamefont
  {Fritz}},\ }\bibfield  {title} {\enquote {\bibinfo {title} {Resource
  convertibility and ordered commutative monoids},}\ }\href {\doibase
  10.1017/S0960129515000444} {\bibfield  {journal} {\bibinfo  {journal}
  {Mathematical Structures in Computer Science}\ }\textbf {\bibinfo {volume}
  {27}},\ \bibinfo {pages} {850--938} (\bibinfo {year} {2017})}\BibitemShut
  {NoStop}%
\bibitem [{\citenamefont {Liu}\ \emph {et~al.}(2017)\citenamefont {Liu},
  \citenamefont {Hu},\ and\ \citenamefont {Lloyd}}]{rdm}%
  \BibitemOpen
  \bibfield  {author} {\bibinfo {author} {\bibfnamefont {Zi-Wen}\ \bibnamefont
  {Liu}}, \bibinfo {author} {\bibfnamefont {Xueyuan}\ \bibnamefont {Hu}}, \
  and\ \bibinfo {author} {\bibfnamefont {Seth}\ \bibnamefont {Lloyd}},\
  }\bibfield  {title} {\enquote {\bibinfo {title} {Resource destroying maps},}\
  }\href {\doibase 10.1103/PhysRevLett.118.060502} {\bibfield  {journal}
  {\bibinfo  {journal} {Phys. Rev. Lett.}\ }\textbf {\bibinfo {volume} {118}},\
  \bibinfo {pages} {060502} (\bibinfo {year} {2017})}\BibitemShut {NoStop}%
\bibitem [{\citenamefont {Gour}(2017)}]{PhysRevA.95.062314}%
  \BibitemOpen
  \bibfield  {author} {\bibinfo {author} {\bibfnamefont {Gilad}\ \bibnamefont
  {Gour}},\ }\bibfield  {title} {\enquote {\bibinfo {title} {Quantum resource
  theories in the single-shot regime},}\ }\href {\doibase
  10.1103/PhysRevA.95.062314} {\bibfield  {journal} {\bibinfo  {journal} {Phys.
  Rev. A}\ }\textbf {\bibinfo {volume} {95}},\ \bibinfo {pages} {062314}
  (\bibinfo {year} {2017})}\BibitemShut {NoStop}%
\bibitem [{\citenamefont {{Sparaciari}}\ \emph {et~al.}(2018)\citenamefont
  {{Sparaciari}}, \citenamefont {{del Rio}}, \citenamefont {{Scandolo}},
  \citenamefont {{Faist}},\ and\ \citenamefont
  {{Oppenheim}}}]{2018arXiv180604937S}%
  \BibitemOpen
  \bibfield  {author} {\bibinfo {author} {\bibfnamefont {Carlo}\ \bibnamefont
  {{Sparaciari}}}, \bibinfo {author} {\bibfnamefont {Lidia}\ \bibnamefont {{del
  Rio}}}, \bibinfo {author} {\bibfnamefont {Carlo~Maria}\ \bibnamefont
  {{Scandolo}}}, \bibinfo {author} {\bibfnamefont {Philippe}\ \bibnamefont
  {{Faist}}}, \ and\ \bibinfo {author} {\bibfnamefont {Jonathan}\ \bibnamefont
  {{Oppenheim}}},\ }\bibfield  {title} {\enquote {\bibinfo {title} {{The first
  law of general quantum resource theories}},}\ }\href@noop {} {\bibfield
  {journal} {\bibinfo  {journal} {arXiv e-prints}\ ,\ \bibinfo {eid}
  {arXiv:1806.04937}} (\bibinfo {year} {2018})},\ \Eprint
  {http://arxiv.org/abs/1806.04937} {arXiv:1806.04937 [quant-ph]} \BibitemShut
  {NoStop}%
\bibitem [{\citenamefont {{Gour}}(2018)}]{2018arXiv180802607G}%
  \BibitemOpen
  \bibfield  {author} {\bibinfo {author} {\bibfnamefont {Gilad}\ \bibnamefont
  {{Gour}}},\ }\bibfield  {title} {\enquote {\bibinfo {title} {{Comparison of
  Quantum Channels with Superchannels}},}\ }\href@noop {} {\bibfield  {journal}
  {\bibinfo  {journal} {arXiv e-prints}\ ,\ \bibinfo {eid} {arXiv:1808.02607}}
  (\bibinfo {year} {2018})},\ \Eprint {http://arxiv.org/abs/1808.02607}
  {arXiv:1808.02607 [quant-ph]} \BibitemShut {NoStop}%
\bibitem [{\citenamefont {Takagi}\ \emph {et~al.}(2019)\citenamefont {Takagi},
  \citenamefont {Regula}, \citenamefont {Bu}, \citenamefont {Liu},\ and\
  \citenamefont {Adesso}}]{2018arXiv180901672T}%
  \BibitemOpen
  \bibfield  {author} {\bibinfo {author} {\bibfnamefont {Ryuji}\ \bibnamefont
  {Takagi}}, \bibinfo {author} {\bibfnamefont {Bartosz}\ \bibnamefont
  {Regula}}, \bibinfo {author} {\bibfnamefont {Kaifeng}\ \bibnamefont {Bu}},
  \bibinfo {author} {\bibfnamefont {Zi-Wen}\ \bibnamefont {Liu}}, \ and\
  \bibinfo {author} {\bibfnamefont {Gerardo}\ \bibnamefont {Adesso}},\
  }\bibfield  {title} {\enquote {\bibinfo {title} {Operational advantage of
  quantum resources in subchannel discrimination},}\ }\href {\doibase
  10.1103/PhysRevLett.122.140402} {\bibfield  {journal} {\bibinfo  {journal}
  {Phys. Rev. Lett.}\ }\textbf {\bibinfo {volume} {122}},\ \bibinfo {pages}
  {140402} (\bibinfo {year} {2019})}\BibitemShut {NoStop}%
\bibitem [{\citenamefont {Anshu}\ \emph {et~al.}(2018)\citenamefont {Anshu},
  \citenamefont {Hsieh},\ and\ \citenamefont {Jain}}]{anshu_2018catalytic}%
  \BibitemOpen
  \bibfield  {author} {\bibinfo {author} {\bibfnamefont {Anurag}\ \bibnamefont
  {Anshu}}, \bibinfo {author} {\bibfnamefont {Min-Hsiu}\ \bibnamefont {Hsieh}},
  \ and\ \bibinfo {author} {\bibfnamefont {Rahul}\ \bibnamefont {Jain}},\
  }\bibfield  {title} {\enquote {\bibinfo {title} {Quantifying resources in
  general resource theory with catalysts},}\ }\href {\doibase
  10.1103/PhysRevLett.121.190504} {\bibfield  {journal} {\bibinfo  {journal}
  {Phys. Rev. Lett.}\ }\textbf {\bibinfo {volume} {121}},\ \bibinfo {pages}
  {190504} (\bibinfo {year} {2018})}\BibitemShut {NoStop}%
\bibitem [{\citenamefont {Regula}(2018)}]{regula}%
  \BibitemOpen
  \bibfield  {author} {\bibinfo {author} {\bibfnamefont {Bartosz}\ \bibnamefont
  {Regula}},\ }\bibfield  {title} {\enquote {\bibinfo {title} {Convex geometry
  of quantum resource quantification},}\ }\href
  {http://stacks.iop.org/1751-8121/51/i=4/a=045303} {\bibfield  {journal}
  {\bibinfo  {journal} {Journal of Physics A: Mathematical and Theoretical}\
  }\textbf {\bibinfo {volume} {51}},\ \bibinfo {pages} {045303} (\bibinfo
  {year} {2018})}\BibitemShut {NoStop}%
\bibitem [{\citenamefont {Lami}\ \emph {et~al.}(2018)\citenamefont {Lami},
  \citenamefont {Regula}, \citenamefont {Wang}, \citenamefont {Nichols},
  \citenamefont {Winter},\ and\ \citenamefont {Adesso}}]{lami_2018}%
  \BibitemOpen
  \bibfield  {author} {\bibinfo {author} {\bibfnamefont {Ludovico}\
  \bibnamefont {Lami}}, \bibinfo {author} {\bibfnamefont {Bartosz}\
  \bibnamefont {Regula}}, \bibinfo {author} {\bibfnamefont {Xin}\ \bibnamefont
  {Wang}}, \bibinfo {author} {\bibfnamefont {Rosanna}\ \bibnamefont {Nichols}},
  \bibinfo {author} {\bibfnamefont {Andreas}\ \bibnamefont {Winter}}, \ and\
  \bibinfo {author} {\bibfnamefont {Gerardo}\ \bibnamefont {Adesso}},\
  }\bibfield  {title} {\enquote {\bibinfo {title} {Gaussian quantum resource
  theories},}\ }\href {\doibase 10.1103/PhysRevA.98.022335} {\bibfield
  {journal} {\bibinfo  {journal} {Phys. Rev. A}\ }\textbf {\bibinfo {volume}
  {98}},\ \bibinfo {pages} {022335} (\bibinfo {year} {2018})}\BibitemShut
  {NoStop}%
\bibitem [{\citenamefont {{Li}}\ \emph {et~al.}(2018)\citenamefont {{Li}},
  \citenamefont {{Bu}},\ and\ \citenamefont {{Liu}}}]{2018arXiv181202572L}%
  \BibitemOpen
  \bibfield  {author} {\bibinfo {author} {\bibfnamefont {Lu}~\bibnamefont
  {{Li}}}, \bibinfo {author} {\bibfnamefont {Kaifeng}\ \bibnamefont {{Bu}}}, \
  and\ \bibinfo {author} {\bibfnamefont {Zi-Wen}\ \bibnamefont {{Liu}}},\
  }\bibfield  {title} {\enquote {\bibinfo {title} {{Quantifying the resource
  content of quantum channels: An operational approach}},}\ }\href@noop {}
  {\bibfield  {journal} {\bibinfo  {journal} {arXiv e-prints}\ ,\ \bibinfo
  {eid} {arXiv:1812.02572}} (\bibinfo {year} {2018})},\ \Eprint
  {http://arxiv.org/abs/1812.02572} {arXiv:1812.02572 [quant-ph]} \BibitemShut
  {NoStop}%
\bibitem [{\citenamefont {Uola}\ \emph {et~al.}(2019)\citenamefont {Uola},
  \citenamefont {Kraft}, \citenamefont {Shang}, \citenamefont {Yu},\ and\
  \citenamefont {G\"uhne}}]{uola2018quantifying}%
  \BibitemOpen
  \bibfield  {author} {\bibinfo {author} {\bibfnamefont {Roope}\ \bibnamefont
  {Uola}}, \bibinfo {author} {\bibfnamefont {Tristan}\ \bibnamefont {Kraft}},
  \bibinfo {author} {\bibfnamefont {Jiangwei}\ \bibnamefont {Shang}}, \bibinfo
  {author} {\bibfnamefont {Xiao-Dong}\ \bibnamefont {Yu}}, \ and\ \bibinfo
  {author} {\bibfnamefont {Otfried}\ \bibnamefont {G\"uhne}},\ }\bibfield
  {title} {\enquote {\bibinfo {title} {Quantifying quantum resources with conic
  programming},}\ }\href {\doibase 10.1103/PhysRevLett.122.130404} {\bibfield
  {journal} {\bibinfo  {journal} {Phys. Rev. Lett.}\ }\textbf {\bibinfo
  {volume} {122}},\ \bibinfo {pages} {130404} (\bibinfo {year}
  {2019})}\BibitemShut {NoStop}%
\bibitem [{\citenamefont {{Takagi}}\ and\ \citenamefont
  {{Regula}}(2019)}]{2019arXiv190108127T}%
  \BibitemOpen
  \bibfield  {author} {\bibinfo {author} {\bibfnamefont {Ryuji}\ \bibnamefont
  {{Takagi}}}\ and\ \bibinfo {author} {\bibfnamefont {Bartosz}\ \bibnamefont
  {{Regula}}},\ }\bibfield  {title} {\enquote {\bibinfo {title} {{General
  resource theories in quantum mechanics and beyond: operational
  characterization via discrimination tasks}},}\ }\href@noop {} {\bibfield
  {journal} {\bibinfo  {journal} {arXiv e-prints}\ ,\ \bibinfo {eid}
  {arXiv:1901.08127}} (\bibinfo {year} {2019})},\ \Eprint
  {http://arxiv.org/abs/1901.08127} {arXiv:1901.08127 [quant-ph]} \BibitemShut
  {NoStop}%
\bibitem [{\citenamefont {Oszmaniec}\ and\ \citenamefont
  {Biswas}(2019)}]{Oszmaniec2019}%
  \BibitemOpen
  \bibfield  {author} {\bibinfo {author} {\bibfnamefont {Micha{\l{}}}\
  \bibnamefont {Oszmaniec}}\ and\ \bibinfo {author} {\bibfnamefont {Tanmoy}\
  \bibnamefont {Biswas}},\ }\bibfield  {title} {\enquote {\bibinfo {title}
  {Operational relevance of resource theories of quantum measurements},}\
  }\href {\doibase 10.22331/q-2019-04-26-133} {\bibfield  {journal} {\bibinfo
  {journal} {{Quantum}}\ }\textbf {\bibinfo {volume} {3}},\ \bibinfo {pages}
  {133} (\bibinfo {year} {2019})}\BibitemShut {NoStop}%
\bibitem [{\citenamefont {{Liu}}\ and\ \citenamefont
  {{Yuan}}(2019)}]{LiuYuan:channel}%
  \BibitemOpen
  \bibfield  {author} {\bibinfo {author} {\bibfnamefont {Yunchao}\ \bibnamefont
  {{Liu}}}\ and\ \bibinfo {author} {\bibfnamefont {Xiao}\ \bibnamefont
  {{Yuan}}},\ }\bibfield  {title} {\enquote {\bibinfo {title} {{Operational
  Resource Theory of Quantum Channels}},}\ }\href@noop {} {\bibfield  {journal}
  {\bibinfo  {journal} {arXiv e-prints}\ ,\ \bibinfo {eid} {arXiv:1904.02680}}
  (\bibinfo {year} {2019})},\ \Eprint {http://arxiv.org/abs/1904.02680}
  {arXiv:1904.02680 [quant-ph]} \BibitemShut {NoStop}%
\bibitem [{\citenamefont {{Liu}}\ and\ \citenamefont
  {{Winter}}(2019)}]{LiuWinter2018}%
  \BibitemOpen
  \bibfield  {author} {\bibinfo {author} {\bibfnamefont {Zi-Wen}\ \bibnamefont
  {{Liu}}}\ and\ \bibinfo {author} {\bibfnamefont {Andreas}\ \bibnamefont
  {{Winter}}},\ }\bibfield  {title} {\enquote {\bibinfo {title} {{Resource
  theories of quantum channels and the universal role of resource erasure}},}\
  }\href@noop {} {\bibfield  {journal} {\bibinfo  {journal} {arXiv e-prints}\
  ,\ \bibinfo {eid} {arXiv:1904.04201}} (\bibinfo {year} {2019})},\ \Eprint
  {http://arxiv.org/abs/1904.04201} {arXiv:1904.04201 [quant-ph]} \BibitemShut
  {NoStop}%
\bibitem [{\citenamefont {Bennett}\ \emph
  {et~al.}(1996{\natexlab{a}})\citenamefont {Bennett}, \citenamefont
  {DiVincenzo}, \citenamefont {Smolin},\ and\ \citenamefont
  {Wootters}}]{PhysRevA.54.3824}%
  \BibitemOpen
  \bibfield  {author} {\bibinfo {author} {\bibfnamefont {Charles~H.}\
  \bibnamefont {Bennett}}, \bibinfo {author} {\bibfnamefont {David~P.}\
  \bibnamefont {DiVincenzo}}, \bibinfo {author} {\bibfnamefont {John~A.}\
  \bibnamefont {Smolin}}, \ and\ \bibinfo {author} {\bibfnamefont {William~K.}\
  \bibnamefont {Wootters}},\ }\bibfield  {title} {\enquote {\bibinfo {title}
  {Mixed-state entanglement and quantum error correction},}\ }\href {\doibase
  10.1103/PhysRevA.54.3824} {\bibfield  {journal} {\bibinfo  {journal} {Phys.
  Rev. A}\ }\textbf {\bibinfo {volume} {54}},\ \bibinfo {pages} {3824--3851}
  (\bibinfo {year} {1996}{\natexlab{a}})}\BibitemShut {NoStop}%
\bibitem [{\citenamefont {Hayden}\ \emph {et~al.}(2001)\citenamefont {Hayden},
  \citenamefont {Horodecki},\ and\ \citenamefont {Terhal}}]{Hayden_2001}%
  \BibitemOpen
  \bibfield  {author} {\bibinfo {author} {\bibfnamefont {Patrick~M}\
  \bibnamefont {Hayden}}, \bibinfo {author} {\bibfnamefont {Michal}\
  \bibnamefont {Horodecki}}, \ and\ \bibinfo {author} {\bibfnamefont
  {Barbara~M}\ \bibnamefont {Terhal}},\ }\bibfield  {title} {\enquote {\bibinfo
  {title} {The asymptotic entanglement cost of preparing a quantum state},}\
  }\href {\doibase 10.1088/0305-4470/34/35/314} {\bibfield  {journal} {\bibinfo
   {journal} {Journal of Physics A: Mathematical and General}\ }\textbf
  {\bibinfo {volume} {34}},\ \bibinfo {pages} {6891--6898} (\bibinfo {year}
  {2001})}\BibitemShut {NoStop}%
\bibitem [{\citenamefont {Bennett}\ \emph
  {et~al.}(1996{\natexlab{b}})\citenamefont {Bennett}, \citenamefont
  {Bernstein}, \citenamefont {Popescu},\ and\ \citenamefont
  {Schumacher}}]{PhysRevA.53.2046}%
  \BibitemOpen
  \bibfield  {author} {\bibinfo {author} {\bibfnamefont {Charles~H.}\
  \bibnamefont {Bennett}}, \bibinfo {author} {\bibfnamefont {Herbert~J.}\
  \bibnamefont {Bernstein}}, \bibinfo {author} {\bibfnamefont {Sandu}\
  \bibnamefont {Popescu}}, \ and\ \bibinfo {author} {\bibfnamefont {Benjamin}\
  \bibnamefont {Schumacher}},\ }\bibfield  {title} {\enquote {\bibinfo {title}
  {Concentrating partial entanglement by local operations},}\ }\href {\doibase
  10.1103/PhysRevA.53.2046} {\bibfield  {journal} {\bibinfo  {journal} {Phys.
  Rev. A}\ }\textbf {\bibinfo {volume} {53}},\ \bibinfo {pages} {2046--2052}
  (\bibinfo {year} {1996}{\natexlab{b}})}\BibitemShut {NoStop}%
\bibitem [{\citenamefont {Bennett}\ \emph
  {et~al.}(1996{\natexlab{c}})\citenamefont {Bennett}, \citenamefont
  {Brassard}, \citenamefont {Popescu}, \citenamefont {Schumacher},
  \citenamefont {Smolin},\ and\ \citenamefont {Wootters}}]{PhysRevLett.76.722}%
  \BibitemOpen
  \bibfield  {author} {\bibinfo {author} {\bibfnamefont {Charles~H.}\
  \bibnamefont {Bennett}}, \bibinfo {author} {\bibfnamefont {Gilles}\
  \bibnamefont {Brassard}}, \bibinfo {author} {\bibfnamefont {Sandu}\
  \bibnamefont {Popescu}}, \bibinfo {author} {\bibfnamefont {Benjamin}\
  \bibnamefont {Schumacher}}, \bibinfo {author} {\bibfnamefont {John~A.}\
  \bibnamefont {Smolin}}, \ and\ \bibinfo {author} {\bibfnamefont {William~K.}\
  \bibnamefont {Wootters}},\ }\bibfield  {title} {\enquote {\bibinfo {title}
  {Purification of noisy entanglement and faithful teleportation via noisy
  channels},}\ }\href {\doibase 10.1103/PhysRevLett.76.722} {\bibfield
  {journal} {\bibinfo  {journal} {Phys. Rev. Lett.}\ }\textbf {\bibinfo
  {volume} {76}},\ \bibinfo {pages} {722--725} (\bibinfo {year}
  {1996}{\natexlab{c}})}\BibitemShut {NoStop}%
\bibitem [{\citenamefont {Bravyi}\ and\ \citenamefont
  {Kitaev}(2005)}]{bravyikitaev}%
  \BibitemOpen
  \bibfield  {author} {\bibinfo {author} {\bibfnamefont {Sergey}\ \bibnamefont
  {Bravyi}}\ and\ \bibinfo {author} {\bibfnamefont {Alexei}\ \bibnamefont
  {Kitaev}},\ }\bibfield  {title} {\enquote {\bibinfo {title} {Universal
  quantum computation with ideal clifford gates and noisy ancillas},}\ }\href
  {\doibase 10.1103/PhysRevA.71.022316} {\bibfield  {journal} {\bibinfo
  {journal} {Phys. Rev. A}\ }\textbf {\bibinfo {volume} {71}},\ \bibinfo
  {pages} {022316} (\bibinfo {year} {2005})}\BibitemShut {NoStop}%
\bibitem [{\citenamefont {Brandao}\ and\ \citenamefont
  {Datta}(2011)}]{brandaodatta}%
  \BibitemOpen
  \bibfield  {author} {\bibinfo {author} {\bibfnamefont {F.~G. S.~L.}\
  \bibnamefont {Brandao}}\ and\ \bibinfo {author} {\bibfnamefont
  {N.}~\bibnamefont {Datta}},\ }\bibfield  {title} {\enquote {\bibinfo {title}
  {One-shot rates for entanglement manipulation under non-entangling maps},}\
  }\href {\doibase 10.1109/TIT.2011.2104531} {\bibfield  {journal} {\bibinfo
  {journal} {IEEE Transactions on Information Theory}\ }\textbf {\bibinfo
  {volume} {57}},\ \bibinfo {pages} {1754--1760} (\bibinfo {year}
  {2011})}\BibitemShut {NoStop}%
\bibitem [{\citenamefont {Buscemi}\ and\ \citenamefont
  {Datta}(2011)}]{BuscemiDatta11}%
  \BibitemOpen
  \bibfield  {author} {\bibinfo {author} {\bibfnamefont {Francesco}\
  \bibnamefont {Buscemi}}\ and\ \bibinfo {author} {\bibfnamefont {Nilanjana}\
  \bibnamefont {Datta}},\ }\bibfield  {title} {\enquote {\bibinfo {title}
  {Entanglement cost in practical scenarios},}\ }\href {\doibase
  10.1103/PhysRevLett.106.130503} {\bibfield  {journal} {\bibinfo  {journal}
  {Phys. Rev. Lett.}\ }\textbf {\bibinfo {volume} {106}},\ \bibinfo {pages}
  {130503} (\bibinfo {year} {2011})}\BibitemShut {NoStop}%
\bibitem [{\citenamefont {Zhao}\ \emph {et~al.}(2018)\citenamefont {Zhao},
  \citenamefont {Liu}, \citenamefont {Yuan}, \citenamefont {Chitambar},\ and\
  \citenamefont {Ma}}]{coh_dilution}%
  \BibitemOpen
  \bibfield  {author} {\bibinfo {author} {\bibfnamefont {Qi}~\bibnamefont
  {Zhao}}, \bibinfo {author} {\bibfnamefont {Yunchao}\ \bibnamefont {Liu}},
  \bibinfo {author} {\bibfnamefont {Xiao}\ \bibnamefont {Yuan}}, \bibinfo
  {author} {\bibfnamefont {Eric}\ \bibnamefont {Chitambar}}, \ and\ \bibinfo
  {author} {\bibfnamefont {Xiongfeng}\ \bibnamefont {Ma}},\ }\bibfield  {title}
  {\enquote {\bibinfo {title} {One-shot coherence dilution},}\ }\href {\doibase
  10.1103/PhysRevLett.120.070403} {\bibfield  {journal} {\bibinfo  {journal}
  {Phys. Rev. Lett.}\ }\textbf {\bibinfo {volume} {120}},\ \bibinfo {pages}
  {070403} (\bibinfo {year} {2018})}\BibitemShut {NoStop}%
\bibitem [{\citenamefont {Regula}\ \emph {et~al.}(2018)\citenamefont {Regula},
  \citenamefont {Fang}, \citenamefont {Wang},\ and\ \citenamefont
  {Adesso}}]{PhysRevLett.121.010401}%
  \BibitemOpen
  \bibfield  {author} {\bibinfo {author} {\bibfnamefont {Bartosz}\ \bibnamefont
  {Regula}}, \bibinfo {author} {\bibfnamefont {Kun}\ \bibnamefont {Fang}},
  \bibinfo {author} {\bibfnamefont {Xin}\ \bibnamefont {Wang}}, \ and\ \bibinfo
  {author} {\bibfnamefont {Gerardo}\ \bibnamefont {Adesso}},\ }\bibfield
  {title} {\enquote {\bibinfo {title} {One-shot coherence distillation},}\
  }\href {\doibase 10.1103/PhysRevLett.121.010401} {\bibfield  {journal}
  {\bibinfo  {journal} {Phys. Rev. Lett.}\ }\textbf {\bibinfo {volume} {121}},\
  \bibinfo {pages} {010401} (\bibinfo {year} {2018})}\BibitemShut {NoStop}%
\bibitem [{\citenamefont {{Zhao}}\ \emph {et~al.}(2018)\citenamefont {{Zhao}},
  \citenamefont {{Liu}}, \citenamefont {{Yuan}}, \citenamefont {{Chitambar}},\
  and\ \citenamefont {{Winter}}}]{fullstory}%
  \BibitemOpen
  \bibfield  {author} {\bibinfo {author} {\bibfnamefont {Qi}~\bibnamefont
  {{Zhao}}}, \bibinfo {author} {\bibfnamefont {Yunchao}\ \bibnamefont {{Liu}}},
  \bibinfo {author} {\bibfnamefont {Xiao}\ \bibnamefont {{Yuan}}}, \bibinfo
  {author} {\bibfnamefont {Eric}\ \bibnamefont {{Chitambar}}}, \ and\ \bibinfo
  {author} {\bibfnamefont {Andreas}\ \bibnamefont {{Winter}}},\ }\bibfield
  {title} {\enquote {\bibinfo {title} {{One-Shot Coherence Distillation:
  Towards Completing the Picture}},}\ }\href@noop {} {\bibfield  {journal}
  {\bibinfo  {journal} {arXiv e-prints}\ ,\ \bibinfo {eid} {arXiv:1808.01885}}
  (\bibinfo {year} {2018})},\ \Eprint {http://arxiv.org/abs/1808.01885}
  {arXiv:1808.01885 [quant-ph]} \BibitemShut {NoStop}%
\bibitem [{\citenamefont {Dahlsten}\ \emph {et~al.}(2011)\citenamefont
  {Dahlsten}, \citenamefont {Renner}, \citenamefont {Rieper},\ and\
  \citenamefont {Vedral}}]{Dahlsten_2011}%
  \BibitemOpen
  \bibfield  {author} {\bibinfo {author} {\bibfnamefont {Oscar C~O}\
  \bibnamefont {Dahlsten}}, \bibinfo {author} {\bibfnamefont {Renato}\
  \bibnamefont {Renner}}, \bibinfo {author} {\bibfnamefont {Elisabeth}\
  \bibnamefont {Rieper}}, \ and\ \bibinfo {author} {\bibfnamefont {Vlatko}\
  \bibnamefont {Vedral}},\ }\bibfield  {title} {\enquote {\bibinfo {title}
  {Inadequacy of von neumann entropy for characterizing extractable work},}\
  }\href {\doibase 10.1088/1367-2630/13/5/053015} {\bibfield  {journal}
  {\bibinfo  {journal} {New Journal of Physics}\ }\textbf {\bibinfo {volume}
  {13}},\ \bibinfo {pages} {053015} (\bibinfo {year} {2011})}\BibitemShut
  {NoStop}%
\bibitem [{\citenamefont {Gour}\ \emph {et~al.}(2015)\citenamefont {Gour},
  \citenamefont {Müller}, \citenamefont {Narasimhachar}, \citenamefont
  {Spekkens},\ and\ \citenamefont {Halpern}}]{GOUR2015inf_oneshot}%
  \BibitemOpen
  \bibfield  {author} {\bibinfo {author} {\bibfnamefont {Gilad}\ \bibnamefont
  {Gour}}, \bibinfo {author} {\bibfnamefont {Markus~P.}\ \bibnamefont
  {Müller}}, \bibinfo {author} {\bibfnamefont {Varun}\ \bibnamefont
  {Narasimhachar}}, \bibinfo {author} {\bibfnamefont {Robert~W.}\ \bibnamefont
  {Spekkens}}, \ and\ \bibinfo {author} {\bibfnamefont {Nicole~Yunger}\
  \bibnamefont {Halpern}},\ }\bibfield  {title} {\enquote {\bibinfo {title}
  {The resource theory of informational nonequilibrium in thermodynamics},}\
  }\href {\doibase https://doi.org/10.1016/j.physrep.2015.04.003} {\bibfield
  {journal} {\bibinfo  {journal} {Physics Reports}\ }\textbf {\bibinfo {volume}
  {583}},\ \bibinfo {pages} {1 -- 58} (\bibinfo {year} {2015})},\ \bibinfo
  {note} {the resource theory of informational nonequilibrium in
  thermodynamics}\BibitemShut {NoStop}%
\bibitem [{\citenamefont {Yunger~Halpern}\ and\ \citenamefont
  {Renes}(2016)}]{Halpern2016_beyond1}%
  \BibitemOpen
  \bibfield  {author} {\bibinfo {author} {\bibfnamefont {Nicole}\ \bibnamefont
  {Yunger~Halpern}}\ and\ \bibinfo {author} {\bibfnamefont {Joseph~M.}\
  \bibnamefont {Renes}},\ }\bibfield  {title} {\enquote {\bibinfo {title}
  {Beyond heat baths: Generalized resource theories for small-scale
  thermodynamics},}\ }\href {\doibase 10.1103/PhysRevE.93.022126} {\bibfield
  {journal} {\bibinfo  {journal} {Phys. Rev. E}\ }\textbf {\bibinfo {volume}
  {93}},\ \bibinfo {pages} {022126} (\bibinfo {year} {2016})}\BibitemShut
  {NoStop}%
\bibitem [{\citenamefont {Halpern}(2018)}]{Halpern_2018_beyond2}%
  \BibitemOpen
  \bibfield  {author} {\bibinfo {author} {\bibfnamefont {Nicole~Yunger}\
  \bibnamefont {Halpern}},\ }\bibfield  {title} {\enquote {\bibinfo {title}
  {Beyond heat baths {II}: framework for generalized thermodynamic resource
  theories},}\ }\href {\doibase 10.1088/1751-8121/aaa62f} {\bibfield  {journal}
  {\bibinfo  {journal} {Journal of Physics A: Mathematical and Theoretical}\
  }\textbf {\bibinfo {volume} {51}},\ \bibinfo {pages} {094001} (\bibinfo
  {year} {2018})}\BibitemShut {NoStop}%
\bibitem [{\citenamefont {{Wang}}\ \emph {et~al.}(2018)\citenamefont {{Wang}},
  \citenamefont {{Wilde}},\ and\ \citenamefont {{Su}}}]{wang_wilde_su}%
  \BibitemOpen
  \bibfield  {author} {\bibinfo {author} {\bibfnamefont {Xin}\ \bibnamefont
  {{Wang}}}, \bibinfo {author} {\bibfnamefont {Mark~M.}\ \bibnamefont
  {{Wilde}}}, \ and\ \bibinfo {author} {\bibfnamefont {Yuan}\ \bibnamefont
  {{Su}}},\ }\bibfield  {title} {\enquote {\bibinfo {title} {{Efficiently
  computable bounds for magic state distillation}},}\ }\href@noop {} {\bibfield
   {journal} {\bibinfo  {journal} {arXiv e-prints}\ ,\ \bibinfo {eid}
  {arXiv:1812.10145}} (\bibinfo {year} {2018})},\ \Eprint
  {http://arxiv.org/abs/1812.10145} {arXiv:1812.10145 [quant-ph]} \BibitemShut
  {NoStop}%
\bibitem [{\citenamefont {Chitambar}\ and\ \citenamefont
  {Gour}(2016)}]{examination}%
  \BibitemOpen
  \bibfield  {author} {\bibinfo {author} {\bibfnamefont {Eric}\ \bibnamefont
  {Chitambar}}\ and\ \bibinfo {author} {\bibfnamefont {Gilad}\ \bibnamefont
  {Gour}},\ }\bibfield  {title} {\enquote {\bibinfo {title} {Critical
  examination of incoherent operations and a physically consistent resource
  theory of quantum coherence},}\ }\href {\doibase
  10.1103/PhysRevLett.117.030401} {\bibfield  {journal} {\bibinfo  {journal}
  {Phys. Rev. Lett.}\ }\textbf {\bibinfo {volume} {117}},\ \bibinfo {pages}
  {030401} (\bibinfo {year} {2016})}\BibitemShut {NoStop}%
\bibitem [{\citenamefont {Liu}\ \emph {et~al.}(2019)\citenamefont {Liu},
  \citenamefont {Takagi},\ and\ \citenamefont {Lloyd}}]{dqd}%
  \BibitemOpen
  \bibfield  {author} {\bibinfo {author} {\bibfnamefont {Zi-Wen}\ \bibnamefont
  {Liu}}, \bibinfo {author} {\bibfnamefont {Ryuji}\ \bibnamefont {Takagi}}, \
  and\ \bibinfo {author} {\bibfnamefont {Seth}\ \bibnamefont {Lloyd}},\
  }\bibfield  {title} {\enquote {\bibinfo {title} {Diagonal quantum discord},}\
  }\href {\doibase 10.1088/1751-8121/ab0774} {\bibfield  {journal} {\bibinfo
  {journal} {Journal of Physics A: Mathematical and Theoretical}\ }\textbf
  {\bibinfo {volume} {52}},\ \bibinfo {pages} {135301} (\bibinfo {year}
  {2019})}\BibitemShut {NoStop}%
\bibitem [{Note1()}]{Note1}%
  \BibitemOpen
  \bibinfo {note} {In fact, the definitions of max- and min-relative entropies
  only require $\sigma $ to be a positive semidefinite operator.}\BibitemShut
  {Stop}%
\bibitem [{\citenamefont {Datta}(2009)}]{dattarel}%
  \BibitemOpen
  \bibfield  {author} {\bibinfo {author} {\bibfnamefont {N.}~\bibnamefont
  {Datta}},\ }\bibfield  {title} {\enquote {\bibinfo {title} {Min- and
  max-relative entropies and a new entanglement monotone},}\ }\href {\doibase
  10.1109/TIT.2009.2018325} {\bibfield  {journal} {\bibinfo  {journal} {IEEE
  Transactions on Information Theory}\ }\textbf {\bibinfo {volume} {55}},\
  \bibinfo {pages} {2816--2826} (\bibinfo {year} {2009})}\BibitemShut {NoStop}%
\bibitem [{sm()}]{sm}%
  \BibitemOpen
  \href@noop {} {}\bibinfo {note} {See Supplemental Material for additional
  results, detailed proofs, and extended discussions, which includes
  Refs.\,\cite{Datta_2014,Wilde2014,PhysRevA.80.012307,marian,heat,Wilde:2013:QIT:2505455,watrous,brandao05,PhysRevLett.119.150405,nielsen,hho,PhysRevA.59.4206,Fuchs-vandeGraaf}}\BibitemShut
  {NoStop}%
\bibitem [{\citenamefont {Wang}\ and\ \citenamefont
  {Renner}(2012)}]{hypothesis}%
  \BibitemOpen
  \bibfield  {author} {\bibinfo {author} {\bibfnamefont {Ligong}\ \bibnamefont
  {Wang}}\ and\ \bibinfo {author} {\bibfnamefont {Renato}\ \bibnamefont
  {Renner}},\ }\bibfield  {title} {\enquote {\bibinfo {title} {One-shot
  classical-quantum capacity and hypothesis testing},}\ }\href {\doibase
  10.1103/PhysRevLett.108.200501} {\bibfield  {journal} {\bibinfo  {journal}
  {Phys. Rev. Lett.}\ }\textbf {\bibinfo {volume} {108}},\ \bibinfo {pages}
  {200501} (\bibinfo {year} {2012})}\BibitemShut {NoStop}%
\bibitem [{\citenamefont {M{\"u}ller-Lennert}\ \emph
  {et~al.}(2013)\citenamefont {M{\"u}ller-Lennert}, \citenamefont {Dupuis},
  \citenamefont {Szehr}, \citenamefont {Fehr},\ and\ \citenamefont
  {Tomamichel}}]{qrenyi}%
  \BibitemOpen
  \bibfield  {author} {\bibinfo {author} {\bibfnamefont {Martin}\ \bibnamefont
  {M{\"u}ller-Lennert}}, \bibinfo {author} {\bibfnamefont {Fr{\'e}d{\'e}ric}\
  \bibnamefont {Dupuis}}, \bibinfo {author} {\bibfnamefont {Oleg}\ \bibnamefont
  {Szehr}}, \bibinfo {author} {\bibfnamefont {Serge}\ \bibnamefont {Fehr}}, \
  and\ \bibinfo {author} {\bibfnamefont {Marco}\ \bibnamefont {Tomamichel}},\
  }\bibfield  {title} {\enquote {\bibinfo {title} {On quantum {R}{\'e}nyi
  entropies: A new generalization and some properties},}\ }\href {\doibase
  10.1063/1.4838856} {\bibfield  {journal} {\bibinfo  {journal} {Journal of
  Mathematical Physics}\ }\textbf {\bibinfo {volume} {54}},\ \bibinfo {pages}
  {122203} (\bibinfo {year} {2013})}\BibitemShut {NoStop}%
\bibitem [{\citenamefont {Lieb}(1973)}]{LIEB1973267}%
  \BibitemOpen
  \bibfield  {author} {\bibinfo {author} {\bibfnamefont {Elliott~H}\
  \bibnamefont {Lieb}},\ }\bibfield  {title} {\enquote {\bibinfo {title}
  {Convex trace functions and the {W}igner-{Y}anase-{D}yson conjecture},}\
  }\href {\doibase https://doi.org/10.1016/0001-8708(73)90011-X} {\bibfield
  {journal} {\bibinfo  {journal} {Advances in Mathematics}\ }\textbf {\bibinfo
  {volume} {11}},\ \bibinfo {pages} {267 -- 288} (\bibinfo {year}
  {1973})}\BibitemShut {NoStop}%
\bibitem [{\citenamefont {Uhlmann}(1977)}]{uhlmann1977}%
  \BibitemOpen
  \bibfield  {author} {\bibinfo {author} {\bibfnamefont {A.}~\bibnamefont
  {Uhlmann}},\ }\bibfield  {title} {\enquote {\bibinfo {title} {Relative
  entropy and the {W}igner-{Y}anase-{D}yson-{L}ieb concavity in an
  interpolation theory},}\ }\href
  {https://projecteuclid.org:443/euclid.cmp/1103900757} {\bibfield  {journal}
  {\bibinfo  {journal} {Comm. Math. Phys.}\ }\textbf {\bibinfo {volume} {54}},\
  \bibinfo {pages} {21--32} (\bibinfo {year} {1977})}\BibitemShut {NoStop}%
\bibitem [{\citenamefont {Petz}(1986)}]{PETZ198657}%
  \BibitemOpen
  \bibfield  {author} {\bibinfo {author} {\bibfnamefont {D{\'e}nes}\
  \bibnamefont {Petz}},\ }\bibfield  {title} {\enquote {\bibinfo {title}
  {Quasi-entropies for finite quantum systems},}\ }\href {\doibase
  https://doi.org/10.1016/0034-4877(86)90067-4} {\bibfield  {journal} {\bibinfo
   {journal} {Reports on Mathematical Physics}\ }\textbf {\bibinfo {volume}
  {23}},\ \bibinfo {pages} {57 -- 65} (\bibinfo {year} {1986})}\BibitemShut
  {NoStop}%
\bibitem [{\citenamefont {Tomamichel}(2012)}]{marco_thesis}%
  \BibitemOpen
  \bibfield  {author} {\bibinfo {author} {\bibfnamefont {Marco}\ \bibnamefont
  {Tomamichel}},\ }\emph {\bibinfo {title} {A Framework for Non-Asymptotic
  Quantum Information Theory}},\ \href@noop {} {Ph.D. thesis},\ \bibinfo
  {school} {ETH Z\"urich} (\bibinfo {year} {2012})\BibitemShut {NoStop}%
\bibitem [{Note2()}]{Note2}%
  \BibitemOpen
  \bibinfo {note} {Note that this class of free operations is uniquely
  determined by the set of free states and do not depend on the choice of RD
  maps in a certain theory}\BibitemShut {NoStop}%
\bibitem [{\citenamefont {Contreras-Tejada}\ \emph {et~al.}(2019)\citenamefont
  {Contreras-Tejada}, \citenamefont {Palazuelos},\ and\ \citenamefont
  {de~Vicente}}]{Contreras2019entanglement}%
  \BibitemOpen
  \bibfield  {author} {\bibinfo {author} {\bibfnamefont {Patricia}\
  \bibnamefont {Contreras-Tejada}}, \bibinfo {author} {\bibfnamefont {Carlos}\
  \bibnamefont {Palazuelos}}, \ and\ \bibinfo {author} {\bibfnamefont
  {Julio~I.}\ \bibnamefont {de~Vicente}},\ }\bibfield  {title} {\enquote
  {\bibinfo {title} {Resource theory of entanglement with a unique multipartite
  maximally entangled state},}\ }\href {\doibase
  10.1103/PhysRevLett.122.120503} {\bibfield  {journal} {\bibinfo  {journal}
  {Phys. Rev. Lett.}\ }\textbf {\bibinfo {volume} {122}},\ \bibinfo {pages}
  {120503} (\bibinfo {year} {2019})}\BibitemShut {NoStop}%
\bibitem [{\citenamefont {Vidal}\ and\ \citenamefont
  {Tarrach}(1999)}]{vidal_1999}%
  \BibitemOpen
  \bibfield  {author} {\bibinfo {author} {\bibfnamefont {Guifr\'e}\
  \bibnamefont {Vidal}}\ and\ \bibinfo {author} {\bibfnamefont {Rolf}\
  \bibnamefont {Tarrach}},\ }\bibfield  {title} {\enquote {\bibinfo {title}
  {Robustness of entanglement},}\ }\href {\doibase 10.1103/PhysRevA.59.141}
  {\bibfield  {journal} {\bibinfo  {journal} {Phys. Rev. A}\ }\textbf {\bibinfo
  {volume} {59}},\ \bibinfo {pages} {141--155} (\bibinfo {year}
  {1999})}\BibitemShut {NoStop}%
\bibitem [{\citenamefont {Steiner}(2003)}]{Steiner2003robustness}%
  \BibitemOpen
  \bibfield  {author} {\bibinfo {author} {\bibfnamefont {Michael}\ \bibnamefont
  {Steiner}},\ }\bibfield  {title} {\enquote {\bibinfo {title} {Generalized
  robustness of entanglement},}\ }\href {\doibase 10.1103/PhysRevA.67.054305}
  {\bibfield  {journal} {\bibinfo  {journal} {Phys. Rev. A}\ }\textbf {\bibinfo
  {volume} {67}},\ \bibinfo {pages} {054305} (\bibinfo {year}
  {2003})}\BibitemShut {NoStop}%
\bibitem [{\citenamefont {Harrow}\ and\ \citenamefont
  {Nielsen}(2003)}]{Harrow2003robustness}%
  \BibitemOpen
  \bibfield  {author} {\bibinfo {author} {\bibfnamefont {Aram~W.}\ \bibnamefont
  {Harrow}}\ and\ \bibinfo {author} {\bibfnamefont {Michael~A.}\ \bibnamefont
  {Nielsen}},\ }\bibfield  {title} {\enquote {\bibinfo {title} {Robustness of
  quantum gates in the presence of noise},}\ }\href {\doibase
  10.1103/PhysRevA.68.012308} {\bibfield  {journal} {\bibinfo  {journal} {Phys.
  Rev. A}\ }\textbf {\bibinfo {volume} {68}},\ \bibinfo {pages} {012308}
  (\bibinfo {year} {2003})}\BibitemShut {NoStop}%
\bibitem [{\citenamefont {Jafarizadeh}\ \emph {et~al.}(2005)\citenamefont
  {Jafarizadeh}, \citenamefont {Mirzaee},\ and\ \citenamefont
  {Rezaee}}]{jafarizadeh_2005}%
  \BibitemOpen
  \bibfield  {author} {\bibinfo {author} {\bibfnamefont {M.~A.}\ \bibnamefont
  {Jafarizadeh}}, \bibinfo {author} {\bibfnamefont {M.}~\bibnamefont
  {Mirzaee}}, \ and\ \bibinfo {author} {\bibfnamefont {M.}~\bibnamefont
  {Rezaee}},\ }\bibfield  {title} {\enquote {\bibinfo {title} {Exact
  calculation of robustness of entanglement via convex semi-definite
  programming},}\ }\href {\doibase 10.1142/S0219749905001043} {\bibfield
  {journal} {\bibinfo  {journal} {Int. J. Quantum Inform.}\ }\textbf {\bibinfo
  {volume} {03}},\ \bibinfo {pages} {511--533} (\bibinfo {year}
  {2005})}\BibitemShut {NoStop}%
\bibitem [{\citenamefont {Ringbauer}\ \emph {et~al.}(2018)\citenamefont
  {Ringbauer}, \citenamefont {Bromley}, \citenamefont {Cianciaruso},
  \citenamefont {Lami}, \citenamefont {Lau}, \citenamefont {Adesso},
  \citenamefont {White}, \citenamefont {Fedrizzi},\ and\ \citenamefont
  {Piani}}]{ringbauer_2018}%
  \BibitemOpen
  \bibfield  {author} {\bibinfo {author} {\bibfnamefont {Martin}\ \bibnamefont
  {Ringbauer}}, \bibinfo {author} {\bibfnamefont {Thomas~R.}\ \bibnamefont
  {Bromley}}, \bibinfo {author} {\bibfnamefont {Marco}\ \bibnamefont
  {Cianciaruso}}, \bibinfo {author} {\bibfnamefont {Ludovico}\ \bibnamefont
  {Lami}}, \bibinfo {author} {\bibfnamefont {W.~Y.~Sarah}\ \bibnamefont {Lau}},
  \bibinfo {author} {\bibfnamefont {Gerardo}\ \bibnamefont {Adesso}}, \bibinfo
  {author} {\bibfnamefont {Andrew~G.}\ \bibnamefont {White}}, \bibinfo {author}
  {\bibfnamefont {Alessandro}\ \bibnamefont {Fedrizzi}}, \ and\ \bibinfo
  {author} {\bibfnamefont {Marco}\ \bibnamefont {Piani}},\ }\bibfield  {title}
  {\enquote {\bibinfo {title} {Certification and {{Quantification}} of
  {{Multilevel Quantum Coherence}}},}\ }\href {\doibase
  10.1103/PhysRevX.8.041007} {\bibfield  {journal} {\bibinfo  {journal} {Phys.
  Rev. X}\ }\textbf {\bibinfo {volume} {8}},\ \bibinfo {pages} {041007}
  (\bibinfo {year} {2018})}\BibitemShut {NoStop}%
\bibitem [{\citenamefont {Johnston}\ \emph {et~al.}(2018)\citenamefont
  {Johnston}, \citenamefont {Li}, \citenamefont {Plosker}, \citenamefont
  {Poon},\ and\ \citenamefont {Regula}}]{johnston_2018}%
  \BibitemOpen
  \bibfield  {author} {\bibinfo {author} {\bibfnamefont {Nathaniel}\
  \bibnamefont {Johnston}}, \bibinfo {author} {\bibfnamefont {Chi-Kwong}\
  \bibnamefont {Li}}, \bibinfo {author} {\bibfnamefont {Sarah}\ \bibnamefont
  {Plosker}}, \bibinfo {author} {\bibfnamefont {Yiu-Tung}\ \bibnamefont
  {Poon}}, \ and\ \bibinfo {author} {\bibfnamefont {Bartosz}\ \bibnamefont
  {Regula}},\ }\bibfield  {title} {\enquote {\bibinfo {title} {Evaluating the
  robustness of {$k$}-coherence and {$k$}-entanglement},}\ }\href {\doibase
  10.1103/PhysRevA.98.022328} {\bibfield  {journal} {\bibinfo  {journal} {Phys.
  Rev. A}\ }\textbf {\bibinfo {volume} {98}},\ \bibinfo {pages} {022328}
  (\bibinfo {year} {2018})}\BibitemShut {NoStop}%
\bibitem [{\citenamefont {Streltsov}\ \emph {et~al.}(2018)\citenamefont
  {Streltsov}, \citenamefont {Kampermann}, \citenamefont {Wölk}, \citenamefont
  {Gessner},\ and\ \citenamefont {Bru{\ss}}}]{Streltsov_2018}%
  \BibitemOpen
  \bibfield  {author} {\bibinfo {author} {\bibfnamefont {Alexander}\
  \bibnamefont {Streltsov}}, \bibinfo {author} {\bibfnamefont {Hermann}\
  \bibnamefont {Kampermann}}, \bibinfo {author} {\bibfnamefont {Sabine}\
  \bibnamefont {Wölk}}, \bibinfo {author} {\bibfnamefont {Manuel}\
  \bibnamefont {Gessner}}, \ and\ \bibinfo {author} {\bibfnamefont {Dagmar}\
  \bibnamefont {Bru{\ss}}},\ }\bibfield  {title} {\enquote {\bibinfo {title}
  {Maximal coherence and the resource theory of purity},}\ }\href {\doibase
  10.1088/1367-2630/aac484} {\bibfield  {journal} {\bibinfo  {journal} {New
  Journal of Physics}\ }\textbf {\bibinfo {volume} {20}},\ \bibinfo {pages}
  {053058} (\bibinfo {year} {2018})}\BibitemShut {NoStop}%
\bibitem [{\citenamefont {{Bravyi}}\ \emph {et~al.}(2018)\citenamefont
  {{Bravyi}}, \citenamefont {{Browne}}, \citenamefont {{Calpin}}, \citenamefont
  {{Campbell}}, \citenamefont {{Gosset}},\ and\ \citenamefont
  {{Howard}}}]{bravyi2018simulation}%
  \BibitemOpen
  \bibfield  {author} {\bibinfo {author} {\bibfnamefont {S.}~\bibnamefont
  {{Bravyi}}}, \bibinfo {author} {\bibfnamefont {D.}~\bibnamefont {{Browne}}},
  \bibinfo {author} {\bibfnamefont {P.}~\bibnamefont {{Calpin}}}, \bibinfo
  {author} {\bibfnamefont {E.}~\bibnamefont {{Campbell}}}, \bibinfo {author}
  {\bibfnamefont {D.}~\bibnamefont {{Gosset}}}, \ and\ \bibinfo {author}
  {\bibfnamefont {M.}~\bibnamefont {{Howard}}},\ }\bibfield  {title} {\enquote
  {\bibinfo {title} {{Simulation of quantum circuits by low-rank stabilizer
  decompositions}},}\ }\href@noop {} {\bibfield  {journal} {\bibinfo  {journal}
  {arXiv e-prints}\ } (\bibinfo {year} {2018})},\ \Eprint
  {http://arxiv.org/abs/1808.00128} {arXiv:1808.00128 [quant-ph]} \BibitemShut
  {NoStop}%
\bibitem [{Note3()}]{Note3}%
  \BibitemOpen
  \bibinfo {note} {Logarithm of the stabilizer extent introduced in \cite
  {bravyi2018simulation} is equivalent to the max-relative entropy of magic
  \cite {regula}.}\BibitemShut {Stop}%
\bibitem [{\citenamefont {Heinrich}\ and\ \citenamefont
  {Gross}(2019)}]{Heinrich2019robustnessofmagic}%
  \BibitemOpen
  \bibfield  {author} {\bibinfo {author} {\bibfnamefont {Markus}\ \bibnamefont
  {Heinrich}}\ and\ \bibinfo {author} {\bibfnamefont {David}\ \bibnamefont
  {Gross}},\ }\bibfield  {title} {\enquote {\bibinfo {title} {Robustness of
  {M}agic and {S}ymmetries of the {S}tabiliser {P}olytope},}\ }\href {\doibase
  10.22331/q-2019-04-08-132} {\bibfield  {journal} {\bibinfo  {journal}
  {{Quantum}}\ }\textbf {\bibinfo {volume} {3}},\ \bibinfo {pages} {132}
  (\bibinfo {year} {2019})}\BibitemShut {NoStop}%
\bibitem [{\citenamefont {Girard}(2017)}]{girard_2017}%
  \BibitemOpen
  \bibfield  {author} {\bibinfo {author} {\bibfnamefont {Mark~W.}\ \bibnamefont
  {Girard}},\ }\emph {\bibinfo {title} {Convex Analysis in Quantum
  Information}},\ \href
  {https://prism.ucalgary.ca/bitstream/handle/11023/4001/ucalgary_2017_girard_mark.pdf}
  {Ph.D. thesis},\ \bibinfo  {school} {University of Calgary} (\bibinfo {year}
  {2017})\BibitemShut {NoStop}%
\bibitem [{\citenamefont {Gour}\ \emph {et~al.}(2018)\citenamefont {Gour},
  \citenamefont {Jennings}, \citenamefont {Buscemi}, \citenamefont {Duan},\
  and\ \citenamefont {Marvian}}]{gour_2018_majorization}%
  \BibitemOpen
  \bibfield  {author} {\bibinfo {author} {\bibfnamefont {Gilad}\ \bibnamefont
  {Gour}}, \bibinfo {author} {\bibfnamefont {David}\ \bibnamefont {Jennings}},
  \bibinfo {author} {\bibfnamefont {Francesco}\ \bibnamefont {Buscemi}},
  \bibinfo {author} {\bibfnamefont {Runyao}\ \bibnamefont {Duan}}, \ and\
  \bibinfo {author} {\bibfnamefont {Iman}\ \bibnamefont {Marvian}},\ }\bibfield
   {title} {\enquote {\bibinfo {title} {Quantum majorization and a complete set
  of entropic conditions for quantum thermodynamics},}\ }\href {\doibase
  10.1038/s41467-018-06261-7} {\bibfield  {journal} {\bibinfo  {journal} {Nat.
  Commun.}\ }\textbf {\bibinfo {volume} {9}},\ \bibinfo {pages} {5352}
  (\bibinfo {year} {2018})}\BibitemShut {NoStop}%
\bibitem [{\citenamefont {Rosset}\ \emph {et~al.}(2018)\citenamefont {Rosset},
  \citenamefont {Buscemi},\ and\ \citenamefont {Liang}}]{rosset_2018}%
  \BibitemOpen
  \bibfield  {author} {\bibinfo {author} {\bibfnamefont {Denis}\ \bibnamefont
  {Rosset}}, \bibinfo {author} {\bibfnamefont {Francesco}\ \bibnamefont
  {Buscemi}}, \ and\ \bibinfo {author} {\bibfnamefont {Yeong-Cherng}\
  \bibnamefont {Liang}},\ }\bibfield  {title} {\enquote {\bibinfo {title}
  {Resource {{Theory}} of {{Quantum Memories}} and {{Their Faithful
  Verification}} with {{Minimal Assumptions}}},}\ }\href {\doibase
  10.1103/PhysRevX.8.021033} {\bibfield  {journal} {\bibinfo  {journal} {Phys.
  Rev. X}\ }\textbf {\bibinfo {volume} {8}},\ \bibinfo {pages} {021033}
  (\bibinfo {year} {2018})}\BibitemShut {NoStop}%
\bibitem [{\citenamefont {{Krishnan Vijayan}}\ \emph
  {et~al.}(2019)\citenamefont {{Krishnan Vijayan}}, \citenamefont
  {{Chitambar}},\ and\ \citenamefont {{Hsieh}}}]{Vijayan19}%
  \BibitemOpen
  \bibfield  {author} {\bibinfo {author} {\bibfnamefont {Madhav}\ \bibnamefont
  {{Krishnan Vijayan}}}, \bibinfo {author} {\bibfnamefont {Eric}\ \bibnamefont
  {{Chitambar}}}, \ and\ \bibinfo {author} {\bibfnamefont {Min-Hsiu}\
  \bibnamefont {{Hsieh}}},\ }\bibfield  {title} {\enquote {\bibinfo {title}
  {{One-shot Distillation in a General Resource Theory}},}\ }\href@noop {}
  {\bibfield  {journal} {\bibinfo  {journal} {arXiv e-prints}\ ,\ \bibinfo
  {eid} {arXiv:1906.04959}} (\bibinfo {year} {2019})},\ \Eprint
  {http://arxiv.org/abs/1906.04959} {arXiv:1906.04959 [quant-ph]} \BibitemShut
  {NoStop}%
\bibitem [{\citenamefont {Datta}\ and\ \citenamefont
  {Leditzky}(2014)}]{Datta_2014}%
  \BibitemOpen
  \bibfield  {author} {\bibinfo {author} {\bibfnamefont {Nilanjana}\
  \bibnamefont {Datta}}\ and\ \bibinfo {author} {\bibfnamefont {Felix}\
  \bibnamefont {Leditzky}},\ }\bibfield  {title} {\enquote {\bibinfo {title} {A
  limit of the quantum r{\'{e}}nyi divergence},}\ }\href {\doibase
  10.1088/1751-8113/47/4/045304} {\bibfield  {journal} {\bibinfo  {journal}
  {Journal of Physics A: Mathematical and Theoretical}\ }\textbf {\bibinfo
  {volume} {47}},\ \bibinfo {pages} {045304} (\bibinfo {year}
  {2014})}\BibitemShut {NoStop}%
\bibitem [{\citenamefont {Wilde}\ \emph {et~al.}(2014)\citenamefont {Wilde},
  \citenamefont {Winter},\ and\ \citenamefont {Yang}}]{Wilde2014}%
  \BibitemOpen
  \bibfield  {author} {\bibinfo {author} {\bibfnamefont {Mark~M.}\ \bibnamefont
  {Wilde}}, \bibinfo {author} {\bibfnamefont {Andreas}\ \bibnamefont {Winter}},
  \ and\ \bibinfo {author} {\bibfnamefont {Dong}\ \bibnamefont {Yang}},\
  }\bibfield  {title} {\enquote {\bibinfo {title} {Strong converse for the
  classical capacity of entanglement-breaking and hadamard channels via a
  sandwiched r{\'e}nyi relative entropy},}\ }\href {\doibase
  10.1007/s00220-014-2122-x} {\bibfield  {journal} {\bibinfo  {journal}
  {Communications in Mathematical Physics}\ }\textbf {\bibinfo {volume}
  {331}},\ \bibinfo {pages} {593--622} (\bibinfo {year} {2014})}\BibitemShut
  {NoStop}%
\bibitem [{\citenamefont {Gour}\ \emph {et~al.}(2009)\citenamefont {Gour},
  \citenamefont {Marvian},\ and\ \citenamefont
  {Spekkens}}]{PhysRevA.80.012307}%
  \BibitemOpen
  \bibfield  {author} {\bibinfo {author} {\bibfnamefont {Gilad}\ \bibnamefont
  {Gour}}, \bibinfo {author} {\bibfnamefont {Iman}\ \bibnamefont {Marvian}}, \
  and\ \bibinfo {author} {\bibfnamefont {Robert~W.}\ \bibnamefont {Spekkens}},\
  }\bibfield  {title} {\enquote {\bibinfo {title} {Measuring the quality of a
  quantum reference frame: The relative entropy of frameness},}\ }\href
  {\doibase 10.1103/PhysRevA.80.012307} {\bibfield  {journal} {\bibinfo
  {journal} {Phys. Rev. A}\ }\textbf {\bibinfo {volume} {80}},\ \bibinfo
  {pages} {012307} (\bibinfo {year} {2009})}\BibitemShut {NoStop}%
\bibitem [{\citenamefont {Marian}\ and\ \citenamefont {Marian}(2013)}]{marian}%
  \BibitemOpen
  \bibfield  {author} {\bibinfo {author} {\bibfnamefont {Paulina}\ \bibnamefont
  {Marian}}\ and\ \bibinfo {author} {\bibfnamefont {Tudor~A.}\ \bibnamefont
  {Marian}},\ }\bibfield  {title} {\enquote {\bibinfo {title} {Relative entropy
  is an exact measure of non-{G}aussianity},}\ }\href {\doibase
  10.1103/PhysRevA.88.012322} {\bibfield  {journal} {\bibinfo  {journal} {Phys.
  Rev. A}\ }\textbf {\bibinfo {volume} {88}},\ \bibinfo {pages} {012322}
  (\bibinfo {year} {2013})}\BibitemShut {NoStop}%
\bibitem [{\citenamefont {{Lloyd}}\ \emph {et~al.}(2015)\citenamefont
  {{Lloyd}}, \citenamefont {{Liu}}, \citenamefont {{Pirandola}}, \citenamefont
  {{Chiloyan}}, \citenamefont {{Hu}}, \citenamefont {{Huberman}},\ and\
  \citenamefont {{Chen}}}]{heat}%
  \BibitemOpen
  \bibfield  {author} {\bibinfo {author} {\bibfnamefont {Seth}\ \bibnamefont
  {{Lloyd}}}, \bibinfo {author} {\bibfnamefont {Zi-Wen}\ \bibnamefont {{Liu}}},
  \bibinfo {author} {\bibfnamefont {Stefano}\ \bibnamefont {{Pirandola}}},
  \bibinfo {author} {\bibfnamefont {Vazrik}\ \bibnamefont {{Chiloyan}}},
  \bibinfo {author} {\bibfnamefont {Yongjie}\ \bibnamefont {{Hu}}}, \bibinfo
  {author} {\bibfnamefont {Samuel}\ \bibnamefont {{Huberman}}}, \ and\ \bibinfo
  {author} {\bibfnamefont {Gang}\ \bibnamefont {{Chen}}},\ }\bibfield  {title}
  {\enquote {\bibinfo {title} {{No energy transport without discord}},}\
  }\href@noop {} {\bibfield  {journal} {\bibinfo  {journal} {arXiv e-prints}\
  ,\ \bibinfo {eid} {arXiv:1510.05035}} (\bibinfo {year} {2015})},\ \Eprint
  {http://arxiv.org/abs/1510.05035} {arXiv:1510.05035 [quant-ph]} \BibitemShut
  {NoStop}%
\bibitem [{\citenamefont {Wilde}(2013)}]{Wilde:2013:QIT:2505455}%
  \BibitemOpen
  \bibfield  {author} {\bibinfo {author} {\bibfnamefont {Mark~M.}\ \bibnamefont
  {Wilde}},\ }\href@noop {} {\emph {\bibinfo {title} {Quantum Information
  Theory}}},\ \bibinfo {edition} {1st}\ ed.\ (\bibinfo  {publisher} {Cambridge
  University Press},\ \bibinfo {address} {New York, NY, USA},\ \bibinfo {year}
  {2013})\BibitemShut {NoStop}%
\bibitem [{\citenamefont {Watrous}(2016)}]{watrous}%
  \BibitemOpen
  \bibfield  {author} {\bibinfo {author} {\bibfnamefont {John}\ \bibnamefont
  {Watrous}},\ }\href@noop {} {\emph {\bibinfo {title} {Theory of Quantum
  Information}}} (\bibinfo {year} {2016}),\ \bibinfo {note}
  {https://cs.uwaterloo.ca/~watrous/TQI/TQI.pdf}\BibitemShut {NoStop}%
\bibitem [{\citenamefont {Brand\~ao}(2005)}]{brandao05}%
  \BibitemOpen
  \bibfield  {author} {\bibinfo {author} {\bibfnamefont {Fernando G. S.~L.}\
  \bibnamefont {Brand\~ao}},\ }\bibfield  {title} {\enquote {\bibinfo {title}
  {Quantifying entanglement with witness operators},}\ }\href {\doibase
  10.1103/PhysRevA.72.022310} {\bibfield  {journal} {\bibinfo  {journal} {Phys.
  Rev. A}\ }\textbf {\bibinfo {volume} {72}},\ \bibinfo {pages} {022310}
  (\bibinfo {year} {2005})}\BibitemShut {NoStop}%
\bibitem [{\citenamefont {Bu}\ \emph {et~al.}(2017)\citenamefont {Bu},
  \citenamefont {Singh}, \citenamefont {Fei}, \citenamefont {Pati},\ and\
  \citenamefont {Wu}}]{PhysRevLett.119.150405}%
  \BibitemOpen
  \bibfield  {author} {\bibinfo {author} {\bibfnamefont {Kaifeng}\ \bibnamefont
  {Bu}}, \bibinfo {author} {\bibfnamefont {Uttam}\ \bibnamefont {Singh}},
  \bibinfo {author} {\bibfnamefont {Shao-Ming}\ \bibnamefont {Fei}}, \bibinfo
  {author} {\bibfnamefont {Arun~Kumar}\ \bibnamefont {Pati}}, \ and\ \bibinfo
  {author} {\bibfnamefont {Junde}\ \bibnamefont {Wu}},\ }\bibfield  {title}
  {\enquote {\bibinfo {title} {Maximum relative entropy of coherence: An
  operational coherence measure},}\ }\href {\doibase
  10.1103/PhysRevLett.119.150405} {\bibfield  {journal} {\bibinfo  {journal}
  {Phys. Rev. Lett.}\ }\textbf {\bibinfo {volume} {119}},\ \bibinfo {pages}
  {150405} (\bibinfo {year} {2017})}\BibitemShut {NoStop}%
\bibitem [{\citenamefont {Nielsen}(1999)}]{nielsen}%
  \BibitemOpen
  \bibfield  {author} {\bibinfo {author} {\bibfnamefont {M.~A.}\ \bibnamefont
  {Nielsen}},\ }\bibfield  {title} {\enquote {\bibinfo {title} {Conditions for
  a class of entanglement transformations},}\ }\href {\doibase
  10.1103/PhysRevLett.83.436} {\bibfield  {journal} {\bibinfo  {journal} {Phys.
  Rev. Lett.}\ }\textbf {\bibinfo {volume} {83}},\ \bibinfo {pages} {436--439}
  (\bibinfo {year} {1999})}\BibitemShut {NoStop}%
\bibitem [{\citenamefont {Horodecki}\ \emph
  {et~al.}(2003{\natexlab{b}})\citenamefont {Horodecki}, \citenamefont
  {Horodecki},\ and\ \citenamefont {Oppenheim}}]{hho}%
  \BibitemOpen
  \bibfield  {author} {\bibinfo {author} {\bibfnamefont {Micha\l{}}\
  \bibnamefont {Horodecki}}, \bibinfo {author} {\bibfnamefont {Pawe\l{}}\
  \bibnamefont {Horodecki}}, \ and\ \bibinfo {author} {\bibfnamefont
  {Jonathan}\ \bibnamefont {Oppenheim}},\ }\bibfield  {title} {\enquote
  {\bibinfo {title} {Reversible transformations from pure to mixed states and
  the unique measure of information},}\ }\href {\doibase
  10.1103/PhysRevA.67.062104} {\bibfield  {journal} {\bibinfo  {journal} {Phys.
  Rev. A}\ }\textbf {\bibinfo {volume} {67}},\ \bibinfo {pages} {062104}
  (\bibinfo {year} {2003}{\natexlab{b}})}\BibitemShut {NoStop}%
\bibitem [{\citenamefont {Horodecki}\ and\ \citenamefont
  {Horodecki}(1999)}]{PhysRevA.59.4206}%
  \BibitemOpen
  \bibfield  {author} {\bibinfo {author} {\bibfnamefont {Micha\l{}}\
  \bibnamefont {Horodecki}}\ and\ \bibinfo {author} {\bibfnamefont {Pawe\l{}}\
  \bibnamefont {Horodecki}},\ }\bibfield  {title} {\enquote {\bibinfo {title}
  {Reduction criterion of separability and limits for a class of distillation
  protocols},}\ }\href {\doibase 10.1103/PhysRevA.59.4206} {\bibfield
  {journal} {\bibinfo  {journal} {Phys. Rev. A}\ }\textbf {\bibinfo {volume}
  {59}},\ \bibinfo {pages} {4206--4216} (\bibinfo {year} {1999})}\BibitemShut
  {NoStop}%
\bibitem [{\citenamefont {Fuchs}\ and\ \citenamefont {{van de
  Graaf}}(1999)}]{Fuchs-vandeGraaf}%
  \BibitemOpen
  \bibfield  {author} {\bibinfo {author} {\bibfnamefont {Chris~A.}\
  \bibnamefont {Fuchs}}\ and\ \bibinfo {author} {\bibfnamefont {Jeroen}\
  \bibnamefont {{van de Graaf}}},\ }\bibfield  {title} {\enquote {\bibinfo
  {title} {Cryptographic distinguishability measures for quantum-mechanical
  states},}\ }\href {\doibase 10.1109/18.761271} {\bibfield  {journal}
  {\bibinfo  {journal} {IEEE Transactions on Information Theory}\ }\textbf
  {\bibinfo {volume} {45}},\ \bibinfo {pages} {1216--1227} (\bibinfo {year}
  {1999})}\BibitemShut {NoStop}%
\bibitem [{\citenamefont {Dupuis}\ \emph {et~al.}(2013)\citenamefont {Dupuis},
  \citenamefont {Kr\"amer}, \citenamefont {Faist}, \citenamefont {Renes},\ and\
  \citenamefont {Renner}}]{dupuis}%
  \BibitemOpen
  \bibfield  {author} {\bibinfo {author} {\bibfnamefont {F.}~\bibnamefont
  {Dupuis}}, \bibinfo {author} {\bibfnamefont {L.}~\bibnamefont {Kr\"amer}},
  \bibinfo {author} {\bibfnamefont {P.}~\bibnamefont {Faist}}, \bibinfo
  {author} {\bibfnamefont {J.~M.}\ \bibnamefont {Renes}}, \ and\ \bibinfo
  {author} {\bibfnamefont {R.}~\bibnamefont {Renner}},\ }\enquote {\bibinfo
  {title} {Generalized entropies},}\ in\ \href {\doibase
  10.1142/9789814449243_0008} {\emph {\bibinfo {booktitle} {XVIIth
  International Congress on Mathematical Physics}}}\ (\bibinfo  {publisher}
  {World Scientific},\ \bibinfo {year} {2013})\ pp.\ \bibinfo {pages}
  {134--153}\BibitemShut {NoStop}%
\bibitem [{\citenamefont {Streltsov}\ \emph
  {et~al.}(2017{\natexlab{b}})\citenamefont {Streltsov}, \citenamefont
  {Adesso},\ and\ \citenamefont {Plenio}}]{coh_rev}%
  \BibitemOpen
  \bibfield  {author} {\bibinfo {author} {\bibfnamefont {Alexander}\
  \bibnamefont {Streltsov}}, \bibinfo {author} {\bibfnamefont {Gerardo}\
  \bibnamefont {Adesso}}, \ and\ \bibinfo {author} {\bibfnamefont {Martin~B.}\
  \bibnamefont {Plenio}},\ }\bibfield  {title} {\enquote {\bibinfo {title}
  {Colloquium: Quantum coherence as a resource},}\ }\href {\doibase
  10.1103/RevModPhys.89.041003} {\bibfield  {journal} {\bibinfo  {journal}
  {Rev. Mod. Phys.}\ }\textbf {\bibinfo {volume} {89}},\ \bibinfo {pages}
  {041003} (\bibinfo {year} {2017}{\natexlab{b}})}\BibitemShut {NoStop}%
\bibitem [{Note4()}]{Note4}%
  \BibitemOpen
  \bibinfo {note} {Note that the noise parameter $p$ here does not necessarily
  represent a probability: in the following we may consider $p>1$, where the
  depolarizing map induces a \protect \emph {pseudomixture} of original state
  and the maximally mixed state, for mathematical purposes}\BibitemShut
  {NoStop}%
\end{thebibliography}%

\appendix
\widetext



\section{Max- and min-relative entropies as quantum R\'enyi relative entropies}\label{app:renyi}

Here we elaborate on the statement that the max- and min-relative entropies, whose close relatives emerge from one-shot formation and distillation tasks respectively, correspond to two extremes of the quantum R\'enyi relative entropy. 
(We refer interested readers to \cite{Datta_2014,qrenyi} for more comprehensive discussions on the limits of quantum R\'enyi entropies.)
To this end, we first formally define two types of quantum R\'enyi relative entropies, which serve as the parent quantities of many important types of divergences between quantum states:
\begin{defn}
Let $\rho$ be a density operator and $\sigma$ be a positive semidefinite operator.  For $\alpha\in(0,1)\cup(1,\infty)$, the \emph{quantum R\'enyi-$\alpha$ relative entropy} between $\rho$ and $\sigma$ is given by
\begin{equation}
    D_\alpha(\rho\|\sigma) := \frac{1}{1-\alpha}\log\Tr\{\rho^\alpha\sigma^{1-\alpha}\},
\end{equation}
while the \emph{sandwiched quantum R\'enyi-$\alpha$ relative entropy} (also frequently referred to as the \emph{quantum R\'enyi-$\alpha$ divergence}) between $\rho$ and $\sigma$ is given by \cite{Wilde2014,qrenyi}
\begin{equation}
    \widetilde{D}_\alpha(\rho\|\sigma) := \frac{1}{1-\alpha}\log\Tr\{(\sigma^{\frac{1-\alpha}{2\alpha}}\rho\sigma^{\frac{1-\alpha}{2\alpha}})^\alpha\},
\end{equation}
which reduces to ${D}_\alpha(\rho\|\sigma)$ when $[\rho,\sigma]=0$.  The values are set to $\infty$ when they are not well-defined.
\end{defn}
We state the following facts:
\begin{itemize}
    \item The max-relative entropy $\dmax(\rho\|\sigma) = \log \min \{\lambda: \rho\leq \lambda \sigma\}$ is equivalent to $\widetilde{D}_\alpha(\rho\|\sigma)$ in the limit $\alpha\rightarrow\infty$;
    \item The min-relative entropy $\dmin(\rho\|\sigma) = -\log\Tr\{\Pi_\rho\sigma\}$ is equivalent to ${D}_0(\rho\|\sigma)$ since $\rho^0 = \Pi_\rho$ (hence it is often called the ``0-relative entropy'').
\end{itemize}

Note that the term ``min-relative entropy'' in some literature may refer to a different quantity which is equivalent to $\widetilde{D}_{1/2}(\rho\|\sigma)$ and also directly related to fidelity \cite{dupuis}: $\dmin'(\rho\|\sigma):= -2\log\norm{\sqrt{\rho}\sqrt{\sigma}}_1 = -\log f(\rho,\sigma)$.

Also note that the quantum R\'enyi relative entropies are monotonically non-decreasing in the order $\alpha$.  Therefore, loosely speaking, the optimality bounds for formation and distillation, in terms of (variants of) the max- and min-relative entropy measures respectively, are normally in their strongest forms.


\section{Non-affine theory without free robustness}\label{app:fr_affine}

It is clear from the definition of the free robustness that if the set of free states $\mF$ is an affine set ($\sigma_1, \sigma_2\in \mF\Rightarrow t\sigma_1 + (1-t)\sigma_2\in\mF$ for any $t\in \mathbb{R}$ such that $t\sigma_1 + (1-t)\sigma_2\geq 0$), then any resource state $\rho\notin\mF$ does not have finite free robustness. 
Here, we show that the non-affinity of $\mF$ does not necessarily imply that all states $\rho\in\mD(\mH)$ possess finite free robustness, which, to our knowledge, had not been explicitly presented in the literature.

For resource theories such that the set of free states is the convex hull of the set of pure free states, the condition for all states to have finite free robustness was shown in Ref.~\cite{regula}: for $\mF={\rm conv}(\mathcal{S}_+)$ where $\mathcal{S}_+=\big\{\dm{\phi}\,\big|\,\ket{\phi}\in\mathcal{V}(\mH),\ \mathcal{V}(\mH)\mbox{ is the set of pure free states}\big\}$, all states $\rho\in\mD(\mH)$ have finite free robustness iff $\mD(\mH)\subseteq {\rm span}(\mathcal{S}_+)$.
A corresponding result for general convex theories can be obtained in a simple manner, where we find that all states $\rho\in\mD(\mH)$ possess finite free robustness iff $\mD(\mH)\subseteq {\rm span}(\mF)$. 
To show this, first observe that if all states possess finite free robustness, then there exists a decomposition $\rho=(1+s)\tau-s\sigma$ where $\tau,\sigma\in\mF$ for all $\rho\in\mD(\mH)$, and thus $\mD(\mH)\subseteq {\rm span}(\mF)$.
On the other hand, if $\mD(\mH)\subseteq {\rm span}(\mF)$, then any state $\rho\in\mD(\mH)$ can be written as a linear combination of free states:
\ba
\rho=\sum_j c_j\sigma_j=\sum_{j\in \mathcal{I}^+} c_j\sigma_j - \sum_{j'\in\mathcal{I}^-} |c_{j'}|\sigma_{j'},
\ea
where $\sigma_j\in\mF$, $c_j\in\mathbb{R}\setminus\{0\}$ for all $j$, and $\mathcal{I}^{\pm}$ denotes the set of indices for which $c_j$ are positive and negative. 
Due to the convexity of $\mF$, we have that $\sum_{j\in \mathcal{I}^{\pm}} |c_j|\sigma_j\in {\rm cone}(\mF)$, and due to the fact that $\Tr\rho=1$, one can write $\sum_{j\in \mathcal{I}^+} c_j\sigma_j=(1+s)\,\tau$ and $\sum_{j\in \mathcal{I}^-} |c_j|\sigma_j=s\,\sigma$ for some $s\geq0$, $\tau,\sigma\in\mF$. 

We now give an example of non-affine theories for which there exist states without finite free robustness. Consider the resource theory defined on qubit systems such that the set of free states is the convex hull of $\{\dm{0},\dm{1},\dm{+},\dm{-}\}$.
It is not affine because any state on the $X$-$Z$ plane of the Bloch sphere can be constructed by an affine combination of these four states.
However, clearly $\mD(\mH)\not\subseteq {\rm span}(\mF)$ because the span of the free states can only involve real coefficients with respect to the computational basis.
Hence, there are states (more explicitly, the states with imaginary coefficients) which do not have finite free robustness.

\section{Exact resource destroying maps}\label{app:erdm}

We identify a special kind of resource destroying map, which ``picks out'' the closest free state with respect to the standard relative entropy of resource:
\begin{defn}[Exact resource destroying map]
A resource destroying map $\lambda$ is called an \emph{exact resource destroying map} if it maps all states $\rho$ to a closest free state 
with respect to the relative entropy, i.e.
\begin{equation}
    D(\rho\|\lambda(\rho)) = \min_{\sigma\in\mF}D(\rho\|\sigma),
\end{equation}
where $D(\rho\|\sigma):=\Tr\{\rho\log\rho-\rho\log\sigma\}$ denotes the relative entropy between $\rho$ and $\sigma$.
\end{defn}

All resource theories have an exact resource destroying map by definition.
However, of particular interest are the cases where the exact resource destroying map is ``simple''---meaning that, for example, it is a quantum channel, or can be characterized by some other natural principles.  There are several good examples (in addition to the trivial example of quantum thermodynamics, where the unique resource destroying map is the constant channel that outputs the Gibbs state for the given Hamiltonian and temperature):
\begin{table}[h]
\caption{\label{tab:example}Notable resource theories with a natural exact resource destroying map.}
\begin{ruledtabular}
\begin{tabularx}{\textwidth}{lp{0.8\textwidth}}
Theory & Exact resource destroying map\\  \hline
Coherence & Complete dephasing channel: $\Delta(\rho) = \sum_i\ketbra{i}{i}\rho\ketbra{i}{i}$ where $\{\ketbra{i}{i}\}$ is the reference basis \cite{coh_rev}.\\
Asymmetry & Uniform twirling channel (over symmetry group $G$): $\mathcal{G}(\rho) = \int_{G} d\mu(U) U\rho U^\dagger$ where the integral is taken over the normalized Haar measure on $G$  (for finite groups, the integral is replaced by the average over the action of group elements) \cite{PhysRevA.80.012307}. \\
Non-Gaussianity\footnote{Note that the Hilbert space of continuous-variable systems is infinite-dimensional, so the discussions in this paper do not directly apply to the theory of non-Gaussianity.  But it still serves as an interesting example where a ``simple'' resource destroying map is exact.} & 
$\lambda(\rho) = \tau_G$ where $\tau_G$ is the Gaussian state with the same mean displacement and covariance matrix as $\rho$ \cite{marian}.
\end{tabularx}
\end{ruledtabular}
\end{table}

The closest free states are not easily identified for many other resource theories.  Or in other words, the exact resource destroying map does not admit a simple description. 
The above theories exhibit the nice feature that the intuitive resource destroying map turns out to be exact. However, for various other theories (such as magic states), it is not immediately clear what a nontrivial resource destroying map is.
A particularly interesting case is quantum discord, where there is a very natural resource destroying map (namely, measurement in the local eigenbasis) \cite{rdm} that leads to physically meaningful and mathematically well-behaved discord measures \cite{dqd,heat}.  Unfortunately, this resource destroying map is not exact, and furthermore, the closest free states in the theory of discord do not seem to exhibit any generic feature (see Ref.~\cite{dqd} for more detailed discussions).

\section{Free robustness implies nonlinearity of resource destroying maps}\label{app:freer}
\begin{prop}
If free robustness always exists in the resource theory, that is, for any resource state $\rho\not\in \mF$ there exists some free states $\sigma,\tau\in \mF$ such that $\rho+s\sigma=(1+s)\tau$ for some $s\geq 0$, then the resource destroying maps of this theory must be nonlinear (and thus resource destroying channels do not exist).
\end{prop}
\begin{proof}
Suppose there exists some linear resource destroying map $\Lambda$. Then 
\begin{equation}
    \Lambda(\rho) = (1+s)\Lambda(\tau) -s\Lambda(\sigma) = (1+s)\tau -s\sigma,
\end{equation}
where the second equality follows from the fact that resource destroying maps are nonresource fixing.   Therefore, we have $\rho = \Lambda(\rho) \in \mF$, which contradicts the assumption that $\rho$ is a resource state.
\end{proof}
This result, which had not been explicitly pointed out in the past, is essentially a combination of the facts that the set of free states must be non-affine for any state to possess finite free robustness (as mentioned in the main text) and that the set of free states must be affine for a resource destroying channel to exist as shown in Ref.~\cite{PhysRevA.95.062314}. 
This improves the basic observation in Ref.~\cite{rdm} that the nonconvexity of $\mF$ implies nonlinearity of resource destroying maps.  
Also note that the absence of finite free robustness for some states is a necessary but not sufficient condition for the existence of resource destroying channels.  The sufficient condition is also characterized in Ref.~\cite{PhysRevA.95.062314}.

\section{Collapse of the optimized modification coefficients (part of Theorem 1)} \label{app:collapse}

Here, we prove the first half of Theorem 1, collapse of the modification coefficients based on resource measures. We prove the other half of the Theorem in the next section. 

We first show that the minimum fidelity with respect to convex and closed sets of free states can always be realized at a pure state. 

\begin{lem}\label{lem:pure}
 Suppose $\mF\subset {\mD}(\mH_d)$ is convex and closed. Then, for any $d$, there exists a pure state $\hat\Phi_d\in{\mD}(\mH_d)$ that achieves the minimum of the free fidelity.

\end{lem}
\begin{proof}
First we show that the square root of the free fidelity $\sqrt{\mathfrak{f}}$ is concave.
Let $\rho=\sum_i p_i\rho_i,p_i\in[0,1]$ be some probabilistic decomposition of state $\rho$, and $\tilde{\delta}_i$ be the optimal free state that achieves $\mathfrak{f}(\rho_i) = \max_{\delta\in \mF}f(\rho_i, \delta)$. We have
\begin{equation}
    \sqrt{\mathfrak{f}(\rho)} = \max_{\delta\in \mF}\sqrt{f(\rho, \delta)} \geq \sqrt{f(\rho,\sum_i p_i\tilde{\delta}_i)}
    \geq \sum_i p_i \sqrt{f(\rho_i,\tilde{\delta}_i)} = \sum_i p_i \sqrt{\mathfrak{f}(\rho_i)},  \label{eq:concavity_f}
\end{equation}
where the first inequality follows from the fact that $\sum_i p_i \tilde{\delta}_i \in \mF$ due to the convexity of $\mF$, and the second inequality follows from the joint concavity of the square root of the Uhlmann fidelity $\sqrt{f}$ \cite{Wilde:2013:QIT:2505455,watrous}.
Suppose for some dimension, the minimum of $\mathfrak{f}(\rho)$ (equivalently, the minimum of $\sqrt{\mathfrak{f}(\rho)}$) is only achieved by mixed states.  Let $\sigma$ be such a mixed state which has a pure state decomposition $\sigma = \sum_i p_i\ketbra{\lambda_i}{\lambda_i}$.  Then since $\mathfrak{f}(\sigma)< \mathfrak{f}(\ketbra{\lambda_i}{\lambda_i})$ for all $i$, we have $\sqrt{\mathfrak{f}(\sigma)} < \sum_i p_i \sqrt{\mathfrak{f}(\ketbra{\lambda_i}{\lambda_i})}$,
which contradicts the concavity of $\sqrt{\mathfrak{f}}$, Eq.~(\ref{eq:concavity_f}).   In conclusion, the minimum-free-fidelity state can be pure for all dimensions.
\end{proof}
The similar proof easily applies to other max-resource states given by convex resource monotones.
We are now in a position to prove the first half of the Theorem. 
\begin{prop}\label{thm:collapse}
Suppose the resource theory satisfies Condition (CH). Then, for any $d$, there exists a pure state $\hat\Phi_d\in\mD(\mH_d)$ such that $m_f(\hat\Phi_d) = m_{\min}(\hat\Phi_d) = m_{\max}(\hat\Phi_d) := g_d$ where $\hat\Phi_d$ achieves the maxima of $m_f,m_{\min},m_{\max}$.
\end{prop}
\begin{proof}
Due to Lemma \ref{lem:pure}, there always exists a pure state $\hat\Phi_d$ which achieves the maximum of $m_f$. 
It is also straightforward to see that $m_f(\hat\Phi_d) = m_{\min}(\hat\Phi_d)$ because they coincide at any pure state. 
We now prove that $m_f(\hat\Phi_d) = m_{\max}(\hat\Phi_d)$.   Recall that $2^{\rdmax(\rho)}$, which is equivalent to $1+R_G(\rho)$, can be computed by the following convex optimization problem \cite{brandao05,regula}:
\begin{eqnarray}
\text{maximize}\quad && \Tr\{\rho X\}\\
\text{subject to}\quad && X\geq 0\\
                    &&  \Tr\{\delta X\}\leq 1, \forall\delta\in \mF  \label{tr_constraint}
\end{eqnarray}
It is shown in \cite{regula} that if $\mF$ satisfies Condition (CH), then the optimal witness operator $X$ for pure state $\rho$ always takes a rank-one form, i.e.~can be written as $X = c\ketbra{w}{w}$ for some $c\geq 0$ and pure state $\ket{w}$.
We shall show that when $\rho = \hat\Phi_d = \ketbra{\hat\Phi_d}{\hat\Phi_d}$, the optimal $X$ is identified by $c^{-1}=\mathfrak{f}(\hat\Phi_d) = \max_{\delta\in \mF}\Tr\{\hat\Phi_d\delta\}$ and $\ket{w}=\ket{\hat\Phi_d}$.
By (\ref{tr_constraint}), for all $\delta\in \mF$, $c\leq\bra{w}\delta\ket{w}^{-1}$.  
So the maximal $c$ as a function of $\ket{w}$ under this constraint is given by $\tilde{c}_{\ket{w}}=(\max_{\delta\in \mF}\bra{w}\delta\ket{w})^{-1} = \mathfrak{f}^{-1}(w)$.  Then, notice that $\max_{\ket{\psi}}\tilde{c}_{\ket{\psi}} = \tilde{c}_{\ket{\hat\Phi_d}}$ simply because $\mathfrak{f}$ achieves the minimum at $\hat\Phi_d$.  Therefore, the target function $\Tr\{\hat\Phi_d X\} = c|\bracket{\hat\Phi_d}{w}|^2$ achieves the maximum when $\ket{w} = \ket{\hat\Phi_d}$ since $c$ and $|\bracket{\hat\Phi_d}{w}|^2$ achieve the maximum simultaneously when $\ket{w} = \ket{\hat\Phi_d}$.   So we have $2^{\rdmax(\hat\Phi_d)} = \tilde{c}_{\ket{\hat\Phi_d}} = \mathfrak{f}^{-1}(\hat\Phi_d)$, which implies that $m_f(\hat\Phi_d) = m_{\max}(\hat\Phi_d)$.

Next, let us see that $\hat\Phi_d$ achieves the maximum of $m_{\min}$ and $m_{\max}$ for the same dimension $d$.
Since $m_f(\Phi_d)=m_{\min}(\Phi_d)$ for pure states $\Phi_d$, $\hat\Phi_d$ gives the maximum value of $m_{\min}$ among all pure states of dimension $d$. 
Recall that $D_{\rm min}(\rho\|\sigma)$ is jointly convex \cite{dattarel}. 
By a similar argument as we saw above, one can show that $\mathfrak{D}_{\rm min}(\rho)=\min_{\delta\in \mF}D_{\rm min}(\rho\|\delta)$ is convex.
Therefore, the maximum of $\mathfrak{D}_{\rm min}(\rho)$ can be achieved by some pure state. 
Since $\hat\Phi_d$ is the state that gives the maximum value among all pure states, it also maximizes $\mathfrak{D}_{\rm min}(\rho)$ among all states.  
To see that it also maximizes $m_{\max}$, recall that $R_G(\rho)$ is convex~\cite{regula}, and so the maximum can be realized at a pure state. 
Let $\ket{\psi}$ be such a pure state and $X=\tilde{c}_{\ket{w}}\dm{w}$ be an optimal witness realizing the generalized robustness. 
Then, $1+R_G(\dm{\psi})=\tilde{c}_{\ket{w}}|\bracket{\psi}{w}|^2\leq \tilde{c}_{\ket{\hat\Phi_d}}\cdot 1 =1+ R_G(\hat\Phi_d)$, which implies that $\hat\Phi_d$ gives the maximum $R_G$ among all pure states, and hence all states.
\end{proof}

We remark that the collapse can still happen for more theories which does not satisfy Condition (CH).
An important example of such a theory is quantum thermodynamics, where the only free state is the Gibbs state. 
One can nevertheless show that the corresponding result still holds for the theory of quantum thermodynamics.

\begin{prop}[Quantum thermodynamics]\label{thm:collapse_thermo}
Consider the theory with $\mF=\{\tau_d\}\subset {\mD}(\mH_d)$ where $\tau_d$ is the Gibbs state defined on the system of dimension $d$ with temperature $T$ and Hamiltonian $H=\sum_i E_i\dm{i}$.
Then for pure minimum-free-fidelity state $\hat\Phi_d$ (whose existence is guaranteed by Lemma \ref{lem:pure}), we have $m_f(\hat\Phi_d) = m_{\min}(\hat\Phi_d) = m_{\max}(\hat\Phi_d)$.
Moreover, $\hat\Phi_d$ also achieves the maximum of $m_{\min}$ and $m_{\max}$.
\end{prop}
\begin{proof}
 Let $\ket{i_{\max}}$ be an eigenstate corresponding to the maximum $E_i$.
 We first observe that $\hat{\Phi}_d=\dm{i_{\max}}$.
 Let $\ket{\psi}\in\mH_d$ be a pure state and consider the fidelity with the Gibbs state: $f(\dm{\psi},\tau_d)=\bra{\psi}\tau_d\ket{\psi}=\sum_i\tau_{d,i}\,|\bracket{\psi}{i}|^2$ where $\tau_{d,i}=\frac{\exp(-E_i/T)}{Z}$ and $Z=\sum_i\exp(-E_i/T)$.
 It clearly takes the minimum when $\ket{\psi}=\ket{i_{\max}}$ because $\tau_{d,i_{\max}}=\min_i \tau_{d,i}$. So we identify  $\hat{\Phi}_d=\dm{i_{\max}}$ and $m_f(\hat{\Phi}_d)=\frac{\log(1/\tau_{d,i_{\max}})}{\log d}$.
 
Recall that $m_f(\hat\Phi_d)=m_{\min}(\hat\Phi_d)$ since $\hat{\Phi}_d$ is pure.  We now show that $m_f(\hat\Phi_d)=m_{\max}(\hat\Phi_d)$. 
 Recalling the definition of max-relative entropy, we have
 \ba
  2^{\mathfrak{D}_{\max}(\hat\Phi_d)} = \min\{\lambda\geq 0\,|\,\lambda\,\tau_d - \hat\Phi_d\geq 0\}.
 \ea
 Since $\hat\Phi_d=\dm{i_{\max}}$, we obtain $2^{\mathfrak{D}_{\max}(\hat\Phi_d)}=1/\tau_{d,i_{\max}}$, and hence $m_{\max}(\hat\Phi_d)=m_f(\hat\Phi_d)$.
 
 Finally, we show that $\hat\Phi_d$ also achieves the maxima of $m_{\min}$ and $m_{\max}$.
 Since the maxima of $\mathfrak{D}_{\min,\max}$ occur at pure states due to their convexity, it suffices to show that $\hat\Phi_d$ achieves the maxima among all pure states. 
 The $m_{\min}$ case follows from the same argument as in the proof of Proposition \ref{thm:collapse}.
 To see that $\hat\Phi_d$ achieves the maximum of $m_{\max}$, the above optimization form of the generalized robustness is useful. 
 Let $\ket{\psi}\in \mH_d$ be a pure state, and suppose that $\ket{\psi}$ has larger generalized robustness (equivalently, the max-relative entropy measure) than $\hat\Phi_d$, i.e. $1+R_G(\dm{\psi})>1+R_G(\hat\Phi_d)=1/\tau_{d,i_{\max}}$.
 Let $X$ be a positive semidefinite operator satisfying Eq.~\eqref{tr_constraint} and $\Tr\{\dm{\psi} X\}=1+R_G(\dm{\psi})$. Writing $X$ in the spectral decomposition form $X=\sum_ix_i\dm{e_i}$ where $x_i\geq 0,\forall i$ and $\{\ket{e_i}\}$ is some orthonormal basis, we get $1+R_G(\dm{\psi})=\sum_j x_j |\bracket{\psi}{e_j}|^2$. 
 Let $p_j:= |\bracket{\psi}{e_j}|^2$. Then, $\{p_j\}$ is a probability distribution and thus $1+R_G(\dm{\psi})=\sum_j x_j p_j\leq x_{\max}$ where $x_{\max}:= \max_i x_i$. By the assumption that $1+R_G(\dm{\psi})>1/\tau_{d,i_{\max}}$, we have $x_{\max}>1/\tau_{d,i_{\max}}$. 
 However, by setting $x_J=x_{\max}$, we also get $\Tr\{\tau_d X\} = \sum_{ji}x_j\tau_{d,i}|\bracket{i}{e_j}|^2\geq\sum_i x_{\max}\tau_{d,i}|\bracket{i}{e_J}|^2>\sum_i|\bracket{i}{e_J}|^2=1$ where for the second inequality we used that $x_{\max}>1/\tau_{d,i_{\max}}$ and that $\tau_{d,i}/\tau_{d,i_{\max}}\geq 1,\forall i$.  
 This means that $X$ violates Eq.~\eqref{tr_constraint}, which is a contradiction. 
 Therefore, $\hat\Phi_d$ achieves the maximum of $R_G$, and hence $m_{\max}$. 
\end{proof}

\section{Collapse of the modification coefficients induced by exact resource destroying maps (part of Theorem 1)}\label{app:collapserd}

Now we prove that, for golden states such that the modification coefficients collapse, the exact resource destroying map defined by the relative entropy also just outputs the closest free state in terms of min- and max-relative entropies, and therefore the corresponding modification coefficients also collapse to the golden coefficient.  The case of min-relative entropy is rather straightforward:
\begin{prop}\label{lemma:exact_min}
 Let $\lambda$ be an exact resource destroying map. Suppose the set of free states $\mF$ satisfies Condition (CH), i.e.~is formed by a convex hull of pure states. Then for pure minimum-free-fidelity state $\hat{\Phi}_d$, $\lambda(\hat{\Phi}_d)\in\mF$ is the closest free state as measured by the min-relative entropy, i.e.
 \begin{equation}\label{eq:min_lem}
     \min_{\sigma\in\mF}\dmin(\hat{\Phi}_d\|\sigma) = \dmin(\hat{\Phi}_d\|\lambda(\hat{\Phi}_d)).
 \end{equation}
 That is to say,
  \begin{equation}
     m_{\min,\lambda}(\hat{\Phi}_d) = m_{\min}(\hat{\Phi}_d) = g_d.
 \end{equation}
\end{prop}
 \begin{proof}
Due to the collapse of modification coefficients shown in Theorem \ref{thm:collapse}, we have $m_{\min}(\hat{\Phi}_d) = m_{\max}(\hat{\Phi}_d)$.  That is, $\min_{\sigma\in\mF}\dmin(\hat{\Phi}_d\|\sigma) = \min_{\sigma\in\mF}\dmax(\hat{\Phi}_d\|\sigma)$.
Since $\dmin(\rho\|\sigma)\leq D(\rho\|\sigma) \leq \dmax(\rho\|\sigma)$ for any states $\rho,\sigma$ \cite{dattarel}, it must hold that
\begin{equation}
    \min_{\sigma\in\mF}\dmin(\hat{\Phi}_d\|\sigma) = \min_{\sigma\in\mF}D(\hat{\Phi}_d\|\sigma) = D(\hat{\Phi}_d\|\lambda(\hat{\Phi}_d)).  \label{eq:minlam}
\end{equation}
The last equality follows from the definition of the exact resource destroying map. Notice that $\min_{\sigma\in\mF}\dmin(\hat{\Phi}_d\|\sigma)\leq \dmin(\hat{\Phi}_d\|\lambda(\hat{\Phi}_d)) \leq D(\hat{\Phi}_d\|\lambda(\hat{\Phi}_d))$.  Then due to Eq.~(\ref{eq:minlam}), all the inequalities become equalities, and thus the statement follows. 
 \end{proof}

To prove the case of max-relative entropy we need to make use of the specific forms of the divergences:
 {
\begin{prop}\label{lemma:exact_max}
 Let $\lambda$ be an exact resource destroying map. Suppose the set of free states $\mF$ satisfies Condition (CH), i.e.~is formed by a convex hull of pure states. Then for pure minimum-free-fidelity state $\hat{\Phi}_d$, $\lambda(\hat{\Phi}_d)\in\mF$ is the closest free state as measured by the max-relative entropy if $\dmax(\hat{\Phi}_d\|\lambda(\hat{\Phi}_d))$ is well defined, i.e.
 \begin{equation}
     \min_{\sigma\in\mF}\dmax(\hat{\Phi}_d\|\sigma) = \dmax(\hat{\Phi}_d\|\lambda(\hat{\Phi}_d)).
 \end{equation}
  That is to say,
  \begin{equation}
     m_{\max,\lambda}(\hat{\Phi}_d) = m_{\max}(\hat{\Phi}_d) = g_d.
 \end{equation}
\end{prop}
 \begin{proof}
 By Proposition \ref{lemma:exact_min}, to prove the statement, it suffices to show that
 \begin{equation} \label{eq:max<min}
 \dmax(\hat{\Phi}_d\|\lambda(\hat{\Phi}_d))
 \leq \dmin(\hat{\Phi}_d\|\lambda(\hat{\Phi}_d)).
 \end{equation}
Due to Eqs.~\eqref{eq:min_lem} and  \eqref{eq:minlam}, for pure minimum-free-fidelity state $\hat{\Phi}_d$, we have
\begin{equation}
\dmin(\hat{\Phi}_d\|\lambda(\hat{\Phi}_d))
=D(\hat{\Phi}_d\|\lambda(\hat{\Phi}_d)),
\end{equation}
which can be rewritten as
\begin{equation}
\log \Tr\{\hat{\Phi}_d \lambda(\hat{\Phi}_d)\} = \Tr \{ \hat{\Phi}_d \log \lambda(\hat{\Phi}_d)\}.
\end{equation}
Suppose that $\lambda(\hat{\Phi}_d)$ has a spectral decomposition $\lambda(\hat{\Phi}_d)=\sum^{K}_i p_i\ket{c_i}\bra{c_i}$ where
 $\{\ket{c_i}\}^K_{i=1}$ is the support, i.e.~the subset of orthonormal basis $\{\ket{c_i}\}^{d}_{i=1}$ such that $\sum^K_{i=1}p_i=1$, $p_i>0$ for any 
 $i\in\{1,..,K\}$. 
Let us denote $q_i=\bra{c_i}\hat{\Phi}_d\ket{c_i}\geq 0$. 
Since
$\dmax(\hat{\Phi}_d\|\lambda(\hat{\Phi}_d))$ exists, we have 
$\sum^K_{i=1}q_i=1$. Otherwise, there would exist $j\geq K+1 $ such that
$q_j>0$, which implies  that
\begin{equation}
0<q_j=\bra{c_j}\hat{\Phi}_d\ket{c_j} \leq 2^{\dmax(\hat{\Phi}_d\|\lambda(\hat{\Phi}_d))}\bra{c_j}\lambda(\hat{\Phi}_d)\ket{c_j}. 
\end{equation}
This is impossible since $\bra{c_j}\lambda(\hat{\Phi}_d)\ket{c_j} = 0$.
Therefore, we obtain
\begin{equation}\label{eq:log_t}
\log\sum^K_{i=1}p_iq_i
=\sum^K_{i=1}q_i\log p_i,
\end{equation}
where $\sum^K_{i=1}q_i=\sum^{K}_{i=1}p_i=1$.  
Due to the concavity of the logarithm function,
Eq.~\eqref{eq:log_t} is true iff 
either i) only one $q_i=1$, or ii) all $p_i$ are equal for $i\in \{ k| q_k \neq 0\}$. 

\emph{Case 1:}
 Only one $q_i=1$.  Without loss of generality, assume that $q_1=1$. Then $\hat{\Phi}_d=\ketbra{c_1}{c_1}$,  and
$\dmin(\hat{\Phi}_d\|\lambda(\hat{\Phi}_d))=-\log p_1$. 
Thus $\dmax(\hat{\Phi}_d\|\lambda(\hat{\Phi}_d))\leq \log \frac{1}{p_1}=\dmin(\hat{\Phi}_d\|\lambda(\hat{\Phi}_d))$. 

\emph{Case 2:}
All $p_i$ are equal for $i\in\{k| q_k\neq 0\}$. Without loss of generality, assume that $1\in \{k| q_k\neq 0\}$.
Then $\dmin(\hat{\Phi}_d\|\lambda(\hat{\Phi}_d))=-\log p_1$ and 
$\hat{\Phi}_d\leq \sum_{i\in \{k| q_k\neq 0\}}\ket{c_i}\bra{c_i}$. 
Thus we also have $\dmax(\hat{\Phi}_d\|\lambda(\hat{\Phi}_d))\leq \log \frac{1}{p_1}=\dmin(\hat{\Phi}_d\|\lambda(\hat{\Phi}_d))$.

Therefore, Eq.~(\ref{eq:max<min}) is confirmed and the statement follows.
 \end{proof}
}

Proposition \ref{lemma:exact_max} in particular leads to roughly matching bounds for the formation cost under commuting operations in corresponding theories. 
Also note that the max-relative entropy could in principle induce different orderings as compared to the relative entropy, in the sense that there exist state (or even pure state) $\rho$ and states $\sigma_1,\sigma_2$ such that $\dmax(\rho\|\sigma_1) > \dmax(\rho\|\sigma_2)$ but $D(\rho\|\sigma_1) < D(\rho\|\sigma_2)$. The following is an example provided by Mil\'{a}n Mosonyi.  Consider a maximally entangled state $\ket{\hat\Sigma}=\frac{1}{\sqrt{d}}\sum_{j=1}^{d}\ket{j}\ket{j}\in\mH_d\otimes\mH_d$, and for any density operator $\sigma\in\mD(\mH_d)$, let 
$\bar\sigma:=\frac{1}{d}I\otimes\sigma$. Then
\begin{align}
    \dmax(\hat\Sigma\|\bar\sigma) &= \log\|\bar\sigma^{-1/2}\hat{\Sigma}\bar\sigma^{-1/2}\| = \log\Tr\sigma^{-1},\\
    D(\hat\Sigma\|\bar\sigma) &= -\Tr\{\hat{\Sigma}\log\bar\sigma\} = \log{d}-\frac{1}{d}\Tr\log\sigma,
\end{align}
where $\|\cdot\|$ denotes the operator norm.   Now let $\vec\lambda^{(i)} = (\lambda^{(i)}_1,...,\lambda^{(i)}_d)$ denote the eigenvalues of $\sigma_i,i=1,2$.  Then
\begin{align}
    \dmax(\hat\Sigma\|\bar\sigma_1) > \dmax(\hat\Sigma\|\bar\sigma_2) &\Longleftrightarrow  \sum_{j=1}^d\frac{1}{\lambda^{(1)}_j} > \sum_{j=1}^d\frac{1}{\lambda^{(2)}_j}\quad\text{or}\quad H(\vec\lambda^{(1)})<H(\vec\lambda^{(2)}),\\
    D(\hat\Sigma\|\bar\sigma_1) < D(\hat\Sigma\|\bar\sigma_2) &\Longleftrightarrow  \prod_{j=1}^d \lambda^{(1)}_j > \prod_{j=1}^d \lambda^{(2)}_j\quad\text{or}\quad G(\vec\lambda^{(1)})>G(\vec\lambda^{(2)}),
\end{align}
where $H$ and $G$ denote the harmonic mean and the geometric mean, respectively.
Note that $H(\vec\lambda)\leq G(\vec\lambda)$.  An explicit example satisfying the above properties found with Matlab is
\begin{equation}
    \vec\lambda^{(1)} = (0.5296, 0.0228, 0.4476),\quad\quad\quad
    \vec\lambda^{(2)} = (0.0368, 0.1570, 0.8062).
\end{equation}
That is to say, the closest free state given by the relative entropy, i.e.~the image of the exact resource destroying map (even if it is unique), is not necessarily the closest in terms of the max-relative entropy (say, consider a theory where $\mF=\{\bar\sigma_1,\bar\sigma_2\}$).   But Proposition \ref{lemma:exact_max} indicates that this does not occur for the references states of most interest in this work, namely the golden states (associated with some theory, which causes the optimized modification coefficients to collapse).

For certain theories, Proposition \ref{lemma:exact_max} may also admit simpler proofs due to the specific structures of the theory. As an example, consider the theory of coherence.  Here, the golden states are the maximally coherent states $\ket{\hat{\Psi}_d} = \frac{1}{\sqrt{d}}\sum_{i=1}^d\ket{i}$  where $\{\ketbra{i}{i}\}$ is the reference basis, the set of free states $\mF = \mathcal{I}$ consists of density matrices that are diagonal in the reference basis (incoherent states), and the exact resource destroying map is the complete dephasing channel $\Delta$. Since $\ket{\hat{\Psi}_d}$ is pure, we have \cite{PhysRevLett.119.150405}
\begin{equation}
    \min_{\sigma\in\mathcal{I}}\dmax(\hat{\Psi}_d\|\sigma) =  \min_{\sigma\in\mathcal{I}}\log(\|\hat{\Psi}_d-\sigma\|_{\ell_1}+1) = \log{d} = \dmax(\hat{\Psi}_d\|\Delta({\hat{\Psi}_d})). \label{eq:merdm_coh}
\end{equation}
Therefore, we have $m_{\max,\Delta}(\hat{\Psi}_d) = m_{\max}(\hat{\Psi}_d)$.

\section{Resource theory of superposition and the constant trace condition} \label{app:ct}
Here we provide a simple example of resource theories that satisfy the Condition (CT), in addition to the theory of coherence, based on the recently formulated resource theory of superposition~\cite{Theurer2017}.   

The theory of superposition is a generalization of coherence theory, where the set of free states is taken to be $\mF={\rm conv}\big\{\dm{\phi}\,\big|\,\ket{\phi}\in \mathcal{V}\big\}$ where $\mathcal{V}$ is a set of pure basis states, which do not have to be orthogonal to each other, unlike the theory of coherence. 
We shall show that any theory of superposition defined on two-dimensional systems with two pure basis states, with the golden state as the reference state, satisfies Condition (CT). Namely, the golden state $\hat\Phi_2$ has the same overlap with all the free states: $\Tr\{\hat\Phi_2\delta\}={\rm constant},\ \forall \delta\in\mF$. 
To see this, geometrical consideration is helpful.
Let $\ket{\phi_a}$ and $\ket{\phi_b}$ be the basis states (so $\mF = {\rm conv}\{\dm{\phi_a},\dm{\phi_b}\}$) and consider the slice of the Bloch sphere which intersects $\ket{\phi_a}, \ket{\phi_b}$ and the center (see Fig.~\ref{fig:superposition_bloch}).
For two-dimensional systems, the fidelity and the trace distance are equivalent for two pure states, and the trace distance is proportional to the Euclidean distance on the Bloch sphere. 
Therefore, $\hat\Phi_2$, the pure state that has the minimum fidelity with the closest free state  (namely the golden state, since $\mF$ satisfies Condition (CH) by definition), is found to be located at the opposite side of $\ket{\phi_a}$ and $\ket{\phi_b}$ on this slice, with which it forms an isosceles triangle, as depicted in Fig.\,\ref{fig:superposition_bloch}.
It is easy to see that $\Tr\{\hat\Phi_2\dm{\phi_a}\}=\Tr\{\hat\Phi_2\dm{\phi_b}\}$, and this ensures $\Tr\{\hat\Phi_2\delta\}={\rm constant},\ \forall \delta\in\mF$ due to the linearity of the trace function. 
To construct a scalable theory, one may, for example, consider the multi-qubit extension where $\mF = {\rm conv}\{\dm{\phi_a}^{\otimes n},\dm{\phi_b}^{\otimes n}\}$ and $\hat\Phi_2^{\otimes n}$ is the reference state for any $n\in\mathbb{Z}_+$.  It can be directly seen that Condition (CT) is still met.
\begin{figure}[htbp]
    \centering
    \includegraphics[scale=0.5]{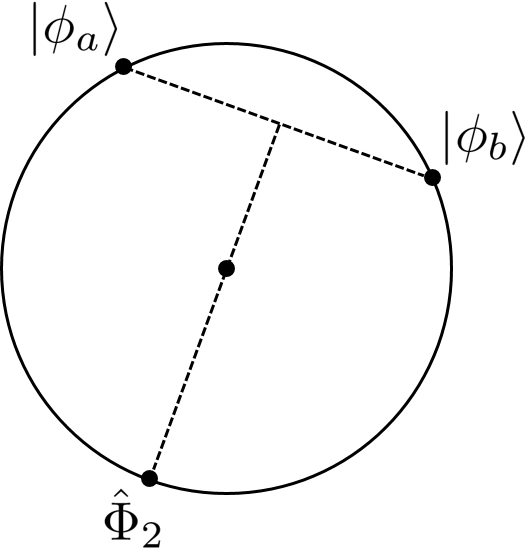}
    \caption{The slice of the Bloch sphere which intersects $\ket{\phi_a},\ket{\phi_b}$ and the center. The golden state $\hat\Phi_2$ locates at the opposite of $\ket{\phi_a},\ket{\phi_b}$ and form an isosceles triangle with $\ket{\phi_a},\ket{\phi_b}$.}
    \label{fig:superposition_bloch}
\end{figure}

\section{Classification of resource theories}\label{app:condition}

Here we provide a Venn diagram (Fig.~\ref{fig:condition}) that illustrates the classification of quantum resource theories according to a number of key properties that are most relevant to our work.  
\begin{figure}[htbp]
    \centering
    \includegraphics[width=0.8\textwidth]{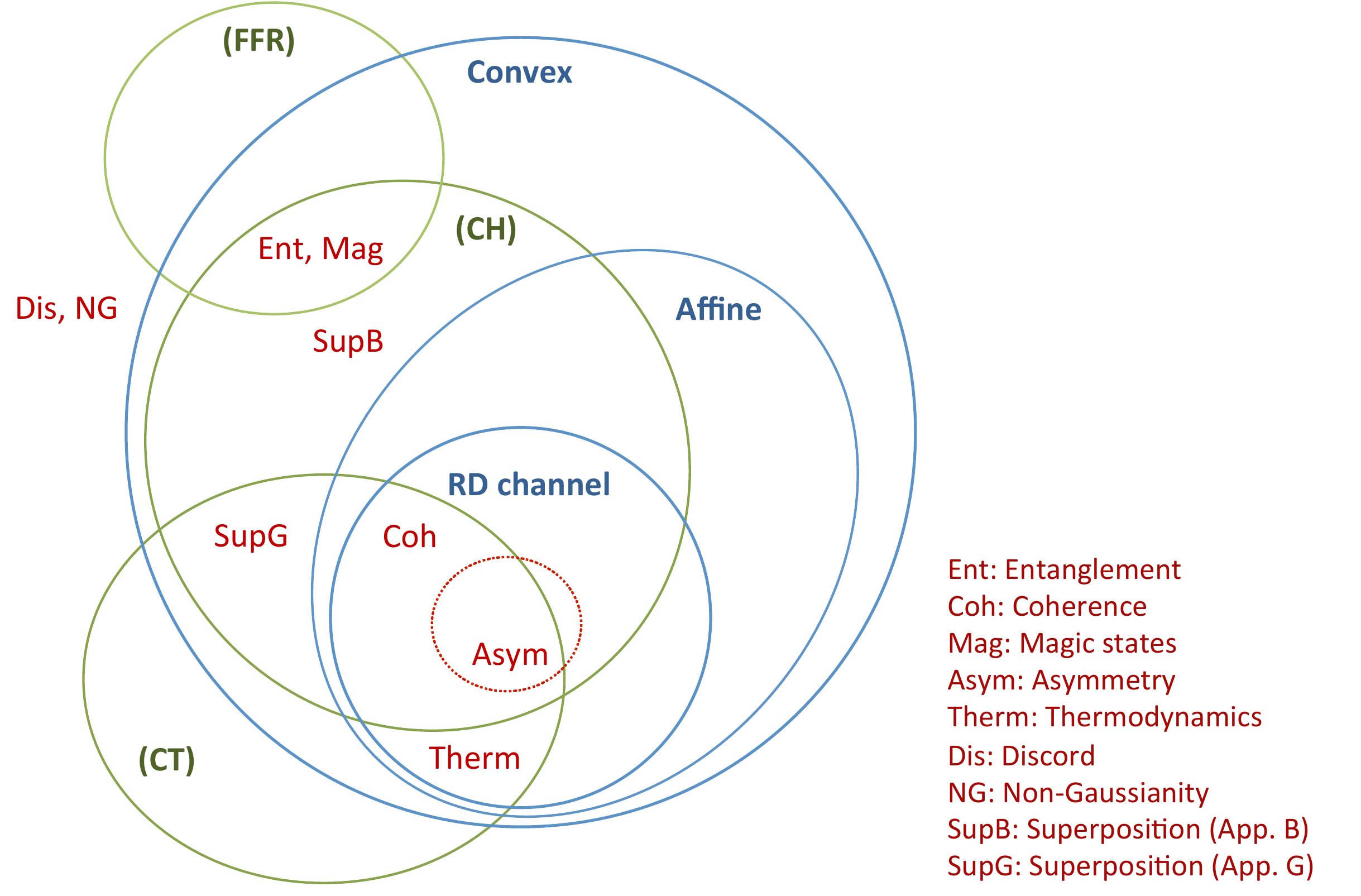}
   \caption{Classification of resource theories according to several basic properties (convexity, affinity, existence of RD channel; blue circles) as well as the three conditions introduced in the main text (Conditions (FFR), (CH), (CT); green circles), and where the important examples (red code names) live. SupB and SupG refer to the specific theories of superposition discussed in Appendix \ref{app:fr_affine} and \ref{app:ct} respectively. Note that whether a theory of asymmetry satisfies Condition (CT) depends on further details of the symmetry of interest, which is why we indicate it as a circle. For instance, when $U(1)$ group is considered, a set of free states defined on a single subsystem satisfies Condition (CT) whereas that defined on a composite system involving two subsystems does not ---  It is due to the fact that free states can be entangled within a degenerate subspace.}  
   \label{fig:condition}
\end{figure}

\section{One-shot formation cost by resource non-generating operations (part of Theorem 2)} \label{app:cost_ng}
In the next two sections, we break down Theorem 2 of the main text about one-shot formation into independent results that may rely on different assumptions, and present their proofs separately.   

Here, we present proofs of the bounds on the one-shot formation cost under resource non-generating operations $\barx$, which constitute the maximal set of free operations.  

The first group of bounds are given by modified versions of smooth max-relative entropy of resource $\rdmax^\epsilon$.  The following is a lower (optimality) bound:
\begin{prop}
Given target state $\rho$ and reference states $\{\phi_d\}$. Let $d_0 = \min\{d\in\mathbb{D}: \rdmax(\phi_d)=m_{\max}(\phi_{d})\log d \geq \rdmax^\epsilon(\rho)\}$.   For $\epsilon\geq 0$, 
\begin{equation}
    \Omega_{C,\barx}^\epsilon(\rho\leftarrow\{\phi_d\}) \geq \log d_0 \geq\frac{\rdmax^\epsilon(\rho)}{m_{\max}(\phi_{d_0})}.
\end{equation}
\end{prop}
\begin{proof}
Suppose $\phi_d$ is a reference state that achieves the formation task with free operation $\mathcal{E}$, i.e.~$\mathcal{E}(\phi_d) = \rho_\epsilon\in\eball(\rho)$   for some $\mathcal{E}\in\barx$. 
Let $\tilde{\sigma}\in \mF$ be the optimal free state that achieves  $\rdmax(\phi_d)$.   The following must hold:
\begin{align}
    \rdmax^\epsilon(\rho) \leq& \rdmax(\rho_\epsilon)\\ =& \min_{\delta\in \mF}\dmax(\rho_\epsilon\|\delta) \\ 
    \leq& \dmax(\rho_\epsilon\|\mathcal{E}(\tilde{\sigma})) \\
    =&  \dmax(\mathcal{E}(\phi_d)\|\mathcal{E}(\tilde{\sigma})) \\
    \leq& \dmax(\phi_d\|\tilde{\sigma}) = \rdmax(\phi_d) = m_{\max}(\phi_{d})\log{d},
\end{align}
where the third line follows from $\mathcal{E}(\sigma)\in \mF$, and the fifth line follows from the data processing inequality for the max-relative entropy \cite{qrenyi}.   So for any $d < d_0$ the possibility that $\phi_d$ achieves the formation is forbidden by the above inequality, that is, $\Omega_{C,\barx}^\epsilon(\rho\leftarrow\{\phi_d\}) \geq \log d_0$.  Also by the definition of $d_0$ we have $m_{\max}(\phi_{d_0})\log {d_0} \geq  \rdmax^\epsilon(\rho)$. So the claimed bound follows.
\end{proof}

We can also obtain the following general upper (achievability) bound:
\begin{prop}
Given target state $\rho$ and pure reference states $\{\Phi_d\}$.  Suppose the resource theory satisfies Condition (CT), i.e.~for any given $d\in\mathbb{D}$, $\Tr\{\Phi_d\sigma\}$ is constant for any $\sigma\in\mF$.  Let $d_0 = \min\{d\in\mathbb{D}:-\log\mathfrak{f}(\Phi_d) = m_f(\Phi_d)\log{d} \geq \rdmax^\epsilon(\rho)\}$.   For  $\epsilon\geq 0$,
\begin{equation}
    \Omega_{C,\barx}^\epsilon(\rho\leftarrow\{\phi_d\}) \leq \log d_0 < \frac{\rdmax^\epsilon(\rho)}{m_f(\Phi_{d_0^\downarrow})} + \log\frac{d_0}{d_0^\downarrow}.
\end{equation}
\end{prop}
\begin{proof}
Let $\rho_\epsilon$ be the state in $\eball(\rho)$ and $\delta$ be the free state that achieves $\rdmax^\epsilon(\rho)$, so $\rho_\epsilon \leq 2^{\rdmax^\epsilon(\rho)}\delta$. 
Define the following map $\mathcal{E}$  on input state $\omega$:
\begin{align}
    \mathcal{E}(\omega)=&\frac{1}{1-d^{-m_f(\Phi_{d_0})}_0}\left(\Tr\{\Phi_{d_0}\omega\}-d^{-m_f(\Phi_{d_0})}_0\right)\rho_\epsilon + \frac{1}{1-d^{-m_f(\Phi_{d_0})}_0}(1-\Tr\{\Phi_{ d_0}\omega\})\delta \label{eq:cost NG CT}\\
=&\frac{1}{1-d^{-m_f(\Phi_{d_0})}_0}(1-\Tr\{\Phi_{d_0}\omega\})\left(\delta-d^{-m_f(\Phi_{d_0})}_0\rho_\epsilon\right)
+\Tr\{\Phi_{d_0}\omega\}\rho_{\epsilon}.\label{eq:cost NG CT 2}
\end{align}
By the definition of $d_0$, we have $\delta - d^{-m_f(\Phi_{d_0})}_0\rho_\epsilon \geq 0$. So $\mathcal{E}$ is a cptp map.  By Condition (CT), we have $\Tr\{\Phi_{d_0}\sigma\}  = d^{-m_f(\Phi_{d_0})}_0$ for any $\sigma\in \mF$, so it can be directly seen from Eq.~(\ref{eq:cost NG CT}) that $\mathcal{E}(\sigma) = \delta \in \mF, \forall \sigma\in\mF$, which implies that $\mathcal{E}\in\barx$.  Finally, it can be directly seen from Eq.~(\ref{eq:cost NG CT 2}) that $\mathcal{E}(\Phi_{d_0}) = \rho_\epsilon$, so $\log d_0$ is an achievable rate.  To obtain a bound in terms of $\rdmax^\epsilon(\rho)$, notice that $m_f(\Phi_{d_0^\downarrow})\log(d_0^\downarrow)<\rdmax^\epsilon(\rho)$ by the definition of $d_0$, and so the claimed bound follows.
\end{proof}

When using the golden states $\{\hat\Phi_d\}$ as the currency, some bounds can be simplified and unified due to the collapse of $m_f(\hat\Phi_d), m_{\min}(\hat\Phi_d), m_{\max}(\hat\Phi_d)$ to the golden modification coefficient $g_d$ (Proposition \ref{thm:collapse}).  In particular, we can establish two-way bounds which are expected to be rather tight when $d_0^\downarrow$ is taken to be close to $d_0$, indicating the one-shot formation cost is rather accurately given by the modified smooth max-relative entropy of resource, for the following case: 
\begin{cor}
Given target state $\rho$, and consider golden states $\{\hat\Phi_d\}$ as the reference states.   Suppose $\mF$ satisfies Condition (CH), i.e.~is formed by a convex hull of pure states, and the resource theory satisfies Condition (CT), i.e.~for any given $d\in\mathbb{D}$, $\Tr\{\Phi_d\sigma\}$ is constant for any $\sigma\in\mF$.  Let $d_0 = \min\{d\in\mathbb{D}: g_d\log{d} \geq \rdmax^\epsilon(\rho)\}$.   For $\epsilon\geq 0$,  
\begin{equation}
  \frac{\rdmax^\epsilon(\rho)}{g_{d_0}}  \leq \Omega_{C,\barx}^\epsilon(\rho\leftarrow\{\hat\Phi_d\})  < \frac{\rdmax^\epsilon(\rho)}{g_{d_0^\downarrow}} + \log\frac{d_0}{d_0^\downarrow}.
\end{equation}
\end{cor}

The bounds could also be given by modified versions of smooth free log-robustness when the free robustness measure is always well-defined.  The following is a lower bound:
\begin{prop}
Given target state $\rho$ and reference states $\{\phi_d\}$.   Suppose $\mF$ satisfies Condition (FFR), i.e.~all states have finite free robustness. 
Let $d_0 = \min\{d\in\mathbb{D}:LR(\phi_d)=m_{LR}(\phi_{d})\log d \geq LR^\epsilon(\rho)\}$.  For $\epsilon\geq 0$,
\begin{equation}
    \Omega_{C,\barx}^\epsilon(\rho\leftarrow\{\phi_d\}) \geq \log d_0 \geq \frac{LR^\epsilon(\rho)}{m_{LR}(\phi_{d_0})}.
\end{equation}
\end{prop}
\begin{proof}
Suppose $\phi_d$ is a reference state that achieves the formation task with free operation $\mathcal{E}$, i.e.~$\mathcal{E}(\phi_d) = \rho_\epsilon\in\eball(\rho)$  for some $\mathcal{E}\in\barx$. 
Then it must hold that 
\begin{equation}
    LR^\epsilon(\rho) = LR(\rho_\epsilon) = LR(\mathcal{E}(\phi_d)) \leq LR(\phi_d) = m_{LR}(\phi_{d})\log d,
\end{equation}
where the inequality follows from the monotonocity of free log-robustness \cite{regula}.   So for any $d < d_0$ the possibility that $\phi_d$ achieves the formation is forbidden by the above inequality, that is, $\Omega_{C,\barx}^\epsilon(\rho\leftarrow\{\phi_d\}) \geq \log d_0$.  Also by the definition of $d_0$ we have $m_{LR}(\phi_{d_0})\log {d_0} \geq  LR^\epsilon(\rho)$. So the claimed bound follows.
\end{proof}

We can also use a different formation map to obtain the following upper bound:
\begin{prop}
Given target state $\rho$ and pure reference states $\{\Phi_d\}$.
Suppose $\mF$ is convex, and satisfies Condition (FFR), i.e.~all states have finite free robustness.  
 Let $d_0 = \min\{d\in\mathbb{D}:-\log\mathfrak{f}(\Phi_d) =m_f(\Phi_d)\log d \geq LR^\epsilon(\rho)\}$.
For $\epsilon\geq 0$, 
\begin{equation}
    \Omega_{C,\barx}^\epsilon(\rho\leftarrow\{\phi_d\}) \leq \log d_0 < \frac{LR^\epsilon(\rho)}{m_f(\Phi_{d_0^\downarrow})} + \log\frac{d_0}{d_0^\downarrow}.
\end{equation}
\end{prop}
\begin{proof}
Let $\rho_\epsilon$ be the state in $\eball(\rho)$ and $\delta$ be the free state that achieves $R^\epsilon(\rho)$.
Let $\alpha = 2^{LR^\epsilon(\rho)} = 1+R^\epsilon(\rho)$. Then there exists $\delta \in \mF$ such that $\delta' = \frac{1}{\alpha}\rho_\epsilon + (1-\frac{1}{\alpha})\delta\in \mF$.
Define the following cptp map  on input state $\omega$:
\begin{align}
    \mathcal{E}(\omega)=& \Tr\{\Phi_{d_0}\omega\}\rho_\epsilon + (1-\Tr\{\Phi_{d_0}\omega\})\delta \label{eq:cost NG FFR}\\
=& \alpha\Tr\{\Phi_{d_0}\omega\}\delta' + (1-\alpha\Tr\{\Phi_{d_0}\omega\})\delta.
\end{align}
When $\omega\in \mF$, by the definition of $d_0$ we directly have $\Tr\{\Phi_{d_0}\omega\}\leq 2^{-LR^\epsilon(\rho)}$, so $\alpha\Tr\{\Phi_{d_0}\omega\}\leq 1$. Since $\delta,\delta'\in \mF$, we have $\mathcal{E}(\omega)\in \mF$ when $\omega\in \mF$ due to the convexity of $\mF$, which implies that $\mathcal{E}\in \barx$.
Finally, it can be directly verified that $\mathcal{E}(\Phi_{d_0}) = \rho_\epsilon$, so $\log d_0$ is an achievable rate.  To obtain a bound in terms of $LR^\epsilon(\rho)$, notice that $m_f(\Phi_{d_0^\downarrow})\log(d_0^\downarrow)<LR^\epsilon(\rho)$ by the definition of $d_0$, and so the claimed bound follows.
\end{proof}

Note again that the separate results may hold under particular assumptions on the resource theory and reference states.   In certain situations where the assumptions for different lower or upper bounds are met, one may compare the valid bounds and take the strongest one (same for the results below).  



\section{One-shot formation cost by commuting operations (part of Theorem 2)} \label{app:cost_comm}
In this section, we present proofs of the bounds on the one-shot formation cost under commuting operations with respect to resource destroying map $\lambda$, i.e.~$\mathscr{F}_{\lambda,{\rm Comm}}$ which is in general a more restricted set of free operations than $\barx$.
The bounds on the one-shot formation cost under $\mathscr{F}_{\lambda,{\rm Comm}}$ are given by modified versions of the smooth max-relative entropy between the states and its resource-destroyed version induced by $\lambda$, namely the $\lambda$-max-relative entropy of resource $\rdmaxl^\epsilon$.  We first present the most general lower bound without any restriction on $\lambda$ and then simplify the result in the case that $\lambda$ is exact:
\begin{prop}
Given target state $\rho$, resource destroying map $\lambda$, and reference states $\{\phi_d\}$. Let $d_0 = \min\{d\in\mathbb{D}:\mathfrak{D}_{\max,\lambda}
({\phi}_d) = m_{\max,\lambda}(\phi_{d})\log d \geq \rdmaxl^\epsilon(\rho)\}$.   For  $\epsilon\geq 0$,
\begin{equation}
    \Omega_{C,\mathscr{F}_{\lambda,{\rm Comm}}}^\epsilon(\rho\leftarrow\{\phi_d\}) \geq \log d_0 \geq\frac{\rdmaxl^\epsilon(\rho)}{m_{\max,\lambda}(\phi_{d_0})}.
\end{equation}

Now consider golden states $\{\hat{\Phi}_d\}$ as the reference states.  Suppose the set of free states $\mF$ satisfies Condition (CH), i.e.~is formed by a convex hull of pure states, and $\tilde\lambda$ is an exact resource destroying map. 
Let $d'_0 = \min\{d\in\mathbb{D}:\mathfrak{D}_{\max,\tilde\lambda}
(\hat{\Phi}_d) = \mathfrak{D}_{\max}
(\hat{\Phi}_d) = m_{\max}(\hat{\Phi}_d)\log d =g_d\log d \geq \mathfrak{D}_{\max,\tilde{\lambda}}^\epsilon(\rho)\}$. For $\epsilon\geq 0$,
\begin{equation}
    \Omega_{C,\mathscr{F}_{\tilde\lambda,{\rm Comm}}}^\epsilon(\rho\leftarrow\{\hat\Phi_d\}) \geq \log d'_0\geq\frac{\mathfrak{D}_{\max,\tilde{\lambda}}^\epsilon(\rho)}{m_{\max}(\hat{\Phi}_{d'_0})} = \frac{\mathfrak{D}_{\max,\tilde{\lambda}}^\epsilon(\rho)}{g_{d'_0}}.
\end{equation}
\end{prop}
\begin{proof}
Suppose $\phi_d$ is a reference state that achieves the formation task with free operation $\mathcal{E}$, i.e.~$\mathcal{E}(\phi_d) = \rho_\epsilon\in\eball(\rho)$ for some $\mathcal{E}\in\mathscr{F}_{\lambda,{\rm Comm}}$. 
The following must hold:
\begin{align}
    \rdmaxl^\epsilon(\rho) \leq& \rdmaxl(\rho_\epsilon)\\ =& \dmax(\rho_\epsilon\|\lambda(\rho_\epsilon)) \\ 
    =& \dmax(\mathcal{E}(\phi_d)\|\lambda(\mathcal{E}(\phi_d))) \\
    =&  \dmax(\mathcal{E}(\phi_d)\|\mathcal{E}(\lambda(\phi_d)))  \\
    \leq& \dmax(\phi_d\|\lambda(\phi_d)) = \rdmaxl(\phi_d) = m_{\max,\lambda}(\phi_{d})\log{d},
\end{align}
where the fourth line follows from $\mathcal{E}\in \mathscr{F}_{\lambda,{\rm Comm}}$, i.e.~$\mathcal{E}$ commutes with $\lambda$, and the fifth line follows from the data processing inequality for the max-relative entropy \cite{qrenyi}.   So for any $d < d_0$ the possibility that $\phi_d$ achieves the formation is forbidden by the above inequality, that is, $\Omega_{C,\mathscr{F}_{\lambda,{\rm Comm}}}^\epsilon(\rho\leftarrow\{\phi_d\}) \geq \log d_0$.  Also by the definition of $d_0$ we have $m_{\max,\lambda}(\phi_{d_0})\log {d_0} \geq  \rdmaxl^\epsilon(\rho)$. So the first claimed bound follows.  

The second bound is a direct consequence of the collapse $m_{\max,\tilde\lambda} = m_{\max}$ handled by Proposition~\ref{lemma:exact_max}.
\end{proof}

We also obtain the following general upper bound:
\begin{prop}
Given target state $\rho$, resource destroying map $\lambda$, and pure reference states $\{\Phi_d\}$.  Suppose the  resource theory satisfies Condition (CT), i.e.~for any given $d\in\mathbb{D}$, $\Tr\{\Phi_d\sigma\}$ is constant for any $\sigma\in\mF$.  Let $d_0 = \min\{d\in\mathbb{D}: -\log\mathfrak{f}(\Phi_d) =m_f(\Phi_{d})\log d\geq \mathfrak{D}^\epsilon_{\max, \lambda}(\rho)\}$.  For $\epsilon\geq 0$,
 \begin{equation} 
     \Omega_{C,\mathscr{F}_{\lambda,{\rm Comm}}}^\epsilon(\rho\leftarrow\{\phi_d\})  \leq \log d_0 < \frac{\mathfrak{D}^\epsilon_{\max, \lambda}(\rho)}{m_f(\Phi_{d_0^\downarrow})} + \log\frac{d_0}{d_0^\downarrow}.
\end{equation}
\end{prop}
\begin{proof}
Let $\rho_\epsilon$ be the state in $\mathcal{B}^\epsilon(\rho)$ that achieves $\rdmaxl^\epsilon(\rho)$, so
$\rho_{\epsilon}\leq 2^{\rdmaxl^\epsilon(\rho)}\lambda(\rho)$. Define the following map on input state $\omega$:
\begin{align}
\mathcal{E}(\omega)=&\frac{1}{1-d^{-m_f(\Phi_{d_0})}_0}\left(\Tr\{\Phi_{d_0}\omega\}-d^{-m_f(\Phi_{d_0})}_0\right)\rho_\epsilon + \frac{1}{1-d^{-m_f(\Phi_{d_0})}_0}(1-\Tr\{\Phi_{ d_0}\omega\})\lambda(\rho_{\epsilon})\\
=&\frac{1}{1-d^{-m_f(\Phi_{d_0})}_0}(1-\Tr\{\Phi_{d_0}\omega\})\left(\lambda(\rho_{\epsilon})-d^{-m_f(\Phi_{d_0})}_0\rho_\epsilon\right)
+\Tr\{\Phi_{d_0}\omega\}\rho_{\epsilon}\label{eq:f8}
\end{align}
By the definition of $d_0$, we have $\lambda(\rho_{\epsilon})-d^{-m_f(\Phi_{d_0})}_0\rho_\epsilon\geq 0$. So $\mathcal{E}$ is a cptp map. 
Now, notice that, on the one hand,
\begin{align}
\mathcal{E}(\lambda(\omega))=&\frac{1}{1-d^{-m_f(\Phi_{d_0})}_0}\left(\Tr\{\Phi_{d_0}\lambda(\omega)\}-d^{-m_f(\Phi_{d_0})}_0\right)\rho_\epsilon + \frac{1}{1-d^{-m_f(\Phi_{d_0})}_0}(1-\Tr\{\Phi_{ d_0}\lambda(\omega)\})\lambda(\rho_{\epsilon})\\
=&\lambda(\rho_\epsilon),
\end{align}
where we used $\Tr\{\Phi_{d_0}\lambda(\omega)\} = d^{-m_f(\Phi_{d_0})}_0$ due to Condition (CT).  On the other hand,
\begin{align}
\lambda(\mathcal{E}(\omega))=&\frac{1}{1-d^{-m_f(\Phi_{d_0})}_0}\left(\Tr\{\Phi_{d_0}\omega\}-d^{-m_f(\Phi_{d_0})}_0\right)\lambda(\rho_{\epsilon}) + \frac{1}{1-d^{-m_f(\Phi_{d_0})}_0}(1-\Tr\{\Phi_{ d_0}\omega\})\lambda(\rho_{\epsilon})\\
=&\lambda(\rho_{\epsilon}),
\end{align}
where we used the fact that $\lambda$ is idempotent.    That is, $\mathcal{E}(\lambda(\omega)) = \lambda(\mathcal{E}(\omega))$, i.e.~$\mathcal{E} \in  \mathscr{F}_{\lambda,{\rm Comm}}$. 
Finally, it can be directly seen from Eq.~(\ref{eq:f8}) that $\mathcal{E}(\Phi_{d_0}) = \rho_\epsilon$, so $\log d_0$ is an achievable rate.
 To obtain a bound in terms of $\rdmaxl^\epsilon(\rho)$, notice that $m_f(\Phi_{d_0^\downarrow})\log(d_0^\downarrow)<\mathfrak{D}^\epsilon_{\max, \lambda}(\rho)$ by the definition of $d_0$, and so the claimed bound follows.
\end{proof}

Again when considering golden states $\{\hat\Phi_d\}$ as the currency and exact resource destroying map $\tilde\lambda$, the collapse results  $m_f(\hat\Phi_d) = m_{\max}(\hat\Phi_d) = g_d$ due to Proposition~\ref{thm:collapse} and also $m_{\max}(\hat\Phi_d) = m_{\max,\tilde\lambda}(\hat\Phi_d)=g_d$  due to Proposition~\ref{lemma:exact_max} allow us to establish two-way bounds which are expected to be rather tight when $d_0^\downarrow$ is taken to be close to $d_0$ for the following case:
\begin{cor}
Given target state $\rho$, and consider golden states $\{\hat\Phi_d\}$ as the reference states. Suppose $\mF$ satisfies Condition (CH), i.e.~is formed by a convex hull of pure states, and $\tilde{\lambda}$ is an exact resource destroying map.  Let $d_0 = \min\{d\in\mathbb{D}: g_d\log{d} \geq \mathfrak{D}_{\max,\tilde{\lambda}}^\epsilon(\rho)\}$.  For $\epsilon\geq 0$,  
\begin{equation}
  \frac{\mathfrak{D}_{\max,\tilde{\lambda}}^\epsilon(\rho)}{g_{d_0}}  \leq \Omega_{C,\mathscr{F}_{\tilde\lambda,{\rm Comm}}}^\epsilon(\rho\leftarrow\{\hat\Phi_d\})  < \frac{\mathfrak{D}_{\max,\tilde{\lambda}}^\epsilon(\rho)}{g_{d_0^\downarrow}} + \log\frac{d_0}{d_0^\downarrow}.
\end{equation}
\end{cor}

\section{Root states and blind resource formation (Corollary 4 and additional results)} \label{app:root}
We first formally define the notion of root states, which represent the strongest notion of maximum resource for the corresponding Hilbert space:
\begin{defn}[Root state]
Let $\mL(\mH)$ be the set of linear operators acting on the Hilbert space $\mH$, and $\mT(\mH,\mH') = \{\Phi| \Phi:\mL(\mH) \rightarrow \mL(\mH')\}$ be the set of linear transformations that map the operators on the Hilbert space $\mH$ to the operators on the Hilbert space $\mH'$, and ${\mD}(\mH)$ be the set of density operators acting on $\mH$. 
A quantum state $\Psi_d\in{\mD}(\mH_d)$ is called a \emph{root state}, if for all states $\rho\in {\mD}(\mH_d)$, there exist a free operation $\mathcal{E}\in \barx \subset \mT(\mH_d,\mH_d)$ such that $\rho = \mathcal{E}(\Psi_d)$.
\end{defn}

By definition, a root state $\psi_d$ must achieve the maximum value of any resource monotone in dimension $d$.
So strictly speaking, the notion of root states is more restrictive than the max-resource states given by maximizing a certain set of resource monotones. 
Also, the maximum of any convex resource monotone is achieved by some pure state, which is a strong evidence that the root states can be pure in convex theories.


We now provide a proof for the following result on the existence of root states in certain cases, which was presented as Corollary 4 in the main text:
\begin{prop}
 For any $\mF(\mH_d)$ such that the maxima of $m\in\{m_f,m_{\min},m_{\max}\}$ coincide at some pure state $\hat\Phi_d$ (e.g.~$\mF(\mH_d)$ satisfying Condition (CH), i.e.~is formed by a convex hull of pure states), $\hat\Phi_d$ serves as a root state if $\mF(\mH_d)$ further satisfies either of the following: i) Condition (CT), i.e.~for any given $d\in\mathbb{D}$, $\Tr\{\Phi_d\sigma\}$ is constant for any $\sigma\in\mF$, ii) Condition (FFR), i.e.~all states have finite free robustness, and $m_{\max}(\Phi_d)=m_{LR}(\Phi_d)$ for any pure state $\Phi_d\in \mD(\mH_d)$. 
\end{prop}
\begin{proof}
 For the case where $\mF(\mH_d)$ satisfies Condition (CT), the map defined by Eq.~\eqref{eq:cost NG CT} with $\epsilon=0$ and $\Phi_{d_0}$ replaced by $\hat\Phi_{d'}$ is a cptp map that gives $\mathcal{E}(\hat\Phi_{d'})=\rho$ for any $\rho\in \mD(\mH_d)$, if $d'$ satisfies $-\log \mathfrak{f}(\hat\Phi_{d'})\geq \mathfrak{D}_{\max}(\rho)$. Since $-\log \mathfrak{f}(\hat\Phi_{d})=\mathfrak{D}_{\max}(\hat\Phi_d)\geq \mathfrak{D}_{\max}(\rho), \forall \rho\in \mD(\mH_d)$ by assumption, one can take $d'=d$, indicating that $\hat\Phi_d$ is a valid root state. 
 
 When $\mF(\mH_d)$ satisfies Condition (FFR) and $m_{\max}(\Phi_d)=m_{LR}(\Phi_d)$ for any pure state $\Phi_d\in \mD(\mH_d)$, one can instead consider the map defined by Eq.~\eqref{eq:cost NG FFR} with $\epsilon=0$  and $\Phi_{d_0}$ replaced by $\hat\Phi_{d'}$, which gives $\mathcal{E}(\hat\Phi_{d'})=\rho$ for any $\rho\in \mD(\mH_d)$ if $-\log \mathfrak{f}(\hat\Phi_{d'})\geq LR(\rho)$. 
 Since $-\log \mathfrak{f}(\hat\Phi_{d})=\mathfrak{D}_{\max}(\hat\Phi_d) = LR(\hat\Phi_d)\geq LR(\rho), \forall \rho\in \mD(\mH_d)$ where we used the assumption and the fact that $LR(\cdot)$ achieves the maximum at a pure state due to its convexity, $\hat\Phi_d$ again serves as a root state. 
\end{proof}

This result gives useful general conditions ensuring the existence of the strongest ``maximally resourceful state'',  including the theories of bipartite ($k$-)entanglement~\cite{nielsen}, ($k$-)coherence~\cite{PhysRevLett.113.140401, ringbauer_2018}, purity~\cite{hho}, and quantum thermodynamics.
It also implies that if there does not exist a root state, the two robustness measures, generalized robustness and free robustness, do not coincide for pure states in general. 
For instance, it has recently been shown that the theory of multipartite entanglement with respect to fully separable states does not allow for a root state~\cite{Contreras2019entanglement}.
Our result then immediately implies the interesting fact that the free and generalized robustnesses do not generally coincide at pure states in contrast to the case of bipartite entanglement~\cite{Steiner2003robustness,Harrow2003robustness}, providing an alternative proof to the one presented in Ref.~\cite{Contreras2019entanglement}.


\smallskip
The strong property of root states directly implies several simple features of resource manipulation when using them as the currency.
For example, it directly ensures that the size (dimension) of the optimal reference state is never larger than that of the target state in one-shot formation and distillation tasks, in which case they achieve the effect of resource ``dilution'' and ``concentration'' respectively.

Moreover, one may naturally consider the following relaxed definition of resource formation tasks without specifying the reference states, whose optimal rates are always achieved by root states:
\begin{defn}[One-shot blind resource formation]
Given target state $\rho$. The \emph{one-shot $\epsilon$-blind formation cost} of state $\rho$ under the set of free operations $\mathscr{F}$ is defined to be the minimum possible size of resource states needed to form an $\epsilon$-approximation $\rho$ by an operation in $\mathscr{F}$:
\begin{equation}
    \hat\Omega_{C,\mathscr{F}}^{\epsilon}(\rho) :=  \min\{\log d: \exists \sigma\in\mD(\mH_d), \exists \mathcal{E}\in \mathscr{F}, \mathcal{E}(\sigma)\in\mathcal{B}^\epsilon(\rho)\}.
\end{equation}
\end{defn}
\begin{prop}\label{lem:min_perfect}
The one-shot blind formation costs must be achieved by root states. In other words, the one-shot blind formation cost is equivalent to the one-shot formation cost with respect to any family of root states $\{\Psi_d\}$, i.e.~$\hat\Omega_{C,\mathscr{F}}^{\epsilon}(\rho) = \Omega_{C,\mathscr{F}}^{\epsilon}(\rho\leftarrow\{\Psi_d\})$.
\end{prop}
\begin{proof}
Suppose some state $\xi_d\in\mD(\mH_d)$ achieves the minimum cost.  By the definition of root states, there exists some $\mathcal{E}'\in \mathscr{F}$ such that $\mathcal{E}'(\Psi_d)=\xi_d$. Free operations should remain free under composition \cite{LiuWinter2018}. So $\Psi_d$ also achieves the minimum cost with free operation $\mathcal{E}\circ\mathcal{E}'$.
\end{proof}


\section{One-shot distillation yield by resource non-generating operations (Theorem 5 and additional results)}
\label{app:yield_ng}
In this section, we present proofs of the bounds on the one-shot distillation yield (the standard version with error tolerance on the target state) under resource non-generating operations $\barx$, that constitute Theorem 5 of the main text.  Here the bounds are given by modified versions of the hypothesis testing relative entropy (operator-smoothing of min-relative entropy) of resource $\mathfrak{D}_{H}^\epsilon$.  
The following is an upper bound:



\begin{prop}
Given primitive state $\rho$ and pure reference states $\{\Phi_d\}$.  Let $d_0 = \max\{d\in\mathbb{D}: -\log\mathfrak{f}(\Phi_d) = m_{f}(\Phi_d)\log d \leq \mathfrak{D}_H^\epsilon(\rho)\}$. 
 For $\epsilon\geq 0$, 
  \begin{equation}
   \Omega_{D,\mathbb{\barx}}^{\epsilon}(\rho\rightarrow\{\Phi_d\}) \leq \log d_0 \leq \frac{\mathfrak{D}_{H}^\epsilon(\rho)}{m_f(\Phi_{d_0})}.  \label{eq:opsm} 
 \end{equation}
\end{prop}
\begin{proof}
First, notice that $\mathfrak{D}_H^\epsilon$ is monotone under any free operation $\mathcal{E}\in\barx$. Let $\tilde\delta$ be the optimal free state that achieves $\mathfrak{D}_H^\epsilon(\rho)$. The following must hold:
\begin{equation}
     \mathfrak{D}_H^\epsilon(\rho) = \min_{\delta\in \mF}D_H^\epsilon(\rho\|\delta)= D_H^\epsilon(\rho\|\tilde\delta)
    \geq D_H^\epsilon(\mathcal{E}(\rho)\|\mathcal{E}(\tilde\delta))
    \geq \min_{\delta'\in \mF}D_H^\epsilon(\mathcal{E}(\rho)\|\delta') = \mathfrak{D}_H^\epsilon(\mathcal{E}(\rho)),
\end{equation}
where the first inequality follows from the data processing inequality for the hypothesis testing relative entropy \cite{hypothesis}.
Now suppose $\Phi_d$ is a distillable reference state (up to $\epsilon$ error) by free operation ${\mathcal{E}}$, meaning that ${\mathcal{E}}(\rho)\in\mathcal{B}^\epsilon(\Phi_d)$, i.e.~$\Tr\{\Phi_d{\mathcal{E}}(\rho)\}\geq 1-\epsilon$.   So we have 
\begin{align}
    \mathfrak{D}_H^\epsilon({\mathcal{E}}(\rho)) 
    =& \min_{\delta\in \mF}\max_{0\leq P\leq I, \Tr\{P\mathcal{E}(\rho)\}\geq 1-\epsilon}(-\log\Tr\{P\delta\})\\
    \geq& \max_{0\leq P\leq I, \Tr\{P\mathcal{E}(\rho)\}\geq 1-\epsilon}\min_{\delta\in \mF}(-\log\Tr\{P\delta\})\\
    \geq& \min_{\delta\in \mF}(-\log\Tr\{\Phi_d\delta\}) = -\log\mathfrak{f}(\Phi_d),
\end{align}
where the second line follows from the max-min inequality, and the third line follows from the fact that $\Phi_d$ satisfies $0\leq\Phi_d\leq I$ and $\Tr\{\Phi_d{\mathcal{E}}(\rho)\}\geq 1-\epsilon$.
So for any $d>d_0$ the possibility that $\Phi_d$ can be distilled is forbidden by the above inequality, that is, $\Omega_{D,\mathbb{\barx}}^{\epsilon}(\rho\rightarrow\{\Phi_d\}) \leq \log d_0$.  Also by the definition of $d_0$ we have $m_f(\Phi_{d_0})\log d_0 \leq \mathfrak{D}_{H}^\epsilon(\rho)$. So the claimed bound follows.
\end{proof}

We can obtain two general lower bounds by considering two different types of distillation maps.  The first one goes as follows:
\begin{prop}\label{prop:fr}
Given primitive state $\rho$ and reference states $\{\phi_d\}$.
Suppose the resource theory satisfies Condition (FFR), i.e.~all states have finite free robustness.
Let $d_0 = \max\{d\in\mathbb{D}: LR(\phi_d) = m_{LR}(\phi_d)\log d \leq \mathfrak{D}_H^\epsilon(\rho)\}$. 
 For $\epsilon\geq 0$,
 \begin{equation}
     \Omega_{D,\mathbb{\barx}}^{\epsilon}(\rho\rightarrow\{\phi_d\}) \geq \log d_0 > \frac{\mathfrak{D}_{H}^\epsilon(\rho)}{m_{LR}(\phi_{d_0^{\uparrow}})}-\log{\frac{d_0^{\uparrow}}{d_0}}.  \label{eq:opsmlr} 
 \end{equation}
\end{prop}
\begin{proof}
Suppose $\rho\in\mD(\mH_d)$.  
For any dimension $d$, According to the definition of $R(\phi_d)$, there exists a free state $\delta_d\in \mF$ such that 
$\frac{1}{1+R(\phi_d)}\phi_d+\frac{R(\phi_d)}{1+R(\phi_d)}\delta_d\in \mF$.
Let $0\leq P\leq I$ be the optimal operator that achieves
$ \mathfrak{D}_H^\epsilon(\rho)$. Define cptp map $\mathcal{E}:\mL(\mH_d)\rightarrow\mL(\mH'_{d_0})$  on input state $\omega\in\mD(\mH_d)$ of the following form:
\begin{equation}\label{eq:freeO_1}
\mathcal{E}(\omega)
=\Tr\{P\omega\}\phi_{d_0}
+(1-\Tr\{P\omega\})\delta_{d_0}.
\end{equation}

On the one hand, 
it is easy to see that 
$\mathcal{E}(\omega)\in \mF$ if and only if
$\frac{1-\Tr\{P\omega\}}{\Tr\{P\omega\}}\geq R(\phi_{d_0})$, that is, 
\begin{equation}
\Tr\{P\omega\}
\leq \frac{1}{1+R(\phi_{d_0})}=2^{-LR(\phi_{d_0})}. \label{eq:fcond}
\end{equation}
By definition, we have
$ \mathfrak{D}_H^\epsilon(\rho)=\min_{\tau\in \mF}(-\log\Tr\{P\tau\})$, which implies that
$\Tr\{P\tau\}\leq 2^{-\mathfrak{D}_H^\epsilon(\rho)}$ for any free state $\tau\in \mF$.
By the definition of $d_0$, we have $LR(\phi_{d_0})\leq \mathfrak{D}_H^\epsilon(\rho)$, so 
$\Tr\{P\tau\}\leq 2^{-LR(\phi_{d_0})}$, satisfying Eq.~(\ref{eq:fcond}), which indicates that $\mathcal{E}(\tau)\in \mF, \forall\tau\in\mF$.  Therefore, $\mathcal{E}\in\barx$. 

On the other hand, since
$\Tr\{P\rho\}\geq 1-\epsilon$, it can be verified that $f(\mathcal{E}(\rho), \phi_{d_0})\geq 1-\epsilon$.  That is, $\mathcal{E}$ achieves the distillation task with rate $\log{d_0}$.
To obtain a bound in terms of 
$\mathfrak{D}_{H}^\epsilon(\rho)$, notice that $m_{LR}(\Phi_{d_0^{\uparrow}})\log d_0^{\uparrow}>\mathfrak{D}_{H}^\epsilon(\rho)$ by the definition of $d_0$, and so the claimed bound follows.
\end{proof}

The second distillation method makes use of a generalized notion of ``isotropic states'' originally introduced in the context of entanglement theory \cite{PhysRevA.59.4206}:
\begin{prop}\label{prop:depol1}
Given primitive state $\rho$ and pure reference states $\{\Phi_d\}$.  
Suppose $\mF$ is convex and contains the maximally mixed state.
The depolarizing map with noise parameter $p$ is defined to be $\mathcal{N}_p(\sigma) := (1-p)\sigma + pI/d$ (where $\sigma\in\mD(\mH_d)$) \footnote{Note that the noise parameter $p$ here does not necessarily represent a probability: in the following we may consider $p>1$, where the depolarizing map induces a \emph{pseudomixture} of original state and the maximally mixed state, for mathematical purposes}. Let $\mathbb{D}' = \{d\in\mathbb{D}: \mathcal{N}_{\frac{d}{d-1}}(\Phi_d)\in\mF\}$ (if $\mathbb{D}'$ is empty then the result does not apply), and let $d_0 = \max\{d\in\mathbb{D}': D_{\min}(\Phi_d\|\mathcal{N}_{\tilde p_d}(\Phi_d)) = m_{\min,(\mathcal{N}_{\tilde p_d})}\log d \leq \mathfrak{D}_H^\epsilon(\rho)\}$  (the brackets in the subscript of $m$ signify that $\mathcal{N}_{\tilde p}$ is not precisely a resource destroying map), where $\tilde{p}_d = \min\{p: \mathcal{N}_{p}(\Phi_d)\in\mF\}\in[0,1]$.  For $\epsilon\geq 0$,
 \begin{equation}
     \Omega_{D,\mathbb{\barx}}^{\epsilon}(\rho\rightarrow\{\Phi_d\}) \geq \log d_0 > \frac{\mathfrak{D}_{H}^\epsilon(\rho)}{m_{\min,(\mathcal{N}_{\tilde p_{d_0^{\uparrow}}})}(\Phi_{d_0^{\uparrow}})} - \log{\frac{d_0^{\uparrow}}{d_0}}.  \label{eq:opsm2} 
 \end{equation}
\end{prop}
\begin{proof}
Suppose $\rho\in\mD(\mH_d)$.
Let $0\leq P\leq I$ be the optimal operator that achieves
$ \mathfrak{D}_H^\epsilon(\rho)$. Define cptp map $\mathcal{E}:\mL(\mH_d)\rightarrow\mL(\mH'_{d_0})$  on input state $\omega\in\mD(\mH_d)$ of the following form:
\begin{align}
\mathcal{E}(\omega)
&=\Tr\{P\omega\}\Phi_{d_0}
+(1-\Tr\{P\omega\})\frac{I-\Phi_{d_0}}{d_0-1}\label{eq:dist2_1}\\
&=\left(1- \frac{d_0(1- \Tr\{P\omega\})}{d_0-1}\right)\Phi_{d_0} + \frac{d_0(1- \Tr\{P\omega\})}{d_0-1}\frac{I}{d_0} \label{eq:dist2_2}\\
&= (1-q(\omega))\Phi_{d_0} + q(\omega) \frac{I}{d_0},  \label{eq:dist2_3}
\end{align}
where $q(\omega) = \frac{d_0(1- \Tr\{P\omega\})}{d_0-1} \in [0,\frac{d_0}{d_0-1}]$, that is, $\mathcal{E}(\omega)$ is a pseudomixture of $\Phi_{d_0}$ and the maximally mixed state $I/d_0$, and is positive semidefinite.

On the one hand, by noticing that $\Tr\{\Phi_{d_0}(I-\Phi_{d_0})\}=0$ and plugging it into Eq.~(\ref{eq:dist2_1}), we obtain
\begin{equation}
    \Tr\{\Phi_{d_0}\mathcal{E}(\omega)\} = \Tr\{P\omega\}\Tr\Phi_{d_0}
+(1-\Tr\{P\omega\})\frac{\Tr\{\Phi_{d_0}(I-\Phi_{d_0})\}}{d_0-1}= \Tr\{P\omega\}.\label{eq:tre}
\end{equation}
By the definition of $P$ we have
$ \mathfrak{D}_H^\epsilon(\rho)=\min_{\tau\in \mF}(-\log\Tr\{P\tau\})$, which implies that
$\Tr\{P\tau\}\leq 2^{-\mathfrak{D}_H^\epsilon(\rho)}$ for any free state $\tau\in \mF$.
Therefore, by Eq.~(\ref{eq:tre}), $\Tr\{\Phi_{d_0}\mathcal{E}(\tau)\} \leq 2^{-\mathfrak{D}_H^\epsilon(\rho)}, \forall\tau\in \mF$.  Now notice that $\Tr\{\Phi_{d_0}\mathcal{N}_{\tilde p_{d_0}}(\Phi_{d_0})\} \geq 2^{-\mathfrak{D}_H^\epsilon(\rho)} $ where $\mathcal{N}_{\tilde p_{d_0}}(\Phi_{d_0})\in\mF$ by assumption.   That is,
\begin{equation}
    \Tr\{\Phi_{d_0}\mathcal{E}(\tau)\} = \Tr\{\Phi_{d_0}\mathcal{N}_{q(\tau)}(\Phi_{d_0})\} \leq\Tr\{\Phi_{d_0}\mathcal{N}_{\tilde p_{d_0}}(\Phi_{d_0})\}, \label{eq:tr_comp}
\end{equation}
where the equality follows from Eq.~(\ref{eq:dist2_3}) and the definition of $q$.
It can be easily verified that
\begin{equation}
    \Tr\{\Phi_{d_0}\mathcal{N}_{p}(\Phi_{d_0})\} = 1-p+\frac{p}{d_0},
\end{equation}
so by Eq.~(\ref{eq:tr_comp}) we have $q(\tau)\geq \tilde p_{d_0}$.  
Now since $\tilde p_{d_0}\leq q(\tau)\leq\frac{d_0}{d_0-1}, \forall\tau\in\mF$, we have the following convex decomposition of $\mathcal{E}(\tau)$: 
\begin{equation}
    \mathcal{E}(\tau) = \mathcal{N}_{q(\tau)}(\Phi_{d_0}) = \frac{d_0 - (d_0-1)q(\tau)}{d_0-(d_0-1)\tilde p_{d_0}}\mathcal{N}_{\tilde p_{d_0}}(\Phi_{d_0}) +  \frac{(d_0-1)(q(\tau)-\tilde p_{d_0})}{d_0-(d_0-1)\tilde p_{d_0}}\ \mathcal{N}_{\frac{d_0}{d_0-1}}(\Phi_{d_0}),  \label{eq:conv_comb1}
\end{equation}
where $\mathcal{N}_{\tilde p_{d_0}}(\Phi_{d_0}), \mathcal{N}_{\frac{d_0}{d_0-1}}(\Phi_{d_0}) \in \mF$ by assumption. Therefore $\mathcal{E}(\tau)\in\mF, \forall\tau\in\mF$ by the convexity of $\mF$, namely $\mathcal{E} \in \barx$.



On the other hand, since $\Tr\{P\rho\}\geq 1-\epsilon$, we have $f(\mathcal{E}(\rho),\Phi_{d_0})\geq 1-\epsilon$.  That is, $\mathcal{E}$ achieves the distillation task with rate $\log{d_0}$.
To obtain a bound in terms of 
$\mathfrak{D}_{H}^\epsilon(\rho)$, notice that $m_{\min,(\mathcal{N}_{\tilde p_{d_0^{\uparrow}}})}\log d_0^{\uparrow}>\mathfrak{D}_{H}^\epsilon(\rho)$ by the definition of $d_0$, and so the claimed bound follows.
\end{proof}



\section{One-shot distillation yield by commuting operations (Theorem 6 and additional results)}\label{app:yield_comm}

In this section, we present proofs of the bounds on the one-shot distillation yield (the standard version with error tolerance on the target state) under commuting operations with respect to resource destroying map $\lambda$, i.e.~$\mathscr{F}_{\lambda,{\rm Comm}}$.   We first prove an upper bound when $\lambda$ is a resource destroying channel, namely Theorem 6 of the main text:

\begin{prop}
  Given primitive state $\rho$,  resource destroying channel $\Lambda$,  and pure reference states $\{\Phi_d\}$.   Let $d_0 = \max\{d\in\mathbb{D}: \mathfrak{f}_\Lambda(\Phi_d) = d^{-m_{f,\Lambda}(\Phi_d)} \geq 2^{-\mathfrak{D}_{H,\Lambda}^\epsilon(\rho)}-2\sqrt{\epsilon}\}$. 
 For $\epsilon\geq 0$,
  \begin{equation}
   \Omega_{D,\mathscr{F}_{\Lambda,{\rm Comm}}}^{\epsilon}(\rho\rightarrow\{\Phi_d\}) \leq \log d_0 \leq \frac{-\log(2^{-\mathfrak{D}_{H,\Lambda}^\epsilon(\rho)}-2\sqrt{\epsilon})}{m_{f,\Lambda}(\Phi_{d_0})}.  \label{eq:opsm2_comm} 
 \end{equation}
\end{prop}
\begin{proof}
Suppose $\Phi_d$ is a distillable reference state by free operation ${\mathcal{E}}\in\mathscr{F}_{\Lambda,{\rm Comm}}$, meaning that ${\mathcal{E}}(\rho)\in\mathcal{B}^\epsilon(\Phi_d)$, i.e.~$\Tr\{\Phi_d{\mathcal{E}}(\rho)\}\geq 1-\epsilon$.   So 
\begin{align}
    \mathfrak{D}_{H,\Lambda}^\epsilon(\rho) =& D_H^\epsilon(\rho\|\Lambda(\rho)) 
    \\\geq& D_H^\epsilon(\mathcal{E}(\rho)\|\mathcal{E}(\Lambda(\rho))) \\
    =& D_H^\epsilon(\mathcal{E}(\rho)\|\Lambda(\mathcal{E}(\rho)))\\
    =& \max_{0\leq P\leq I, \Tr\{P\mathcal{E}(\rho)\}\geq 1-\epsilon}(-\log\Tr\{P\Lambda(\mathcal{E}(\rho))\}) \\
    \geq& -\log\Tr\{\Phi_d\Lambda(\mathcal{E}(\rho))\},
\end{align}
where the second line follows from the data processing inequality for the hypothesis testing relative entropy \cite{hypothesis}, the third line follows from $\mathcal{E}\in\mathscr{F}_{\Lambda,{\rm Comm}}$, and the fifth line follows from the fact that $\Phi_d$ satisfies $0\leq\Phi_d\leq I$ and $\Tr\{\Phi_d{\mathcal{E}}(\rho)\}\geq 1-\epsilon$.
If $\Lambda$ is a channel, then we have $\Lambda(\mathcal{E}(\rho)) \in \mathcal{B}^\epsilon(\Lambda(\Phi_d))$ due to the data processing inequality for the purified distance $P(\rho,\sigma) = \sqrt{1-f(\rho,\sigma)}$  \cite{marco_thesis}: $P(\Lambda(\rho),\Lambda(\sigma))\leq P(\rho,\sigma)$, so $f(\Lambda(\Phi_d),\Lambda({\mathcal{E}}(\rho)))\geq f(\Phi_d,{\mathcal{E}}(\rho))\geq 1-\epsilon$.  Then we have 
\begin{equation}
    |\Tr\{\Phi_d\left(\Lambda(\mathcal{E}(\rho))-\Lambda(\Phi_d)\right)\}| \leq \|\Lambda(\mathcal{E}(\rho))-\Lambda(\Phi_d)\|_1 \leq 2\sqrt{1-f(\Lambda(\Phi_d),\Lambda({\mathcal{E}}(\rho)))} \leq 2\sqrt{\epsilon},
\end{equation}
where the first inequality is due to the submultiplicativity of the trace norm and $\|\Phi_d\|_1=1$, and the second inequality follows from the Fuchs-van de Graaf inequality (relations between trace norm and fidelity) \cite{Fuchs-vandeGraaf}.
Therefore,
\begin{equation}
  2^{-\mathfrak{D}_{H,\Lambda}^\epsilon(\rho)} \leq \Tr\{\Phi_d\Lambda(\mathcal{E}(\rho))\} \leq \Tr\{\Phi_d\Lambda(\Phi_d)\} + 2\sqrt{\epsilon} = \mathfrak{f}_\Lambda(\Phi_d)+2\sqrt{\epsilon}.
\end{equation}
So for any $d>d_0$ the possibility that $\Phi_d$ can be distilled is forbidden by the above inequality, that is, $\Omega_{D,\mathscr{F}_{\Lambda,{\rm Comm}}}^{\epsilon}(\rho\rightarrow\{\Phi_d\}) \leq \log d_0$.  Also by the definition of $d_0$ we have $d_0^{-m_{f,\Lambda}(\Phi_{d_0})} \geq 2^{-\mathfrak{D}_{H,\lambda}^\epsilon(\rho)}-2\sqrt{\epsilon}$. So the claimed bound follows. 
\end{proof}

To find general lower bounds, we now consider an adjusted notion of commuting operations.  To this end, we define a relaxed version of resource destroying maps which no longer enforces the nonresource-fixing property, and the corresponding commuting operations:
\begin{defn}[Pseudo-resource destroying map, pseudo-commuting operations]
A map $\lambda$ from states to states is a \emph{pseudo-resource destroying map} if it maps all states to free states, i.e.~$\lambda(\rho)\in\mF,\forall\rho$.   Then the set of \emph{pseudo-commuting operations} is defined to be $\mathscr{F}_{\langle\lambda\rangle,{\rm Comm}}:=\{\mathcal{E}|\lambda\circ\mathcal{E} = \mathcal{E}\circ\lambda\}$ (the angle brackets signify that $\lambda$ is a pseudo-resource destroying map).
\end{defn}
The importance of the compromised nonresource-fixing property is mainly two-fold: i) It guarantees that the commuting operations are free (resource non-generating); ii) The simple distance measure induced by the resource destroying map is directly a faithful resource monotone under commuting operations.  Can these issues still be handled?
Regarding the first point, note that $\mathscr{F}_{\langle\lambda\rangle,{\rm Comm}}\subset\barx$ still holds when $\lambda$ is a surjection onto $\mF$ (which is the case for standard resource destroying maps, due to the nonresource-fixing property), i.e.~every free state is an image of $\lambda$, because $\mathcal{E}\circ\lambda(\rho) = \lambda\circ\mathcal{E}(\rho)\in\mF$ and $\lambda(\rho)$ covers $\mF$.  However, this does not necessarily hold when the image of $\lambda$ is a proper subset of $\mF$. So in general we let the set of free operations be $\bar{\mathscr{F}}_{\langle\lambda\rangle,{\rm Comm}} = \mathscr{F}_{\langle\lambda\rangle,{\rm Comm}}\cap\barx$.   
Regarding the second one, the problem is that the distance measure $\delta_{\lambda}(\rho)=\delta(\rho,\lambda(\rho))$ is no longer a faithful resource measure, that is, $\delta_{\lambda}(\sigma)>0$ for $\sigma\in\mF$, which is implausible. To resolve this, we define the following variant by simply setting the value to zero for free states:
\begin{equation}
    \bar\delta_{\lambda}(\rho):=
    \left\{\begin{array}{ll}{\delta(\rho,\lambda(\rho))} & {\text { if } \rho\notin\mF } \\ {0} & {\text { if } \rho\in\mF}\end{array}\right.,
\end{equation}
where $\delta$ is a contractive distance measure.    It can be easily seen that $\bar\delta_{\lambda}$ still serves as a simple resource monotone:
\begin{prop}
Let $\lambda$ be a pseudo-resource destroying map.
Then $\bar\delta_{\lambda}(\rho)$ is monotone non-increasing under pseudo-commuting resource non-generating operations, i.e.
\begin{equation}
    \bar\delta_{\lambda}(\mathcal{E}(\rho)) \leq \bar\delta_{\lambda}(\rho), ~ \forall\mathcal{E}\in{\bar{\mathscr{F}}_{\langle\lambda\rangle,{\rm Comm}}= \mathscr{F}_{\langle\lambda\rangle,{\rm Comm}}\cap\barx}.
\end{equation}
\end{prop}
\begin{proof}
If $\rho\in\mF$, then  $\mathcal{E}(\rho)\in\mF$ since $\mathcal{E}\in\barx$, so $\bar\delta_{\lambda}(\mathcal{E}(\rho)) = \bar\delta_{\lambda}(\rho) =0$ by definition.
Otherwise if $\rho\notin\mF$, then
\begin{equation}
    \bar\delta_{\lambda}(\mathcal{E}(\rho))=
    \left\{\begin{array}{ll}{\delta(\mathcal{E}(\rho),\mathcal{E}(\lambda(\rho)))\leq \delta(\rho,\lambda(\rho)) = \bar\delta_{\lambda}(\rho)} & {\text { if } \mathcal{E}(\rho)\notin\mF } \\ {0\leq \bar\delta_{\lambda}(\rho)} & {\text { if } \mathcal{E}(\rho)\in\mF}\end{array}\right.,
\end{equation}
where we used $\lambda\circ\mathcal{E} = \mathcal{E}\circ\lambda$ and data processing inequality for the first case.
\end{proof}

Now we show that the map defined in Eq.~(\ref{eq:dist2_3}) at least commutes with the depolarizing map if $\Tr P=1$, which leads to the following lower bound for pseudo-commuting operations with respect to depolarizing-type pseudo-resource destroying maps:
\begin{prop}\label{prop:depol_comm}
Given primitive state $\rho$, a pseudo-resource destroying map $\lambda_\mathcal{N}$ given by depolarizing map of any degree, and pure reference states $\{\Phi_d\}$.   
Suppose $\mF$ is convex and contains the maximally mixed state,
and that the optimal operator $P$ that achieves $\mathfrak{D}_H^\epsilon(\rho)$ satisfies $\Tr P = 1$. Let $\mathbb{D}' = \{d\in\mathbb{D}: \mathcal{N}_{\frac{d}{d-1}}(\Phi_d)\in\mF\}$ (if $\mathbb{D}'$ is empty then the result does not apply), and let $d_0 = \max\{d\in\mathbb{D}': D_{\min}(\Phi_d\|\mathcal{N}_{\tilde p_d}(\Phi_d)) = m_{\min,(\mathcal{N}_{\tilde p_d})}\log d \leq \mathfrak{D}_H^\epsilon(\rho)\}$  (the brackets in the subscript of $m$ signify that $\mathcal{N}_{\tilde p}$ is not precisely a resource destroying map), where $\tilde{p}_d = \min\{p: \mathcal{N}_{p}(\Phi_d)\in\mF\}\in[0,1]$.  For $\epsilon\geq 0$,
 \begin{equation}
     \Omega_{D,\bar{\mathscr{F}}_{\langle\lambda_\mathcal{N}\rangle,{\rm Comm}}}^{\epsilon}(\rho\rightarrow\{\phi_d\}) \geq \log d_0 > \frac{\mathfrak{D}_{H}^\epsilon(\rho)}{m_{\min,(\mathcal{N}_{\tilde p_d})}(\Phi_{d_0^{\uparrow}})} - \log{\frac{d_0^{\uparrow}}{d_0}}.  \label{eq:depol_comm} 
 \end{equation}
\end{prop}
\begin{proof}
Consider the distillation map $\mathcal{E}$ defined by Eqs.~(\ref{eq:dist2_1})--(\ref{eq:dist2_3}).  Suppose $\lambda_\mathcal{N}$ is the depolarizing map of degree $p$, i.e.~$\lambda_\mathcal{N}=\mathcal{N}_p$.
Simply compare the actions of $\mathcal{N}_p\circ\mathcal{E}$ and $\mathcal{E}\circ\mathcal{N}_p$ on any state $\omega$.
On the one hand,
\begin{align}
    \mathcal{N}_p(\mathcal{E}(\omega)) &= (1-p)\mathcal{E}(\omega) + p\frac{I}{d_0}\\
    &= (1-p-q(\omega) + pq(\omega))\Phi_{d_0} + (p+q(\omega)-pq(\omega))\frac{I}{d_0},
\end{align}
where $q(\omega) = \frac{d_0(1- \Tr\{P\omega\})}{d_0-1}$.
On the other hand, notice that 
\begin{equation}
    \Tr\{P\mathcal{N}_p(\omega)\} = (1-p)\Tr\{P\omega\} + \frac{p}{d_0}\Tr P,
\end{equation}
and thus we obtain
\begin{align}
   q(\mathcal{N}_p(\omega)) &= \frac{d_0}{d_0-1}(1-\Tr\{P\mathcal{N}_p(\omega)\})\\
    &= \frac{d_0}{d_0-1}\left(1-(1-p)\Tr\{P\omega\} - \frac{p}{d_0}\Tr P\right)  \\
    &= \frac{d_0-\Tr P}{d_0-1}p+q(\omega) - pq(\omega).
\end{align}
Therefore
\begin{align}
    \mathcal{E}(\mathcal{N}_p(\omega)) &= (1-q(\mathcal{N}_p(\omega)))\Phi_{d_0} + q(\mathcal{N}_p(\omega))\frac{I}{d_0} \\
    &= \left(1- \frac{d_0-\Tr P}{d_0-1}p - q(\omega) + pq(\omega)\right)\Phi_{d_0}
    + \left(\frac{d_0-\Tr P}{d_0-1}p+q(\omega) - pq(\omega)\right)\frac{I}{d_0}.
\end{align}
It can be seen that, when  $\Tr P=1$, $\mathcal{E}(\mathcal{N}_p(\omega)) = \mathcal{N}_p(\mathcal{E}(\omega))$, that is, $\mathcal{E}\in\mathscr{F}_{\langle\lambda_\mathcal{N}\rangle,{\rm Comm}}$.  
It is already shown in Proposition~\ref{prop:depol1} that $\mathcal{E}\in\barx$.
Therefore, $\mathcal{E}\in\bar{\mathscr{F}}_{\langle\lambda_\mathcal{N}\rangle,{\rm Comm}}$.
Finally, recall that $\mathcal{E}$ achieves the distillation task with rate $\log{d_0}$, which gives rise to the same bound as in Proposition~\ref{prop:depol1}.
\end{proof}

\section{One-shot distillation yield with error tolerance on the input}\label{app:input}

In this section, we consider a variant of the standard one-shot distillation tasks that accounts for noise or error on the primitive state.   In this case we are able to establish a number of bounds in slightly different forms.

The formal definition of the corresponding optimal rate is defined as follows:
\begin{defn}[One-shot resource distillation, error tolerance on the input]
Given primitive state $\rho$ and reference states $\{\phi_d\}$. The \emph{one-shot $\underline\epsilon$-distillation yield} under the set of free operations $\mathscr{F}$ is defined to be the maximum possible size of output reference state distillable from an $\epsilon$-approximation of $\rho$ by an operation in $\mathscr{F}$:
\begin{equation}
    \Omega_{D,\mathscr{F}}^{\underline{\epsilon}}(\rho,\{\phi_d\}) := \log\max\{d\in\mathbb{D}: \exists \mathcal{E}\in \mathscr{F}, \exists\rho'\in\mathcal{B}^\epsilon(\rho), \mathcal{E}(\rho')=\phi_d\}.
\end{equation}
\end{defn}

\smallskip

First, observe that this distillation task places a stronger error constraint than the standard output error model, in the sense that the $\underline{\epsilon}$-distillation yield is upper-bounded by the $\epsilon$-distillation yield:
\begin{prop}
Given primitive state $\rho$ and reference states $\{\phi_d\}$. For any free operations $\mathscr{F}$ and $\epsilon\geq 0$, it holds that
\begin{equation}
   \Omega_{D,\mathscr{F}}^{\underline\epsilon}(\rho\rightarrow\{\phi_d\}) 
   \leq
   \Omega_{D,\mathscr{F}}^{\epsilon}(\rho\rightarrow\{\phi_d\}).
\end{equation}
\end{prop}
\begin{proof}
Suppose $\mathcal{E}\in\mathscr{F}$ achieves the $\underline\epsilon$-distillation task with rate $\log d$, i.e.~$\mathcal{E}(\rho_\epsilon) = \phi_d$ for some $\rho_\epsilon\in\mathcal{B}^\epsilon(\rho)$.
Then by the data processing inequality for the purified distance $P(\rho,\sigma) = \sqrt{1-f(\rho,\sigma)}$  \cite{marco_thesis},
we have $P(\phi_d,\mathcal{E}(\rho))=P(\mathcal{E}(\rho_\epsilon),\mathcal{E}(\rho))\leq P(\rho_\epsilon,\rho)$ or equivalently $f(\phi_d,\mathcal{E}(\rho))\geq f(\rho_\epsilon,\rho)$, so $\mathcal{E}(\rho) \in \mathcal{B}^\epsilon(\phi_d)$.  In other words, the existence of a protocol that achieves $\underline\epsilon$-distillation with rate $\log d$ guarantees $\epsilon$-distillation with at least the same rate.
\end{proof}

A direct corollary is that all the upper bounds for $\Omega_{D,\mathscr{F}}^{\epsilon}(\rho\rightarrow\{\phi_d\})$ derived in Appendix~\ref{app:yield_comm} are also upper bounds for $\Omega_{D,\mathscr{F}}^{\underline\epsilon}(\rho\rightarrow\{\phi_d\}) $.

In this case, we are also able to prove general upper bounds in terms of both state- and operator-smoothing. 
We note that the state-smoothed min-relative entropy $D_{\min}^\epsilon$ exhibits the stringent property that when the non-smoothed version is zero, it remains zero for sufficiently small $\epsilon$.

For the resource non-generating operations $\barx$ we obtain the following:
\begin{prop}
 Given primitive state $\rho$ and reference states $\{\phi_d\}$. Let $d_0 = \max\{d\in\mathbb{D}: \mathfrak{D}_{\min}(\phi_d) = m_{\min}(\phi_d)\log d \leq \mathfrak{D}_{\min}^\epsilon(\rho)\}, d'_0 = \max\{d\in\mathbb{D}: \mathfrak{D}_{\min}(\phi_d) = m_{\min}(\phi_d)\log d \leq \mathfrak{D}_{H}^\epsilon(\rho)\}$. 
 For $\epsilon\geq 0$,
  \begin{eqnarray}
\Omega_{D,\mathbb{\barx}}^{\underline{\epsilon}}(\rho\rightarrow\{\phi_d\}) &\leq& \log d_0 \leq \frac{\mathfrak{D}_{\mathrm{min}}^\epsilon(\rho)}{m_{\min}(\phi_{d_0})}, \label{eq:statesm} \\
  \Omega_{D,\mathbb{\barx}}^{\underline{\epsilon}}(\rho\rightarrow\{\phi_d\}) &\leq& \log d'_0 \leq \frac{\mathfrak{D}_{H}^\epsilon(\rho)}{m_{\min}(\phi_{d'_0})}. \label{eq:statesm2} 
 \end{eqnarray}
\end{prop}
\begin{proof}
Suppose $\rho_\epsilon$ is the state in $\mathcal{B}^\epsilon(\rho)$ that achieves the distillation with free operation $\mathcal{E}$, i.e.~$\mathcal{E}(\rho_\epsilon) = \phi_d$ for some $\mathcal{E}\in\barx$.  
Let $\tilde\delta\in \mF$ be the optimal free state that achieves $\mathfrak{D}_{\mathrm{min}}(\rho_\epsilon)$.
The following must hold:
\begin{align}
    \mathfrak{D}_{\mathrm{min}}^\epsilon(\rho)\geq& \mathfrak{D}_{\mathrm{min}}(\rho_\epsilon)\\
    =&   D_{\mathrm{min}}(\rho_\epsilon\|\tilde\delta) \\
    \geq& D_{\mathrm{min}}(\phi_d\|\mathcal{E}(\tilde\delta)) \\
    \geq& \min_{\delta'\in \mF} D_{\mathrm{min}}(\phi_d\|\delta')  = \mathfrak{D}_{\min}(\phi_d),
\end{align}
where the third line follows from the data processing inequality for the non-sandwiched R\'enyi-$\alpha$ relative entropy in the regime $\alpha\in[0,2]$ which covers $D_{\mathrm{min}}$ \cite{LIEB1973267,uhlmann1977,PETZ198657}, and the fourth line follows from $\mathcal{E}\in\barx$.   
So for any $d>d_0$ the possibility that $\phi_d$ can be distilled is forbidden by the above inequality, that is, $\Omega_{D,\barx}^{\underline\epsilon}(\rho\rightarrow\{\phi_d\}) \leq \log d_0$.  Also by the definition of $d_0$ we have $m_{\min}(\Phi_{d_0})\log d_0 \leq \mathfrak{D}_{\mathrm{min}}^\epsilon(\rho)$. So the claimed bound follows.  The same argument applies to operator-smoothing. 
\end{proof}
For commuting operations $\mathscr{F}_{\lambda,{\rm Comm}}$, we first prove the most general upper bounds without any restriction on $\lambda$, and then simplify the result in the case that $\lambda$ is exact:
\begin{prop}
Given primitive state $\rho$, any resource destroying map $\lambda$, and reference states $\{\phi_d\}$.   Let $d_0 = \max\{d\in\mathbb{D}: \mathfrak{D}_{\min,\lambda}(\phi_d) = m_{\min,\lambda}(\phi_d)\log d \leq \mathfrak{D}_{\min,\lambda}^\epsilon(\rho)\}, d'_0 = \max\{d\in\mathbb{D}: \mathfrak{D}_{\min,\lambda}(\phi_d) = m_{\min,\lambda}(\phi_d)\log d \leq \mathfrak{D}_{H,\lambda}^\epsilon(\rho)\}$. 
 For $\epsilon\geq 0$,
   \begin{eqnarray}
   \Omega_{D,\mathscr{F}_{\lambda,{\rm Comm}}}^{\underline\epsilon}(\rho\rightarrow\{\phi_d\}) &\leq& \log d_0 \leq \frac{\mathfrak{D}_{\min,\lambda}^\epsilon(\rho)}{m_{\min,\lambda}(\phi_{d_0})}.  \label{eq:opsmx} \\
   \Omega_{D,\mathscr{F}_{\lambda,{\rm Comm}}}^{\underline\epsilon}(\rho\rightarrow\{\phi_d\}) &\leq& \log d'_0 \leq \frac{\mathfrak{D}_{H,\lambda}^\epsilon(\rho)}{m_{\min,\lambda}(\phi_{d'_0})}.  \label{eq:opsmx2} 
 \end{eqnarray}
 
{Now consider golden states $\{\hat{\Phi}_d\}$ as the reference states.  Suppose the set of free states $\mF$ satisfies Condition (CH), i.e.~is formed by a convex hull of pure states, and $\tilde\lambda$ is an exact resource destroying map.  Let $d''_0 = \min\{d\in\mathbb{D}:\mathfrak{D}_{\min,\tilde\lambda}
(\hat{\Phi}_d) = \mathfrak{D}_{\min}
(\hat{\Phi}_d) = m_{\min}(\hat{\Phi}_d)\log d =g_d\log d\leq \mathfrak{D}_{\min,\tilde{\lambda}}^\epsilon(\rho)\}, 
d'''_0 = \min\{d\in\mathbb{D}:\mathfrak{D}_{\min,\tilde\lambda}
(\hat{\Phi}_d) = \mathfrak{D}_{\min}
(\hat{\Phi}_d) = m_{\min}(\hat{\Phi}_d)\log d =g_d\log d \leq \mathfrak{D}_{H,\tilde{\lambda}}^\epsilon(\rho)\}$. For $\epsilon\geq 0$,
\begin{eqnarray}
    \Omega_{D,\mathscr{F}_{\tilde\lambda,{\rm Comm}}}^{\underline\epsilon}(\rho\rightarrow\{\hat\Phi_d\}) \leq \log d''_0\leq\frac{\mathfrak{D}_{\min,\tilde{\lambda}}^\epsilon(\rho)}{m_{\min}(\hat{\Phi}_{d''_0})} = \frac{\mathfrak{D}_{\min,\tilde{\lambda}}^\epsilon(\rho)}{g_{d''_0}}, \\
   \Omega_{D,\mathscr{F}_{\tilde\lambda,{\rm Comm}}}^{\underline\epsilon}(\rho\rightarrow\{\hat\Phi_d\}) \leq \log d'''_0\leq\frac{\mathfrak{D}_{H,\tilde{\lambda}}^\epsilon(\rho)}{m_{\min}(\hat{\Phi}_{d'''_0})} = \frac{\mathfrak{D}_{H,\tilde{\lambda}}^\epsilon(\rho)}{g_{d'''_0}}.
\end{eqnarray}
}
\end{prop}
\begin{proof}
Suppose $\rho_\epsilon$ is the state in $\mathcal{B}^\epsilon(\rho)$ that achieves the distillation with free operation $\mathcal{E}$, i.e.~$\mathcal{E}(\rho_\epsilon) = \Phi_d$ for some $\mathcal{E}\in\barx$. 
The following must hold:
\begin{align}
    \mathfrak{D}_{\mathrm{min},\lambda}^\epsilon(\rho) 
   \geq& D_{\mathrm{min}}(\rho_\epsilon\|\lambda(\rho_\epsilon)) \\
    \geq& D_{\mathrm{min}}(\mathcal{E}(\rho_\epsilon)\|\mathcal{E}(\lambda(\rho_\epsilon))) \\ =& D_{\mathrm{min}}(\mathcal{E}(\rho_\epsilon)\|\lambda(\mathcal{E}(\rho_\epsilon))) \\
    =& D_{\mathrm{min}}(\phi_d\|\lambda(\phi_d))  =  \mathfrak{D}_{\mathrm{min},\lambda}(\phi_d),
\end{align}
where the second line follows from the data processing inequality for the non-sandwiched R\'enyi-$\alpha$ relative entropy in the regime $\alpha\in[0,2]$ which covers $D_{\mathrm{min}}$ \cite{LIEB1973267,uhlmann1977,PETZ198657}, and the third line follows from $\mathcal{E}\in\mathscr{F}_{\lambda,{\rm Comm}}$.
So for any $d>d_0$ the possibility that $\phi_d$ can be distilled is forbidden by the above inequality, that is, $\Omega_{D,\mathscr{F}_{\lambda,{\rm Comm}}}^{\underline\epsilon}(\rho\rightarrow\{\phi_d\}) \leq \log d_0$.  Also by the definition of $d_0$ we have $m_{\min,\lambda}(\phi_{d_0})\log d_0 \leq \mathfrak{D}_{\mathrm{min},\lambda}^\epsilon(\rho)$. So the claimed bound follows.   The same argument applies to operator-smoothing. 

The last pair of bounds are a direct consequence of the collapse $m_{\min,\tilde\lambda} = m_{\min}$ handled by Proposition~\ref{lemma:exact_min}.
\end{proof}

Next, we modify the two distillation methods used in Appendix~\ref{app:yield_comm} to obtain lower bounds for the input error model.  The first one goes as follows:
\begin{prop}
Given \emph{pure} primitive state $\Psi$ and reference states $\{\phi_d\}$.
Suppose the resource theory satisfies Condition (FFR), i.e.~all states have finite free robustness.
Let $d_0 = \max\{d\in\mathbb{D}: 2^{-LR(\phi_d)} = d^{-m_{LR}(\phi_d)} \geq 2^{-\mathfrak{D}_{\min}(\Psi)} + 2\sqrt{\epsilon}\}$. 
 For $\epsilon\geq 0$,
 \begin{equation}
     \Omega_{D,\mathbb{\barx}}^{\underline\epsilon}(\Psi\rightarrow\{\phi_d\}) \geq \log d_0 > \frac{-\log\left(2^{-\mathfrak{D}_{\min}(\Psi)}+2\sqrt{\epsilon}\right)}{m_{LR}(\Phi_{d_0^{\uparrow}})}-\log{\frac{d_0^{\uparrow}}{d_0}}.  \label{eq:opsmlr2} 
 \end{equation}
\end{prop}
\begin{proof}
Suppose $\Psi\in\mD(\mH_d)$.    
For any dimension $d$, According to the definition of $R(\phi_d)$, there exists a free state $\delta_d\in \mF$ such that 
$\frac{1}{1+R(\phi_d)}\phi_d+\frac{R(\phi_d)}{1+R(\phi_d)}\delta_d\in \mF$.
Now let $\Psi_\epsilon\in\mathcal{B}^\epsilon(\Psi)$ be a pure state that satisfies the input error constraint (which always exists for any $\epsilon$ since $\Psi$ itself is a pure state). Define cptp map $\mathcal{E}:\mL(\mH_d)\rightarrow\mL(\mH'_{d_0})$  on input state $\omega$ of the following form:
\begin{equation}\label{eq:freeO_2}
\mathcal{E}(\omega)
=\Tr\{\Psi_\epsilon\omega\}\phi_{d_0}
+(1-\Tr\{\Psi_\epsilon\omega\})\delta_{d_0}.
\end{equation}

On the one hand, as in the proof of Proposition~\ref{prop:fr} one can verify that $\mathcal{E}(\omega)\in \mF$ if and only if
$\frac{1-\Tr\{\Psi_\epsilon\omega\}}{\Tr\{\Psi_\epsilon\omega\}}\geq R(\phi_{d_0})$, that is, 
\begin{equation}
LR(\phi_{d_0}) \leq -\log\Tr\{\Psi_\epsilon\omega\}. \label{eq:fcond2}
\end{equation}
Notice the continuity property
\begin{equation}
    |2^{-\dmin(\Psi_\epsilon\|\omega)} - 2^{-\dmin(\Psi\|\omega)}| = |\Tr(\Psi_\epsilon-\Psi)\omega| \leq \|(\Psi_\epsilon-\Psi)\omega\|_1 \leq
    \|\Psi_\epsilon-\Psi\|_1 \leq
    2\sqrt{1-f(\Psi_\epsilon,\Psi)}\leq 2\sqrt{\epsilon},   \label{eq:continuity}
\end{equation}
where the second inequality is due to the submultiplicativity of the trace norm and $\|\omega\|_1=1$, and the third inequality follows from the Fuchs-van de Graaf inequality  \cite{Fuchs-vandeGraaf}.  Then, for any free state $\tau\in\mF$, we have
\begin{equation}
    2^{-LR(\phi_{d_0})} \geq 2^{-\mathfrak{D}_{\min}(\Psi)} + 2\sqrt{\epsilon} \geq 2^{-\dmin(\Psi\|\tau)}+ 2\sqrt{\epsilon} \geq 2^{-\dmin(\Psi_\epsilon\|\tau)} = 2^{\log\Tr\{\Psi_\epsilon\tau\}}, 
\end{equation}
where the first inequality follows from the definition of $d_0$, and the third inequality follows from Eq.~(\ref{eq:continuity}).  So Eq.~(\ref{eq:fcond2}) is satisfied, which indicates that $\mathcal{E}(\tau)\in \mF, \forall\tau\in\mF$.   Therefore, $\mathcal{E}\in\barx$.

On the other hand, it can be directly seen that $\mathcal{E}(\Psi_\epsilon) = \phi_{d_0}$, so that $\mathcal{E}$ achieves the distillation task with rate $\log{d_0}$.
To obtain a bound in terms of 
$\mathfrak{D}_{\min}(\Psi)$, notice that ${d_0^{\uparrow}}^{-m_{LR}(\Phi_{d_0^{\uparrow}})}<2^{-\mathfrak{D}_{\min}(\Psi)}+2\sqrt{\epsilon}$ by the definition of $d_0$. That is, $m_{LR}(\Phi_{d_0^{\uparrow}})\log d_0^{\uparrow}>-\log\left(2^{-\mathfrak{D}_{\min}(\Psi)}+2\sqrt{\epsilon}\right)$, from which the claimed bound follows.
\end{proof}

The second one based on isotropic states goes as follows:

\begin{prop}
Given \emph{pure} primitive state $\Psi$ and pure reference states $\{\Phi_d\}$.
Suppose $\mF$ is convex and contains the maximally state.
 Let $\mathbb{D}' = \{d\in\mathbb{D}: \mathcal{N}_{\frac{d}{d-1}}(\Phi_d)\in\mF\}$ (if $\mathbb{D}'$ is empty then the result does not apply), and let $d_0 = \max\{d\in\mathbb{D}': 2^{-D_{\min}(\Phi_d\|\mathcal{N}_{\tilde p_d}(\Phi_d))} = d^{-m_{\min,(\mathcal{N}_{\tilde p_d})}} \geq 2^{-\mathfrak{D}_{\min}(\Psi)}+2\sqrt{\epsilon}\}$  (the brackets in the subscript of $m$ signify that $\mathcal{N}_{\tilde p}$ is not precisely a resource destroying map), where $\tilde{p}_d = \min\{p: \mathcal{N}_{p}(\Phi_d)\in\mF\}\in[0,1]$.  For $\epsilon\geq 0$,
 \begin{equation}
     \Omega_{D,\mathbb{\barx}}^{\underline\epsilon}(\Psi\rightarrow\{\Phi_d\}) \geq \log d_0 > \frac{-\log\left(2^{-\mathfrak{D}_{\min}(\Psi)}+2\sqrt{\epsilon}\right)}{m_{\min,(\mathcal{N}_{\tilde p_{d_0^{\uparrow}}})}(\Phi_{d_0^{\uparrow}})} - \log{\frac{d_0^{\uparrow}}{d_0}}.  \label{eq:opsm2_2} 
 \end{equation}
\end{prop}
\begin{proof}
Suppose $\Psi\in\mD(\mH_d)$.
Let $\Psi_\epsilon\in\mathcal{B}^\epsilon(\Psi)$ be a pure state that satisfies the input error constraint (which always exists for any $\epsilon$ since $\Psi$ itself is a pure state). Define cptp map $\mathcal{E}:\mL(\mH_d)\rightarrow\mL(\mH'_{d_0})$  on input state $\omega\in\mD(\mH_d)$ of the following form (by simply substituting operator $P$ by $\Psi_\epsilon$ in the distillation map for Proposition~\ref{prop:depol1}):
\begin{align}
\mathcal{E}(\omega)
&=\Tr\{\Psi_\epsilon\omega\}\Phi_{d_0}
+(1-\Tr\{\Psi_\epsilon\omega\})\frac{I-\Phi_{d_0}}{d_0-1}\label{eq:dist3_1}\\
&=\left(1- \frac{d_0(1- \Tr\{\Psi_\epsilon\omega\})}{d_0-1}\right)\Phi_{d_0} + \frac{d_0(1- \Tr\{\Psi_\epsilon\omega\})}{d_0-1}\frac{I}{d_0} \label{eq:dist3_2}\\
&= (1-q(\omega))\Phi_{d_0} + q(\omega) \frac{I}{d_0},  \label{eq:dist3_3}
\end{align}
where $q(\omega) = \frac{d_0(1- \Tr\{\Psi_\epsilon\omega\})}{d_0-1} \in [0,\frac{d_0}{d_0-1}]$, that is, $\mathcal{E}(\omega)$ is a pseudomixture of $\Phi_{d_0}$ and the maximally mixed state $I/d_0$, and is positive semidefinite.

On the one hand, again by noticing that $\Tr\{\Phi_{d_0}(I-\Phi_{d_0})\}=0$ and plugging it into Eq.~(\ref{eq:dist3_1}), we obtain
\begin{equation}
    \Tr\{\Phi_{d_0}\mathcal{E}(\omega)\} =
    \Tr\{\Psi_\epsilon\omega\}\Tr\Phi_{d_0}
+(1-\Tr\{\Psi_\epsilon\omega\})\frac{\Tr\{\Phi_{d_0}(I-\Phi_{d_0})\}}{d_0-1}=
\Tr\{\Psi_\epsilon\omega\}.\label{eq:tre2}
\end{equation}
By the same continuity property Eq.~(\ref{eq:continuity}), for any free state $\tau\in\mF$ we have 
\begin{equation}
    2^{-\dmin(\Psi\|\tau)} \geq  2^{-\dmin(\Psi_\epsilon\|\tau)}-2\sqrt{\epsilon}  = \Tr\{\Psi_\epsilon\tau\} -2\sqrt{\epsilon}.
\end{equation}
Therefore, by Eq.~(\ref{eq:tre2}), $\Tr\{\Phi_{d_0}\mathcal{E}(\tau)\} \leq 2^{-\mathfrak{D}_{\min}(\Psi)}+2\sqrt{\epsilon}, \forall\tau\in \mF$.  Notice that $\Tr\{\Phi_{d_0}\mathcal{N}_{\tilde p_{d_0}}(\Phi_{d_0})\} \geq 2^{-\mathfrak{D}_{\min}(\Psi)}+2\sqrt{\epsilon} $ where $\mathcal{N}_{\tilde p_{d_0}}(\Phi_{d_0})\in\mF$ by assumption, so we have
\begin{equation}
    \Tr\{\Phi_{d_0}\mathcal{E}(\tau)\} = \Tr\{\Phi_{d_0}\mathcal{N}_{q(\tau)}(\Phi_{d_0})\} \leq\Tr\{\Phi_{d_0}\mathcal{N}_{\tilde p_{d_0}}(\Phi_{d_0})\}, \label{eq:tr_comp2}
\end{equation}
where the equality follows from Eq.~(\ref{eq:dist3_3}) and the definition of $q$.
The rest of the proof goes similarly as Proposition~\ref{prop:depol1}.
 Eq.~(\ref{eq:tr_comp}) implies that $\tilde p_{d_0}\leq q(\tau)\leq\frac{d_0}{d_0-1}, \forall\tau\in\mF$, and therefore we have the following convex decomposition of $\mathcal{E}(\tau)$: 
\begin{equation}
    \mathcal{E}(\tau) = \mathcal{N}_{q(\tau)}(\Phi_{d_0}) = \frac{d_0 - (d_0-1)q(\tau)}{d_0-(d_0-1)\tilde p_{d_0}}\mathcal{N}_{\tilde p_{d_0}}(\Phi_{d_0}) +  \frac{(d_0-1)(q(\tau)-\tilde p_{d_0})}{d_0-(d_0-1)\tilde p_{d_0}}\ \mathcal{N}_{\frac{d_0}{d_0-1}}(\Phi_{d_0}),  \label{eq:conv_comb_2}
\end{equation}
where $\mathcal{N}_{\tilde p_{d_0}}(\Phi_{d_0}), \mathcal{N}_{\frac{d_0}{d_0-1}}(\Phi_{d_0}) \in \mF$ by assumption. Therefore $\mathcal{E}(\tau)\in\mF, \forall\tau\in\mF$ by the convexity of $\mF$, namely $\mathcal{E} \in \barx$.



It can be directly seen that $\mathcal{E}(\Psi_\epsilon) = \Phi_{d_0}$, so that $\mathcal{E}$ achieves the distillation task with rate $\log{d_0}$.
To obtain a bound in terms of 
$\mathfrak{D}_{\min}(\Psi)$, notice that
${d_0^{\uparrow}}^{-m_{\min,(\mathcal{N}_{\tilde p_{d_0^{\uparrow}}})}}<2^{-\mathfrak{D}_{\min}(\Psi)}+2\sqrt{\epsilon}$ by the definition of $d_0$. That is, $m_{\min,(\mathcal{N}_{\tilde p_{d_0^{\uparrow}}})}\log d_0^{\uparrow}>-\log\left(2^{-\mathfrak{D}_{\min}(\Psi)}+2\sqrt{\epsilon}\right)$, from which the claimed bound follows.
\end{proof}

Following the same argument as in the proof of Proposition~\ref{prop:depol_comm}, we see that the distillation map defined by Eqs.~(\ref{eq:dist3_1})--(\ref{eq:dist3_3}) commutes with the depolarizing map, which leads to the same bound for commuting operations.   Here, notice that $\Psi_\epsilon$ is a pure state so we no longer need to further impose the unit trace assumption.
\begin{prop}
Given pure primitive state $\Psi$, a pseudo-resource destroying map $\lambda_\mathcal{N}$ given by depolarizing map of any degree, and pure reference states $\{\Phi_d\}$.   
Suppose $\mF$ is convex and contains the maximally mixed state.
Let $\mathbb{D}' = \{d\in\mathbb{D}: \mathcal{N}_{\frac{d}{d-1}}(\Phi_d)\in\mF\}$, and let $d_0 = \max\{d\in\mathbb{D}': 2^{-D_{\min}(\Phi_d\|\mathcal{N}_{\tilde p_d}(\Phi_d))} = d^{-m_{\min,(\mathcal{N}_{\tilde p_d})}} \geq 2^{-\mathfrak{D}_{\min}(\Psi)}+2\sqrt{\epsilon}\}$  (the brackets in the subscript of $m$ signify that $\mathcal{N}_{\tilde p}$ is not precisely a resource destroying map), where $\tilde{p}_d = \min\{p: \mathcal{N}_{p}(\Phi_d)\in\mF\}\in[0,1]$.  For $\epsilon\geq 0$,
 \begin{equation}
     \Omega_{D,\bar{\mathscr{F}}_{\langle\lambda_\mathcal{N}\rangle,{\rm Comm}}}^{\underline\epsilon}(\Psi\rightarrow\{\Phi_d\}) \geq
     \frac{-\log\left(2^{-\mathfrak{D}_{\min}(\Psi)}+2\sqrt{\epsilon}\right)}{m_{\min,(\mathcal{N}_{\tilde p_{d_0^{\uparrow}}})}(\Phi_{d_0^{\uparrow}})} - \log{\frac{d_0^{\uparrow}}{d_0}}.\label{eq:depol_comm2} 
 \end{equation}
\end{prop}


\end{document}